\documentclass{article}
\usepackage{ijcai22}

\usepackage{times}
\renewcommand*\ttdefault{txtt}
\usepackage{soul}
\usepackage{url}
\usepackage[hidelinks]{hyperref}
\usepackage[utf8]{inputenc}
\usepackage[small]{caption}
\usepackage{graphicx}
\usepackage{amsmath,amsthm,amssymb}
\usepackage{booktabs}
\usepackage{xspace}
\usepackage{subcaption}
\usepackage{comment}
\usepackage{paralist}
\usepackage{enumitem} 

\usepackage[boxed]{algorithm} 
\usepackage{algorithmic} 

\usepackage{cleveref}
 \RequirePackage{tikz}
\usetikzlibrary{calc, shapes}

\usepackage{thm-restate} 

 \newtheorem{theorem}{Theorem}[section]

 \newtheorem{lemma}[theorem]{Lemma}

 \theoremstyle{definition}
 
 \newtheorem{example}[theorem]{Example}

 \RequirePackage[round]{natbib}

\usepackage{pgfplots}
\pgfplotsset{compat=1.10}
\usepgfplotslibrary{statistics}

\tikzset{ 
protovertex/.style={
  draw,
  fill = blue,
  circle,
  inner sep=0,
  minimum size=.3cm}
}

\tikzset{ 
protodiam/.style={
  draw,
  fill = red,
  diamond,
  inner sep=0,
  minimum size=.3cm}
}

\tikzset{ 
protosq/.style={
  draw,
  fill = green,
  inner sep=0,
  minimum size=.3cm}
}

\graphicspath{{figures/}}

 \RequirePackage{todonotes} 
  \presetkeys{todonotes}{inline}{}

\title{Network Creation with Homophilic Agents}

\author{
Martin Bullinger$^1$\footnote{Contact Author}\and
Pascal Lenzner$^2$\And
Anna Melnichenko$^2$\\
\affiliations
$^1$ Technical University of Munich\\
$^2$ Hasso Plattner Institute, University of Potsdam\\
\emails
bullinge@in.tum.de,
\{pascal.lenzner, anna.melnichenko\}@hpi.de}

\newcommand{\dist}[2]{\mathrm{d}_{#1}({#2})}
\newcommand{\diam}{\mathrm{diam}} 
\renewcommand{\deg}[2]{\mathrm{deg}_{#1}({#2})} 
\newcommand{\nbh}[2]{N_{#1}({#2})} 
\newcommand{\sta}{\mathbf{S}} 
\newcommand{\dsta}{\mathbf{DS}} 
\newcommand{\dss}{\mathbf{DSX}} 
\newcommand{\com}{\mathbf{K}} 

\newcommand{\numb}[1]{n_{#1}} 
\newcommand{\ag}[1]{V_{#1}} 
\newcommand{\SP}{\textbf{s}\xspace} 
\DeclareMathOperator{\SC}{SC} 
\newcommand{\Ncost}[2]{a_{#1}({#2})}	
\newcommand{\Dcost}[2]{\dist{#1}{#2}}	
\newcommand{\Tcost}[2]{c_{#1}({#2})}	
\newcommand{\Tot}{C}	
\newcommand{\frs}[2]{\mathrm{F}_{#1}({#2})} 
\newcommand{\ens}[2]{\mathrm{E}_{#1}({#2})} 
\newcommand{\fr}[2]{\mathrm{f}_{#1}({#2})} 
\newcommand{\en}[2]{\mathrm{e}_{#1}({#2})} 
\newcommand{\lsm}{\mathit{LS}} 
\newcommand{\gsm}{\mathit{GS}} 
\newcommand{\icfgame}{ICF-NCG\xspace} 
\newcommand{\deigame}{DEI-NCG\xspace} 
\newcommand{\thr}{\tau} 

\begin{document}

\maketitle

\begin{abstract}
Network Creation Games are an important framework for understanding the formation of real-world networks. 
These games usually assume a set of indistinguishable agents strategically buying edges at a uniform price leading to a network among them.
However, in real life, agents are heterogeneous and their relationships often display a bias towards similar agents, say of the same ethnic group. 
This homophilic behavior on the agent level can then lead to the emergent global phenomenon of social segregation.       
We initiate the study of Network Creation Games with multiple types of homophilic agents and non-uniform edge cost. Specifically, we introduce and compare two models, focusing on the perception of same-type and different-type neighboring agents, respectively. 
Despite their different initial conditions, both our theoretical and experimental analysis show that the resulting stable networks are almost identical in the two models, indicating a robust structure of social networks under homophily. Moreover, we investigate the segregation strength of the formed networks and thereby offer new insights on understanding segregation.
\end{abstract}
	
\section{Introduction}

Networks play an eminent role in today's world. 
They are crucial for our energy supply (power grid networks), our information exchange (the Internet and the World Wide Web), and our social relationships (friendship networks, email exchange, or co-author networks). There exists an abundance of approaches to provide formal frameworks for modeling networks, see, for example, the books by \citet{jackson_book}, \citet{newman_book}, and \citet{barabasi_book}. In many of these models, the nodes of the network correspond to agents that strategically create connections which is particularly suitable for our main focus of modeling social networks.    

One such stream of research considers variants of the \emph{Network Creation Game (NCG)} as proposed by \citet{FLM+03a}. There, selfish agents create edges to form a network among themselves. 
On the one hand, forming edges is costly and hence agents try to create only the most useful edges. On the other hand, the force that causes agents to form edges at all is well-connectivity within the network which is captured by a desire to occupy a central position.

The NCG is a stylized model of social interaction, providing valuable insight to agents' decision processes when interacting with each other. However, it is important to refine the basic model to spotlight specific details of this decision making. In this sense, we study network creation under the additional assumption that agents are separated into various \emph{types} that model ethnic groups or affiliations.

Our goal is to contribute a new perspective on the simple causes that lead to the segregation of a society, similar to Schelling's checkerboard model for residential segregation~\citep{Sche69a, Sche71a}. Therefore, our agents' cost functions have a bias towards the creation of relationships with agents of the same type. Specifically, we study two models based on two seemingly orthogonal treatments of other agents. In the first model, agents incur a fixed cost for every created edge and a variable cost that only depends on the number of edges towards same-type agents. In the second model, edges towards different-type agents are initially more expensive but their cost drops with an inverse linear decay. 
Both models give a different point-of-view on the same underlying principle, namely homophily of agents, i.e., the tendency to form connections with like-minded people. This is often summarized with the proverb ``birds of a feather flock together'', a dominant intrinsic force repeatedly observed in the creation of social networks, see~\citet{MSC01} for a survey on the extensive sociological research on homophily in social networks. 
While our first model expresses homophily explicitly by an increasing comfort among friends, the second model incorporates homophily indirectly by accounting for a decreasing effort of integration once first contact is established. The latter paradigm is closely related to the well-known effect in social sciences called the ``contact hypothesis'' which states that stereotypes and prejudices between ethnic groups can be weakened by intensified contact~\citep{Allport54, Amir69, DGK03}.

We measure the desirability of networks by means of stability. Since we consider social networks, we assume a bilateral model where two agents have to cooperate to  connect. Consequently, we use pairwise stability~\citep{JaWo96a} as solution concept, rather than, for instance, Nash stability which is more appropriate for unilateral models.

Interestingly, we find an almost identical structure of stable networks for both models. This reveals their close relationship, hinting at a robust structure of networks created under homophily incentives.
Naturally, a very small edge cost causes extremely high connectivity, treating agents as being indistinguishable. For moderately small edge cost, we provide characterizations of stable networks which are all highly segregated. We interpret this as identifying a sweet spot of high sensitivity towards agent types. For larger edge cost, stability causes in theory a large spectrum of networks to form with respect to segregation strength. 
We accompany this theoretical limitation with an average-case analysis by detailed simulations of a simple distributed dynamics, where agents perform improvements towards stable networks. It would be plausible if a generally high edge cost causes less distinction of agent types. While this is sometimes confirmed, we also identify contrasting tendencies towards extreme segregation. An important driver for the different behavior is the initial segregation level, indicating that segregation can be avoided by a high initial effort without constant further interaction.

\section{Related Work}

In the original NCG~\citep{FLM+03a} the cost of every edge is $\alpha >0$, where $\alpha$ is a parameter of the game that allows to adjust the tradeoff between the agents' cost for creating edges and their cost resulting from their centrality in the network, e.g., the sum of distances to all other nodes. Stable networks always exist, in particular for $\alpha < 1$ only cliques are stable whereas for $1\leq \alpha < n$ stars, other trees, and also non-tree networks can be stable~\citep{MMM15}. For $\alpha \geq n$ it is conjectured that all stable networks are trees and a recent line of works has proven this to hold for $\alpha > 3n-3$~\citep{AM17,BL20,DV21}. 
Bilateral NCGs with uniform edge price have been introduced by~\cite{CoPa05a}. 

Also NCG variants with non-uniform edge cost have been studied: a version where edges of differing quality can be bought~\citep{CMH14}, and NCGs where the edge cost depends on the node degrees~\citep{CLMM17}, on the length of the edges in a geometric setting~\citep{BFLM19}, or on the hop-distance of the endpoints~\citep{BFLLM21}. Especially the latter is also motivated by social networks and also bilateral edge formation with pairwise stability as solution concept is considered.  
The NCG variant by \citet{MMO14} that focuses on the creation of communication networks features different types of agents and different but fixed edge costs for each agent type. 

Closest to our work is the model proposed by \citet{MZ17} that is a variant of the connections model~\citep{JaWo96a} with different types of agents. Similar to our model, the cost for maintaining an inter-type connection depends on the homogeneity of the neighborhoods of the involved agents. In contrast to us, the cost for intra-type edges is fixed and the distance cost is defined differently. The authors study the existence and structure of equilibria but do not focus on investigating the segregation strength. The latter has been done by~\citet{HPZ11} using a stochastic process that starts with a randomly drawn network with nodes of different types and then edges are randomly rewired with a built-in bias towards favoring intra-type edges. As main result, the authors show that the network strongly segregates over time, even if the built-in bias is very low.  

Residential segregation has recently received a lot of attention by a stream of research developing a game-theoretic framework based on Schelling's checkerboard model~\citep{CLM18,AEGISV21,E+19,BBLM20,KKV21,BSV21,KKV21a,BBLM22}. 
There, agents of several types strategically select positions on a given \emph{fixed} network and try to optimize the number of same-type agents in their neighborhood. Also hedonic diversity games~\citep{BEI19,BE20,Darmann21} are similar.

\section{Preliminaries and Model}

We consider a set $\ag{} = \{1,\dots, \numb{}\}$ of $n$ agents partitioned into $k\ge 2$ disjoint \emph{types}. The set of types is denoted by $\mathcal T$, and for every type $T\in \mathcal T$, let $\ag{T}$ be the set of agents of type~$T$, i.e., $\ag{} = \bigcup_{T\in \mathcal T}\ag{T}$ and $\ag{T}\cap \ag{T'} = \emptyset$ for $T, T'\in \mathcal T$, with $T\neq T'$. For an agent $u\in \ag{}$, we denote by $\mathcal T(u)$ her type, i.e., $u\in \ag{\mathcal T(u)}$. Given a type $T\in \mathcal T$, then $\numb{T} = |\ag{T}|$ denotes the number of agents of type $T$. We identify types with colors and we assume that there are specific types $B$ and $R$ of \emph{blue} and \emph{red} agents, respectively, which are associated with an agent type having the smallest and largest number of agents, respectively. Thus, for every type $T\in \mathcal T$, we have $\numb{B}\le \numb{T}\le \numb{R}$. In particular, with exactly two agent types we have precisely a blue minority and a red majority type. 

In a network creation game, agents will buy edges to eventually form a network, which is an undirected graph $G = (\ag{}, E)$. Therefore, it is useful to introduce some common concepts and notation from graph theory. Consider an undirected graph $G = (\ag{}, E)$ together with vertices $u,v\in \ag{}$. We denote the (potential) edge between $u$ and $v$ by $uv$ (whether it is present or not). For two agents $u, v\in \ag{}$, the edge $uv$ is called \emph{monochromatic} if $u$ and $v$ are of the same type, and \emph{bichromatic}, otherwise. If $uv\in E$, we use the notation $G -uv = (\ag{}, E\setminus \{uv\})$, otherwise we use $G + uv = (\ag{}, E\cup \{uv\})$. Hence, $G - uv$ and $G + uv$ are the graphs obtained from $G$ by deleting or adding the edge $uv$, respectively.
Further, let $\nbh{G}{u} = \{v\in V\colon uv \in E\}$ denote the \emph{neighborhood} of $u$ in $G$, let $\deg{G}{u} = |\nbh{G}{u}|$ be the \emph{degree} of $u$ in $G$, i.e., the size of its neighborhood, and let $\dist{G}{u,v}$ be the \emph{distance} from $u$ to $v$ in $G$, i.e., the length of a shortest path from $u$ to $v$ in $G$. The \emph{diameter} of $G$ is defined as $\diam(G) = \max_{u,v\in \ag{}}\dist{G}{u,v}$, i.e., the maximum length of any shortest path in $G$. Finally, a useful measure for the centrality of a vertex in a network is its distance to a set of vertices. Given a subset $\ag{}'\subseteq \ag{}$ of vertices, let $\dist{G}{u,\ag{}'}=\sum_{v\in \ag{}'}\dist{G}{u,v}$ denote the sum of distances from $u$ to all vertices in $\ag{}'$. Also, given a subset of agents $C\subseteq \ag{}$, we denote by $G[C]$ the subgraph of $G$ \emph{induced by} $C$, i.e., $G[C] = (C,F)$, where $F = \{uv\in E\colon u,v, \in C\}$.

Before formally defining our network creation model, we introduce some relevant special types of graphs. The graph $\com_n = (\ag{},E)$ is called \emph{complete} if $E = \{uv \colon u,v,\in \ag{}\}$, i.e., all possible edges are present. Further, $\sta_n = (\ag{},E)$ is called \emph{star} if some $u\in \ag{}$ exists such that $E = \{uv\colon v\in \ag{}\setminus \{u\}\}$. We also define two networks for the special case of $2$~types. Given two agents $u\in \ag{B}$ and $v\in \ag{R}$, the network $\dsta_n = (\ag{},E)$ is called \emph{double star} if $E = uv \cup \{uw\colon w\in \ag{B}\} \cup \{vw\colon w\in \ag{R}\}$ and $\dss_n = (\ag{},E)$ is called \emph{double star with switched centers} if $E = uv \cup \{uw\colon w\in \ag{R}\} \cup \{vw\colon w\in \ag{B}\}$.
An undirected graph $G$ is called \emph{complete}, \emph{star}, \emph{double star}, or \emph{double star with exchanged centers} if it is isomorphic\footnote{Here, isomorphisms must preserve agent types, i.e., map vertices associated to blue and red agents to such vertices, respectively.} to $\com_n$, $\sta_n$, $\dsta_n$, or $\dss_n$, respectively.

\paragraph*{Network Creation Games with Homophilic Agents}

We study network creation within a cost-oriented bilateral model \`a la \citet{CoPa05a}, where the agent cost is separated into a neighborhood cost encompassing the cost of sponsoring edges and a distance cost encompassing the cost of the agents' centrality. In both of our models, a created network $G$ has a \emph{distance cost} for agent $u$ of $\Dcost{G}{u} := \dist{G}{u,\ag{}}$, i.e., the sum of agent $u$'s distances to all other agents. The neighborhood cost is different in our two models and will be specified in the definition of our network creation games.

To model the cost dependency on the types of neighbors, we define the set of same-type agents in the neighborhood of agent $u$ as
$\frs{G}{u} = \ag{T}\cap \nbh{G}{u}$, if $u \in \ag{T}$ for some type $T\in\mathcal T$.
Also, the set of other-type neighboring agents is $\ens{G}{u} = \nbh{G}{u}\setminus \frs{G}{u}$, and let the cardinalities of these sets be $\fr{G}{u} = |\frs{G}{u}|$ and $\en{G}{u} = |\ens{G}{u}|$, respectively. 

Now we define our network creation games.
A \emph{network creation game with increasing comfort among friends} (\icfgame) with cost parameter $\alpha > 0$ is a network creation game where the neighborhood cost is given by 
$$a_G^{\mathit{ICF}}(u) = \deg{G}{u}\cdot\alpha\left(1+\frac{1}{\fr{G}{u}+1}\right),$$ i.e., there is a fixed cost of $\alpha$ for every edge and an additional cost that decreases with an increasing number of friends.

A \emph{network creation game with decreasing effort of integration} (\deigame) with cost parameter $\alpha > 0$ is a network creation game where the neighborhood cost ist given by $$a_G^{\mathit{DEI}}(u) = \alpha\Bigg(\deg{G}{u} + \sum_{k=1}^{\en{G}{u}}\frac 1k \Bigg).$$ Hence, there is a fixed edge cost of $\alpha$ for every edge to an agent in the neighborhood together with a harmonically decreasing additional cost for edges towards other-type agents. Note that the sum is empty for $\en{G}{u} = 0$, and therefore, the game is identical to the single-type bilateral network creation game by \citet{CoPa05a} if $k = 1$.

For the neighborhood cost, we omit the superscript indicating the type of network creation game, whenever this is clear from the context. Also, for both of our models, we define the total cost as $\Tcost{G}{u} = \Ncost{G}{u} + \Dcost{G}{u}$. 

The cost functions mimic the two effects that we want to model, namely a general homophilic behavior via the \icfgame and diminishing prejudices with intensified contact via the \deigame.
	In both models, edge costs are in the range $[\alpha, 2\alpha]$. In the \icfgame, we assume that the cost of edges is $2\alpha$ for each edge if an agent has no friends, and the edge cost is approaching $\alpha$ when the number of neighboring friends is growing. In the \deigame, the cost for edges to friends is always $\alpha$ and the variable cost only affects other-type agents, where we approach $\alpha$ with a harmonic decay. Our cost functions represent one way to capture these ``monotonicities'', having a similar decay structure and cost range to ensure their comparability.

\paragraph*{Measures for Desirable Networks}
We analyze networks by the incentives of agents to maintain the network in terms of stability and by the diversity of their neighborhood with respect to other agent types. Following \citet{JaWo96a}, a network $G = (\ag{},E)$ is called \emph{pairwise stable} if the following two conditions are satisfied:

\begin{compactenum}
	\item[(i)] For all agents $u\in \ag{}$ and neighbors $v\in \nbh{G}{u}$, it holds that $\Tcost{G}{u} \le \Tcost{G-uv}{u}$, i.e., no agent can benefit from unilaterally severing an edge, and
	\item[(ii)] for all agents $u\in \ag{}$ and non-neighbors $v\notin \nbh{G}{u}$, it holds that $\Tcost{G}{u} \le \Tcost{G+uv}{u}$ or $\Tcost{G}{v} \le \Tcost{G+uv}{v}$, i.e., no pair of agents can bilaterally create an edge such that the individual cost for both agents decreases.
\end{compactenum}

\noindent Connectivity is an important aspect in network analysis. With multiple agent types, the internal connectivity per type deserves special consideration. Formally, a network $G= (\ag{},E)$ is called \emph{fully intra-connected} if, for every pair of same-type agents $u,v\in \ag{}$, it holds that $uv\in E$. Further, $G$ is \emph{fully connected} if $G$ is complete.

For the evaluation of diversity, we consider two segregation measures. Given a network $G = (\ag{},E)$, its \emph{local segregation}, denoted by $\lsm(G)$, is defined as the average fraction of agents of the same type, i.e., $\lsm(G) = \frac 1{|\ag{}|} \sum_{u\in \ag{}} \frac{\fr{G}{u}}{\deg{G}{u}}$. The \emph{global segregation}, called $\gsm(G)$, is the proportion of monochromatic edges, i.e., $\gsm(G) = \frac{\sum_{u\in \ag{}}\fr{G}{u}}{2|E|}$. Note that $\frac{1}{2}\sum_{u\in \ag{}}\fr{G}{u}$ is the number of monochromatic edges.\footnote{$\lsm$ and $\gsm$ are (related to) standard measures in social sciences to capture the agents' \emph{exposure}~\citep{massey1988dimensions}. $\lsm$ is used by~\cite{paolillo2018different} and $\gsm$ is used in the simulation framework Netlogo~\citep{netlogo} and by~\cite{zhang2011tipping}.}

Finally, the minimum willingness to integration by an agent can be evaluated by checking if she entertains any bichromatic edge. Therefore, we call an agent \emph{curious} if she is part of a bichromatic edge. Similarly, a type of agents is called \emph{curious} if it solely consists of curious agents. Note that this concept is related to the degree of integration, which is identical to the number of curious agents and has been studied in game-theoretic segregation models~\citep{AEGISV21}.

\section{Increasing Comfort among Friends}\label{sect:icf}
In this section we perform our theoretical analysis of the \icfgame. Unless explicitly stated otherwise, all statements hold for an arbitrary number of agent types. All missing proofs can be found in the technical appendix.

We start by gathering some statements concerning structural properties and simple pairwise stable networks. Their proof follows by a careful analysis of the cost difference after the creation and deletion of edges.

\begin{restatable}{proposition}{ICFproperties}\label{prop:ICF-properties}
	For the \icfgame the following holds:
	\begin{compactenum}
		\item If $\alpha < \frac 67$, then every pairwise stable network is fully intra-connected.\label{prop:fullintraconn}
		\item If $\alpha < \frac43$, then $\diam(G)\le 2$ for every pairwise stable network $G$. In particular, $G$ contains a curious type.
		\label{prop:curioustype}
		\item Let $\alpha < 1$, $G$ a pairwise stable network, and $C\subseteq \ag{}$  such that every agent in $C$ is curious and $C\subseteq \ag{T}$ for some type $T\in \mathcal T$.
		Then, $G[C]$ is a clique. In particular, every curious type of agents is fully intra-connected.
		\label{prop:curiouscliques}
		\item If $\alpha\leq\frac{\numb{B}}{\numb{B}+1}$, then the complete network $\com_n$ is pairwise stable. Moreover for $\alpha < \min\{\frac{6}{7}, \frac{\numb{B}}{\numb{B}+1}\}$, $\com_n$ is the unique pairwise stable network.
		\label{prop:ICF-cliques}
		\item If $\alpha\geq 1$, then the star $\sta_n$ is pairwise stable.
		\label{prop:ICF-stars}
	\end{compactenum}
	
\end{restatable}

The uniqueness in \Cref{prop:ICF-properties}(\ref{prop:ICF-cliques}) excludes the parameter range $\frac 67 \le \alpha \le \frac{\numb{B}}{\numb{B}+1}$, which can only happen for sufficiently many blue agents. In fact, there the uniqueness ceases to hold, as we show in \Cref{ex:PSN-exist-unique-small} in the appendix. 

For the existence of stable networks, we still have to consider the intermediate parameter range $\frac{\numb{B}}{\numb{B}+1}<\alpha<1$. We provide the construction for two agent types. The general case is covered in the appendix.

\begin{restatable}{proposition}{ExPSNintermed}\label{prop:exist-PSN-intermed}
	In the \icfgame, there exists a pairwise stable network for every $\frac{\numb{B}}{\numb{B}+1} \le \alpha < 1$.
\end{restatable}

\begin{proof}[Construction for two agent types]
	Consider an instance of the \icfgame and let $\frac{\numb{B}}{\numb{B}+1} \le \alpha < 1$. We will define a stable network for $\alpha$ dependent on the threshold $\tau = \frac {\numb{B}(\numb{B} + 1)}{\numb{B}(\numb{B}+1)+1}$. Note that $\frac{\numb{B}}{\numb{B}+1} < \tau <1$, as $\numb{B}(\numb{B} + 1) > \numb{B}$.
	
	We assume $\ag{B} = \{b_1,\dots, b_{\numb{B}}\}$ and $\ag{R} = \{r_1,\dots, r_{\numb{R}}\}$ and define the edge set of the graph $G = (\ag{},E)$ as follows:
	
	\begin{compactitem}
		\item $\{x_i,x_j\}\in E$, for $x\in \{b,r\}, i,j\in \{1,\dots, \numb{B}\}$,
		\item $\{r_i,b_i\}\in E$, for $i\in \{1,\dots, \numb{B}\}$,
		\item $\{r_i, r_j\}\in E$, for $i\in \{1,\dots, \numb{B}\}$ and $j\in \{\numb{B}+1,\dots, \numb{R}\}$,
		\item if $\alpha < \tau$, then $\{r_i, r_j\}\in E$, for $i,j\in \{\numb{B}+1,\dots, \numb{R}\}$, and no further edges are in $E$;\\ otherwise, no further edges are in $E$.
	\end{compactitem} 
	The two cases for the network $G$ are illustrated in \Cref{fig:PSN-exist-intermed}. They can be shown to be pairwise stable for their respective parameter range. 
\end{proof}
	\begin{figure}
		\centering
		\begin{tikzpicture}[scale = .7]
		\node[protovertex] (b2) at (0.5,0) {};
		\node[protovertex] (b1) at (108:1) {};
		\node[protovertex] (b3) at (252:1) {};

		\node (c) at (3,0){};

		\node[protodiam] (r4) at ($(c) + (36:1)$) {};
		\node[protodiam] (r1) at ($(c) + (108:1)$) {};
		\node[protodiam] (r2) at ($(c) + (180:1)$) {};
		\node[protodiam] (r3) at ($(c) + (252:1)$) {};
		\node[protodiam] (r5) at ($(c) + (324:1)$) {};

		\foreach \i in {1,2,3}
		{\draw (b\i) edge (r\i);
			\draw (r4) edge (r\i);
			\draw (r5) edge (r\i);
		}
		\foreach \x/\y in {1/2,2/3,3/1}
		{\draw  (b\x) edge (b\y);
			\draw  (r\x) edge (r\y);}
		\draw (r4) edge (r5);
		\end{tikzpicture}
		\qquad
		\begin{tikzpicture}[scale = .7]
		\node[protovertex] (b2) at (0.5,0) {};
		\node[protovertex] (b1) at (108:1) {};
		\node[protovertex] (b3) at (252:1) {};

		\node (c) at (3,0){};

		\node[protodiam] (r4) at ($(c) + (36:1)$) {};
		\node[protodiam] (r1) at ($(c) + (108:1)$) {};
		\node[protodiam] (r2) at ($(c) + (180:1)$) {};
		\node[protodiam] (r3) at ($(c) + (252:1)$) {};
		\node[protodiam] (r5) at ($(c) + (324:1)$) {};

		\foreach \i in {1,2,3}
		{\draw (b\i) edge (r\i);
			\draw (r4) edge (r\i);
			\draw (r5) edge (r\i);
		}
		\foreach \x/\y in {1/2,2/3,3/1}
		{\draw  (b\x) edge (b\y);
			\draw  (r\x) edge (r\y);}
		\end{tikzpicture}
		\caption{Pairwise stable networks for $\frac{\numb{B}}{\numb{B}+1} \le \alpha < \tau$ (left) and $\tau \le \alpha < 1$ (right).}
		\label{fig:PSN-exist-intermed}
	\end{figure}
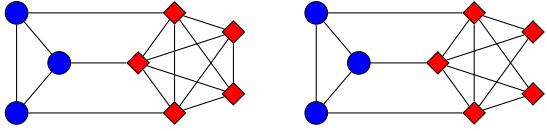

Interestingly, the stable networks constructed in the previous proof give an almost full characterization of stable networks for the considered range of edge costs when $k = 2$.

\begin{restatable}{theorem}{UniquePSNintermed}\label{thm:unique-PSN-intermed}
	Consider the \icfgame with parameter $\alpha$ and $k = 2$ agent types. Let $\frac{\numb{R}}{\numb{R}+1} < \alpha < 1$ and assume that $G$ is pairwise stable.  
	Then, the blue agents are fully intra-connected, the bichromatic edges form a matching of size $\numb{B}$, and curious red agents are connected to all other red agents.
\end{restatable}

\begin{proof}[Proof sketch]
	Let $\frac{\numb{R}}{\numb{R}+1} < \alpha < 1$ and assume that $G$ is pairwise stable network in the \icfgame with cost parameter $\alpha$. By \Cref{prop:ICF-properties}(\ref{prop:curioustype}), the diameter of $G$ is bounded by $2$ and there exists a curious type of agents. By \Cref{prop:ICF-properties}(\ref{prop:curiouscliques}), the curious type of agents forms a clique $C$ and the curious agents of the other type form a clique as well.

	Now, it can be shown that the bichromatic edges form a matching by proving that any agent incident to two such edges can sever one of them. Therefore, only a minority type can be a curious type and we can conclude that the blue agents are fully intra-connected and that the matching of bichromatic edges is of size $\numb{B}$. It remains to show that all curious red agents maintain edges with non-curious red agents. Assume that $y$ is a curious red agent forming a bichromatic edge to the blue agent $x$ and that there is no edge to a non-curious red agent $z$, i.e., $yz$ is not present in $G$. But then, $\dist{G}{x,z}\ge 3$, contradicting \Cref{prop:ICF-properties}(\ref{prop:curioustype}).
\end{proof}

\begin{example}\label{ex:PSN-exist-intermed-small}
	The characterization encountered in \Cref{thm:unique-PSN-intermed} does not cover the whole range of \Cref{prop:exist-PSN-intermed}. In fact, it does not hold for $\frac{\numb{B}}{\numb{B} + 1}\le \alpha \le \frac{\numb{R}}{\numb{R} + 1}$. Hence, further pairwise stable networks exist. Assume that $\numb{R}\ge 2$ and let $r^*\in \ag{R}$. Consider the network $G=(\ag{},E)$, where $E = \{\{v,w\}\colon v,w\in \ag{R}\}\cup \{\{v,w\}\colon v,w\in \ag{B}\}\cup \{\{v,r^*\}\colon v\in \ag{B}\}$, i.e., the network is fully intra-connected and there is a special agent $r^*$ to which all blue agents are connected. The structure of this network is depicted in \Cref{fig:PSN-exist-intermed-small}.	
	It is straightforward to check that the network is pairwise stable. \hfill $\lhd$
\end{example}
	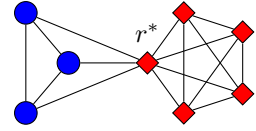
\begin{figure}
		\centering
		\begin{tikzpicture}[scale = .7]
		\node[protovertex] (b2) at (.5,0) {};
		\node[protovertex] (b1) at (108:1) {};
		\node[protovertex] (b3) at (252:1) {};
		
		\node (c) at (3,0){};
		
		\node[protodiam] (r4) at ($(c) + (36:1)$) {};
		\node[protodiam] (r1) at ($(c) + (108:1)$) {};
		\node[protodiam, label = $r^*$] (r2) at ($(c) + (180:1)$) {};
		\node[protodiam] (r3) at ($(c) + (252:1)$) {};
		\node[protodiam] (r5) at ($(c) + (324:1)$) {};
		
		\foreach \i in {1,2,3}
		{\draw (b\i) edge (r2);
			\draw (r4) edge (r\i);
			\draw (r5) edge (r\i);
		}
		\foreach \x/\y in {1/2,2/3,3/1}
		{\draw  (b\x) edge (b\y);
			\draw  (r\x) edge (r\y);}
		\draw (r4) edge (r5);
		\end{tikzpicture}
		
		\caption{Pairwise stable network for $\frac{\numb{B}}{\numb{B}+1} \le \alpha \le \frac{\numb{R}}{\numb{R} + 1}$ with $\numb{B} =3$ and $\numb{R} = 5$ blue and red agents, respectively.}
		\label{fig:PSN-exist-intermed-small}
	\end{figure}

Until now, we set our focus on the existence of pairwise stable networks. In the remainder of the section, we want to consider the segregation of pairwise stable networks. 
First, \Cref{thm:unique-PSN-intermed} yields very high segregation for $\frac{\numb{R}}{\numb{R}+1} < \alpha < 1$.

\begin{restatable}{corollary}{SegregationICFintermed}\label{cor:SegrICFintermed}
	Consider the \icfgame with parameter $\alpha$ and $k = 2$ agent types. Let $\frac{\numb{R}}{\numb{R}+1} < \alpha < 1$ and assume that $G$ is pairwise stable. Then, $\gsm(G) \ge 1 - \frac 1n$ and $\lsm(G) \ge 1 - \frac 2n$.
\end{restatable}

Hence, we know that segregation is low for sufficiently low parameter $\alpha$, where cliques are (uniquely) pairwise stable. Then, there is a transition at $\alpha = \frac{\numb{R}}{\numb{R}+1}$, where segregation is provably high regardless of further parameters like the distribution of agents into types. Once, the cost parameter increases to $\alpha \ge 1$, the picture becomes less clear. Stars can have very high and very low segregation.

\begin{restatable}{proposition}{IFCstarSeg}
	Consider the \icfgame with parameter $\alpha \ge 1$. Then, for every $n\ge 2$, there exist pairwise stable networks $G$ and $G'$ on $n$ nodes such that $\gsm(G) = \lsm(G) = 1$ and
	$\gsm(G') = \lsm(G') = \frac 1{n-1}$.
\end{restatable}

The networks in the previous proposition have the drawback that we need to fix the exact numbers of agents of each type to obtain the desired segregation. By contrast, for $\alpha\ge \frac 43$, the double star is always highly segregated.

\begin{restatable}{proposition}{DSstability}\label{thm:DS-stability}
	 Consider the \icfgame with $\alpha \ge \frac 43$. Then, the double star $\dsta_n$ is a pairwise stable network with $\gsm(\dsta_n) = 1 - \frac 1{\numb{}-1}$ and $\lsm(\dsta_n) \ge 1 - \frac 2 {\numb{}}$.
\end{restatable}

\section{Decreasing Effort of Integration}
In this section, we consider the \deigame. We start by collecting some results determining simple stable networks for sufficiently small and large values of~$\alpha$, respectively. Note that we implicitly assume the restriction to two agent types when considering the networks $\dsta_n$ and $\dss_n$. All other statements hold for an arbitrary number of agent types.

\begin{restatable}{proposition}{DEIoverview}\label{prop:PSN-overview}
	For the \deigame the following holds:
	
\begin{compactenum}
	\item If $\alpha < \frac 12$, then $\com_n$ is the unique pairwise stable network.
	
	\item If $\alpha < 1$, then every pairwise stable network is fully intra-connected.\label{prop:PSN-over-intra}
	
	\item If $\alpha < 1$, then every pairwise stable network $G$ satisfies $\diam(G) \le 2$.\label{prop:PSN-over-diam}
	
	\item The network $\com_n$ is pairwise stable if $\alpha \le \frac{n - \numb{R}}{n - \numb{R} + 1}$.\label{prop:PSN-over-kn}
	
	\item If $\alpha \ge 1$, then $\sta_n$ and $\dsta_n$ are pairwise stable networks.
	
	\item If $\alpha \ge \frac 43$, then $\dss_n$ is a pairwise stable network.
\end{compactenum}
\end{restatable}

\Cref{prop:PSN-overview}(\ref{prop:PSN-over-intra}) and \Cref{prop:PSN-overview}(\ref{prop:PSN-over-diam}) imply that, for $\alpha < 1$, every pairwise stable network consists of two monochromatic cliques and one type of agents is curious. Still, there are highly segregated pairwise stable networks.
Also, note that the highly integrated clique investigated in \Cref{prop:PSN-overview}(\ref{prop:PSN-over-kn}) is not the unique stable network, as the next examples shows for the case $k=2$. 
\begin{example}\label{ex:DEI-stability}
	Assume $k = 2$ and $\frac 12 \le \alpha \le \frac{\numb{R}}{\numb{R} + 1}$. Recall that $\numb{R}$ is the size of the majority type of agents. In particular, this covers the case $\alpha \le \frac{\numb{B}}{\numb{B} + 1} = \frac{n - \numb{R}}{n - \numb{R} + 1}$. Assume that $\numb{B}\ge 2$ and let $b^*$ be some fixed blue agent, i.e., an agent from the minority type. Consider the network $G = (\ag{},E)$ with $E = \{vw\colon v,w\in R\}\cup \{vw\colon v,w\in B\}\cup \{vb^*\colon v\in R\}$, i.e., the network is fully intra-connected and there is a special blue agent $b^*$ to which all red agents are connected. There are no further bichromatic edges. See \Cref{fig:PSN-example-matchhub}.
	
	\begin{figure}[h]
		\centering
		\begin{tikzpicture}[scale = .7]
		\node[protovertex, label = $b^*$] (b2) at (.5,0) {};
		\node[protovertex] (b1) at (108:1) {};
		\node[protovertex] (b3) at (252:1) {};
		
		\node (c) at (3,0){};
		
		\node[protodiam] (r4) at ($(c) + (36:1)$) {};
		\node[protodiam] (r1) at ($(c) + (108:1)$) {};
		\node[protodiam] (r2) at ($(c) + (180:1)$) {};
		\node[protodiam] (r3) at ($(c) + (252:1)$) {};
		\node[protodiam] (r5) at ($(c) + (324:1)$) {};
		
		\foreach \i in {1,2,3}
		{\draw (r4) edge (r\i);
		\draw (r5) edge (r\i);
		}
		\foreach \i in {1,2,3,4,5}
		{\draw (r\i) edge (b2);
		}
		\foreach \x/\y in {1/2,2/3,3/1}
		{\draw  (b\x) edge (b\y);
			\draw  (r\x) edge (r\y);}
		\draw (r4) edge (r5);
		\end{tikzpicture}
		
		\caption{Pairwise stable network for $\frac 12 \le \alpha \le \frac{\numb{R}}{\numb{R} + 1}$. The example contains $\numb{B} =3$ and $\numb{R} = 5$ blue and red agents, respectively.}
		\label{fig:PSN-example-matchhub}
	\end{figure}
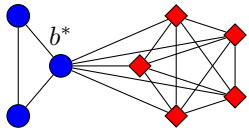
	
	This network is pairwise stable. Indeed, no agent can sever a monochromatic edge. Red agents cannot sever the bichromatic edge, because this decreases the distance to every blue agent by $1$. The blue agent $b^*$ cannot sever a bichromatic edge, because this increases her cost by $1- \alpha \frac{\numb{R} +1}{\numb{R}}\ge 0$. Also, further bichromatic edges cannot be added since their cost is more than $1$ for a blue agent while decreasing the distance cost only by $1$. \hfill $\lhd$
\end{example}

In the previous example, it was still possible to simultaneously have full intra-connectivity while there are agents entertaining several bichromatic edges. This is not possible anymore if we further increase $\alpha$.

\begin{restatable}{lemma}{FewBichrom}\label{lem:few-bichrom}
	Let $k = 2$ in the \deigame. Consider a fully intra-connected and pairwise stable network $G$.
	
	\begin{compactenum}
		\item If $\alpha > \frac{\numb{B}}{\numb{B} + 1}$, then every red agent in $G$ entertains at most one bichromatic edge.
		\item If $\alpha > \frac{\numb{R}}{\numb{R} + 1}$, then every agent in $G$ entertains at most one bichromatic edge.
	\end{compactenum}
\end{restatable}

As a consequence, we can even characterize all pairwise stable networks for $\frac{\numb{R}}{\numb{R} + 1} < \alpha < 1$ and $k = 2$.

\begin{theorem}\label{thm:DEIcharactPSN}
	Let $k = 2$ in the \deigame. Assume that $\frac{\numb{R}}{\numb{R} + 1} < \alpha < 1$ and consider a network $G$. Then, $G$ is pairwise stable if and only if it is fully intra-connected and its bichromatic edges form a matching covering $\ag{B}$.
\end{theorem}

\begin{proof}
	Clearly, if $k = 2$ and $\numb{R} = 1$, then the unique stable network consists of a neighboring blue and red agent. Hence, the assertion is true. Thus, we may assume that $\numb{R}\ge 2$.
	
	Let $\frac{\numb{R}}{\numb{R} + 1} < \alpha < 1$ and assume first that $G$ is a pairwise stable network. By \Cref{prop:PSN-overview}(\ref{prop:PSN-over-intra}), the network is fully intra-connected. By \Cref{lem:few-bichrom}, the bichromatic edges form a matching. Finally, by \Cref{prop:PSN-overview}(\ref{prop:PSN-over-diam}), one type of the agents must be curious, and therefore the matching covers the minority type of agents.
	
	Conversely, assume that $G$ is a fully intra-connected network such that the bichromatic edges form a matching covering one type of agents. Then, no edge can be severed because monochromatic edges only decrease the neighborhood cost by $\alpha < 1$ while increasing the distance cost by $1$. Also, bichromatic edges decrease the neighborhood cost by $2\alpha < 2$ while increasing the distance cost by $2$. Finally, it is impossible to create another bichromatic edge. This edge would be the second bichromatic edge incident to its endpoint from the minority type of agents. This agent would only decrease her distance cost by $1$ while increasing her neighborhood cost by $\frac 32\alpha \ge \frac 32 \frac{\numb{R}}{\numb{R} + 1} \ge 1$, where we use $\numb{R}\ge 2$ in the last step. 
\end{proof}

The second part of the above proof shows that the networks characterized in the theorem are even stable for $\frac 23 \le \alpha < 1$. Putting together \Cref{prop:PSN-overview}, \Cref{ex:DEI-stability}, and \Cref{thm:DEIcharactPSN}, we have proved the existence of pairwise stable networks for almost every \deigame if $k = 2$ (except a limit case when $\numb{B} = 1$). By generalizing the encountered networks, we can show their existence for an arbitrary number of types. Interestingly, the generalization of the network in \Cref{ex:DEI-stability} is straightforward, maintaining the property that there exists one specific agent entertaining all bichromatic edges. On the other hand, the generalization of the network in \Cref{thm:DEIcharactPSN} is a bit disguised. We define the network by providing an efficient algorithm. This algorithm initially considers a fully intra-connected network and adds edges by having agents create bichromatic edges via specific better responses. In the special case of $k = 2$, this results precisely in the matchings encountered in the previous theorem.

\begin{restatable}{theorem}{DEIexistPSN}\label{thm:DEIexistPSN}
	In the \deigame pairwise stable networks always exist.
\end{restatable}

Finally, we want to consider the segregation of pairwise stable networks in the \deigame. Clearly, the segregation does only depend on the networks, not on the type of NCG. Hence, we transfer from the investigation of \icfgame{s} that cliques provide low segregation for small $\alpha$, stars provide high or low segregation for high $\alpha$, but require a specific distribution of agents. Independently of this distribution, double stars provide high segregation and it is clear that $\gsm(\dss_n) = \lsm(\dss_n) = 0$. Finally, for an intermediate range of $\alpha$, high segregation is guaranteed. 

\begin{restatable}{corollary}{DEIuniquePSN}\label{thm:DEIuniquePSN}
	Let $k = 2$ and $\frac{\numb{R}}{\numb{R} + 1} < \alpha < 1$. Then, every pairwise stable network $G$ in the \deigame with parameter~$\alpha$ satisfies
	$\gsm(G) \ge 1 -\frac 2n$ and $\lsm(G) \ge 1 - \frac 2n$.
\end{restatable}

\section{Experimental Analysis}\label{sec:experiments}
While our theoretical results indicate a clear structure of stable networks for $\alpha \leq 1$, there is a broad range of possibilities for larger $\alpha$. Therefore, we support the theoretical findings for $\alpha > 1$ by a detailed experimental analysis.
To this end, we simulate a simple dynamic process based on distributed and strategic edge creation and deletion over time, incentivized by optimizing the cost functions of our two models.

The dynamics start with sparse initial networks (spanning tree or grid) and distribute agents of two equally-sized types such that the segregation of the initial network is very low or high. In each step, one agent is activated uniformly at random and can either create or delete an edge, performing a best response with respect to the cost function under consideration. In particular, we study also an \emph{add-only} variant of the model, where agents can only create edges. This dynamics is particularly natural when modeling social networks, as confirmed by the observation that many real-world social networks get denser over time~\citep{leskovec2005graphs}. The dynamics proceed until the consideration of no agent changes the network.   

\Cref{appendix:plots} contains a detailed discussion of our experimental setup and further results. An exemplary consideration of the dynamics based on the cost function of the \deigame can be found in \Cref{plot:DEINCG} and \Cref{plot:DEINCG_AddOnly} for the general and add-only version, respectively.\footnote{As discussed in \Cref{appendix:plots}, for computational efficiency, we consider convergence to $1.01$-approximate pairwise stable states which is  qualitatively similar to pairwise stability. Also, note that ``random'' means that the initial agent distribution is chosen uniformly at random, which implies low initial segregation.}
 Interestingly, the results for the \icfgame are qualitatively the same, regardless of measuring segregation with the local or global segregation measure. The experiments indicate that the segregation strength is proportional to $\alpha$, with low segregation for low $\alpha$, despite the theoretical necessity of high segregation for $\alpha$ close to $1$.\footnote{The provably high segregation for $\alpha <1$ close to $1$ is not contradicting the experimental results. Just before we transition to an edge parameter of at least $\alpha$, we hit the sweet spot where buying monochromatic edges is desirable while buying bichromatic edges is not.} Moreover, except for high $\alpha$, the initial agent distribution influences the segregation, with more observed segregation for segregated initial states. On the other hand, the structure of the initial network seems less important for the qualitative behavior.
 Interestingly, the add-only version displays a similar behavior for low $\alpha$, but the behavior changes drastically for moderately high $\alpha$. Instead of high segregation, we find that initially integrated networks converge to only moderately segregated states, whereas this is not true for initially segregated networks, suggesting an escape route from segregation.

\begin{figure}[h]
\centering
\begin{minipage}{0.23\textwidth}
\resizebox {\textwidth} {!} {
\begin{tikzpicture}
\begin{axis}[
legend columns=-1,
legend entries={segregated grid;,random grid;,segregated tree;,random tree},
legend to name=named,
legend style={nodes={scale=0.65, transform shape}},
xlabel= {$\alpha$},
ylabel= {local segregation},
boxplot/draw direction=y,
baseline,
xtick = {1,2,3, 4, 5, 6, 7, 8},
xticklabels = {5, 10, 15, 20, 25, 30, 35, 40},
ymin=0.5,
ymax=1
]
\addplot[red!50!black, domain=1.1:1.11]{0.55};
\addplot[blue!50!black, domain=1.1:1.11]{0.55};
\addplot[yellow!70!black, domain=1.1:1.11]{0.55};
\addplot[lime!70!black, domain=1.1:1.11]{0.55};

\addplot+ [color = lime!70!black,solid,boxplot prepared = {box extend=0.3, draw position = 1, lower whisker = 0.587877, lower quartile = 0.595792, median = 0.599611, upper quartile = 0.600736, upper whisker = 0.600876},]coordinates{}; 
\addplot+ [color = lime!70!black,solid,boxplot prepared = {box extend=0.3, draw position = 2, lower whisker = 0.614817, lower quartile = 0.639993, median = 0.648220, upper quartile = 0.655185, upper whisker = 0.659965},]coordinates{}; 
\addplot+ [color = lime!70!black,solid,boxplot prepared = {box extend=0.3, draw position = 3, lower whisker = 0.672204, lower quartile = 0.680865, median = 0.686433, upper quartile = 0.691970, upper whisker = 0.704313},]coordinates{}; 
\addplot+ [color = lime!70!black,solid,boxplot prepared = {box extend=0.3, draw position = 4, lower whisker = 0.693762, lower quartile = 0.710872, median = 0.720265, upper quartile = 0.725355, upper whisker = 0.743811},]coordinates{}; 
\addplot+ [color = lime!70!black,solid,boxplot prepared = {box extend=0.3, draw position = 5, lower whisker = 0.725430, lower quartile = 0.738375, median = 0.745606, upper quartile = 0.752434, upper whisker = 0.769720},]coordinates{}; 
\addplot+ [color = lime!70!black,solid,boxplot prepared = {box extend=0.3, draw position = 6, lower whisker = 0.746930, lower quartile = 0.759186, median = 0.773552, upper quartile = 0.779921, upper whisker = 0.783034},]coordinates{}; 

\addplot+ [color = yellow!70!black,solid,boxplot prepared = {box extend=0.3, draw position = 1, lower whisker = 0.625312, lower quartile = 0.645360, median = 0.649989, upper quartile = 0.654888, upper whisker = 0.664035},]coordinates{}; 
\addplot+ [color = yellow!70!black,solid,boxplot prepared = {box extend=0.3, draw position = 2, lower whisker = 0.715601, lower quartile = 0.723320, median = 0.726892, upper quartile = 0.729812, upper whisker = 0.739498},]coordinates{}; 
\addplot+ [color = yellow!70!black,solid,boxplot prepared = {box extend=0.3, draw position = 3, lower whisker = 0.764764, lower quartile = 0.776444, median = 0.782217, upper quartile = 0.785451, upper whisker = 0.796037},]coordinates{}; 
\addplot+ [color = yellow!70!black,solid,boxplot prepared = {box extend=0.3, draw position = 4, lower whisker = 0.802477, lower quartile = 0.815618, median = 0.821320, upper quartile = 0.823649, upper whisker = 0.830700},]coordinates{}; 
\addplot+ [color = yellow!70!black,solid,boxplot prepared = {box extend=0.3, draw position = 5, lower whisker = 0.835204, lower quartile = 0.842199, median = 0.845808, upper quartile = 0.849384, upper whisker = 0.856944},]coordinates{}; 
\addplot+ [color = yellow!70!black,solid,boxplot prepared = {box extend=0.3, draw position = 6, lower whisker = 0.857273, lower quartile = 0.865692, median = 0.870443, upper quartile = 0.872591, upper whisker = 0.883314},]coordinates{}; 

\addplot+ [color = red!70!black,solid,boxplot prepared = {box extend=0.3, draw position = 1, lower whisker = 0.697566, lower quartile = 0.705443, median = 0.709720, upper quartile = 0.711501, upper whisker = 0.719408},]coordinates{}; 
\addplot+ [color = red!70!black,solid,boxplot prepared = {box extend=0.3, draw position = 2, lower whisker = 0.760929, lower quartile = 0.767087, median = 0.769896, upper quartile = 0.772985, upper whisker = 0.779684},]coordinates{}; 
\addplot+ [color = red!70!black,solid,boxplot prepared = {box extend=0.3, draw position = 3, lower whisker = 0.809113, lower quartile = 0.817723, median = 0.820637, upper quartile = 0.822316, upper whisker = 0.832187},]coordinates{}; 
\addplot+ [color = red!70!black,solid,boxplot prepared = {box extend=0.3, draw position = 4, lower whisker = 0.841635, lower quartile = 0.849097, median = 0.851167, upper quartile = 0.853647, upper whisker = 0.859318},]coordinates{}; 
\addplot+ [color = red!70!black,solid,boxplot prepared = {box extend=0.3, draw position = 5, lower whisker = 0.863108, lower quartile = 0.869798, median = 0.871972, upper quartile = 0.874686, upper whisker = 0.878572},]coordinates{}; 
\addplot+ [color = red!70!black,solid,boxplot prepared = {box extend=0.3, draw position = 6, lower whisker = 0.880598, lower quartile = 0.884663, median = 0.886926, upper quartile = 0.888980, upper whisker = 0.895380},]coordinates{}; 

\addplot+ [color = blue!50!black,solid,boxplot prepared = {box extend=0.3, draw position = 1, lower whisker = 0.553215, lower quartile = 0.567747, median = 0.572520, upper quartile = 0.575852, upper whisker = 0.583499},]coordinates{}; 
\addplot+ [color = blue!50!black,solid,boxplot prepared = {box extend=0.3, draw position = 2, lower whisker = 0.605987, lower quartile = 0.612051, median = 0.617258, upper quartile = 0.620476, upper whisker = 0.633279},]coordinates{}; 
\addplot+ [color = blue!50!black,solid,boxplot prepared = {box extend=0.3, draw position = 3, lower whisker = 0.634162, lower quartile = 0.641935, median = 0.648416, upper quartile = 0.651590, upper whisker = 0.659823},]coordinates{}; 
\addplot+ [color = blue!50!black,solid,boxplot prepared = {box extend=0.3, draw position = 4, lower whisker = 0.638457, lower quartile = 0.661246, median = 0.664685, upper quartile = 0.669107, upper whisker = 0.674004},]coordinates{}; 
\addplot+ [color = blue!50!black,solid,boxplot prepared = {box extend=0.3, draw position = 5, lower whisker = 0.660377, lower quartile = 0.670119, median = 0.677976, upper quartile = 0.683920, upper whisker = 0.692616},]coordinates{}; 
\addplot+ [color = blue!50!black,solid,boxplot prepared = {box extend=0.3, draw position = 6, lower whisker = 0.671119, lower quartile = 0.684051, median = 0.690288, upper quartile = 0.695877, upper whisker = 0.711555},]coordinates{}; 

\end{axis}
\end{tikzpicture}
}
\end{minipage}
\begin{minipage}{0.23\textwidth}
\resizebox {\textwidth} {!} {
\begin{tikzpicture}
\begin{axis}[
xlabel= {$\alpha$},
ylabel= {local segregation},
boxplot/draw direction=y,
baseline,
xtick = {1, 2, 3, ..., 11},
xticklabels  = {5, 30, 55, 80, 105, 130, 155, 180, 205, 230, 255},
ymin=0.5,
ymax=1
]

\addplot+ [color = lime!70!black,solid,boxplot prepared = {box extend=0.4, draw position = 1, lower whisker = 0.587877, lower quartile = 0.595792, median = 0.599611, upper quartile = 0.600736, upper whisker = 0.600876},]coordinates{}; 
\addplot+ [color = lime!70!black,solid,boxplot prepared = {box extend=0.4, draw position = 2, lower whisker = 0.746930, lower quartile = 0.759186, median = 0.773552, upper quartile = 0.779921, upper whisker = 0.783034},]coordinates{}; 
\addplot+ [color = lime!70!black,solid,boxplot prepared = {box extend=0.4, draw position = 3, lower whisker = 0.873544, lower quartile = 0.875830, median = 0.878203, upper quartile = 0.882616, upper whisker = 0.884611},]coordinates{}; 
\addplot+ [color = lime!70!black,solid,boxplot prepared = {box extend=0.4, draw position = 4, lower whisker = 0.923530, lower quartile = 0.925024, median = 0.930053, upper quartile = 0.931133, upper whisker = 0.935467},]coordinates{}; 
\addplot+ [color = lime!70!black,solid,boxplot prepared = {box extend=0.4, draw position = 5, lower whisker = 0.957246, lower quartile = 0.959246, median = 0.959575, upper quartile = 0.961969, upper whisker = 0.962539},]coordinates{}; 
\addplot+ [color = lime!70!black,solid,boxplot prepared = {box extend=0.4, draw position = 6, lower whisker = 0.961707, lower quartile = 0.961947, median = 0.964808, upper quartile = 0.967901, upper whisker = 0.970280},]coordinates{}; 
\addplot+ [color = lime!70!black,solid,boxplot prepared = {box extend=0.4, draw position = 7, lower whisker = 0.968187, lower quartile = 0.968823, median = 0.972421, upper quartile = 0.974994, upper whisker = 0.975260},]coordinates{}; 
\addplot+ [color = lime!70!black,solid,boxplot prepared = {box extend=0.4, draw position = 8, lower whisker = 0.969983, lower quartile = 0.973441, median = 0.975758, upper quartile = 0.977213, upper whisker = 0.978670},]coordinates{}; 
\addplot+ [color = lime!70!black,solid,boxplot prepared = {box extend=0.4, draw position = 9, lower whisker = 0.973462, lower quartile = 0.976978, median = 0.978195, upper quartile = 0.978944, upper whisker = 0.978955},]coordinates{}; 
\addplot+ [color = lime!70!black,solid,boxplot prepared = {box extend=0.4, draw position = 10, lower whisker = 0.977007, lower quartile = 0.977624, median = 0.978862, upper quartile = 0.980195, upper whisker = 0.981108},]coordinates{}; 
\addplot+ [color = lime!70!black,solid,boxplot prepared = {box extend=0.4, draw position = 11, lower whisker = 0.978110, lower quartile = 0.980106, median = 0.982154, upper quartile = 0.982968, upper whisker = 0.983385},]coordinates{}; 

\addplot+ [color = yellow!70!black,solid,boxplot prepared = {box extend=0.4, draw position = 1, lower whisker = 0.625312, lower quartile = 0.645360, median = 0.649989, upper quartile = 0.654888, upper whisker = 0.664035},]coordinates{}; 
\addplot+ [color = yellow!70!black,solid,boxplot prepared = {box extend=0.4, draw position = 2, lower whisker = 0.857273, lower quartile = 0.865692, median = 0.870443, upper quartile = 0.872591, upper whisker = 0.883314},]coordinates{}; 
\addplot+ [color = yellow!70!black,solid,boxplot prepared = {box extend=0.4, draw position = 3, lower whisker = 0.923923, lower quartile = 0.926922, median = 0.930204, upper quartile = 0.933215, upper whisker = 0.941737},]coordinates{}; 
\addplot+ [color = yellow!70!black,solid,boxplot prepared = {box extend=0.4, draw position = 4, lower whisker = 0.949148, lower quartile = 0.951784, median = 0.954276, upper quartile = 0.955444, upper whisker = 0.959144},]coordinates{}; 
\addplot+ [color = yellow!70!black,solid,boxplot prepared = {box extend=0.4, draw position = 5, lower whisker = 0.952666, lower quartile = 0.958743, median = 0.962441, upper quartile = 0.964854, upper whisker = 0.971552},]coordinates{}; 
\addplot+ [color = yellow!70!black,solid,boxplot prepared = {box extend=0.4, draw position = 6, lower whisker = 0.959560, lower quartile = 0.963026, median = 0.964291, upper quartile = 0.966124, upper whisker = 0.970617},]coordinates{}; 
\addplot+ [color = yellow!70!black,solid,boxplot prepared = {box extend=0.4, draw position = 7, lower whisker = 0.966569, lower quartile = 0.968616, median = 0.970596, upper quartile = 0.972148, upper whisker = 0.975831},]coordinates{}; 
\addplot+ [color = yellow!70!black,solid,boxplot prepared = {box extend=0.4, draw position = 8, lower whisker = 0.969493, lower quartile = 0.973342, median = 0.974507, upper quartile = 0.975862, upper whisker = 0.979846},]coordinates{}; 
\addplot+ [color = yellow!70!black,solid,boxplot prepared = {box extend=0.4, draw position = 9, lower whisker = 0.973960, lower quartile = 0.975806, median = 0.977811, upper quartile = 0.978961, upper whisker = 0.982225},]coordinates{}; 
\addplot+ [color = yellow!70!black,solid,boxplot prepared = {box extend=0.4, draw position = 10, lower whisker = 0.974930, lower quartile = 0.978464, median = 0.979766, upper quartile = 0.980507, upper whisker = 0.983006},]coordinates{}; 
\addplot+ [color = yellow!70!black,solid,boxplot prepared = {box extend=0.4, draw position = 11, lower whisker = 0.977565, lower quartile = 0.980369, median = 0.982019, upper quartile = 0.983316, upper whisker = 0.985345},]coordinates{}; 

\addplot+ [color = blue!50!black,solid,boxplot prepared = {box extend=0.4, draw position = 1, lower whisker = 0.553215, lower quartile = 0.567747, median = 0.572520, upper quartile = 0.575852, upper whisker = 0.583499},]coordinates{}; 
\addplot+ [color = blue!50!black,solid,boxplot prepared = {box extend=0.4, draw position = 2, lower whisker = 0.671119, lower quartile = 0.684051, median = 0.690288, upper quartile = 0.695877, upper whisker = 0.711555},]coordinates{}; 
\addplot+ [color = blue!50!black,solid,boxplot prepared = {box extend=0.4, draw position = 3, lower whisker = 0.769120, lower quartile = 0.792679, median = 0.799753, upper quartile = 0.805371, upper whisker = 0.824674},]coordinates{}; 
\addplot+ [color = blue!50!black,solid,boxplot prepared = {box extend=0.4, draw position = 4, lower whisker = 0.902491, lower quartile = 0.911603, median = 0.915007, upper quartile = 0.916627, upper whisker = 0.923923},]coordinates{}; 
\addplot+ [color = blue!50!black,solid,boxplot prepared = {box extend=0.4, draw position = 5, lower whisker = 0.948760, lower quartile = 0.953995, median = 0.957272, upper quartile = 0.959069, upper whisker = 0.964313},]coordinates{}; 
\addplot+ [color = blue!50!black,solid,boxplot prepared = {box extend=0.4, draw position = 6, lower whisker = 0.959348, lower quartile = 0.962614, median = 0.965023, upper quartile = 0.966140, upper whisker = 0.971082},]coordinates{}; 
\addplot+ [color = blue!50!black,solid,boxplot prepared = {box extend=0.4, draw position = 7, lower whisker = 0.966164, lower quartile = 0.969773, median = 0.971402, upper quartile = 0.973260, upper whisker = 0.977559},]coordinates{}; 
\addplot+ [color = blue!50!black,solid,boxplot prepared = {box extend=0.4, draw position = 8, lower whisker = 0.968674, lower quartile = 0.974483, median = 0.975671, upper quartile = 0.976556, upper whisker = 0.980299},]coordinates{}; 
\addplot+ [color = blue!50!black,solid,boxplot prepared = {box extend=0.4, draw position = 9, lower whisker = 0.974160, lower quartile = 0.976849, median = 0.978065, upper quartile = 0.978941, upper whisker = 0.981402},]coordinates{}; 
\addplot+ [color = blue!50!black,solid,boxplot prepared = {box extend=0.4, draw position = 10, lower whisker = 0.976702, lower quartile = 0.978616, median = 0.979923, upper quartile = 0.980969, upper whisker = 0.983412},]coordinates{}; 
\addplot+ [color = blue!50!black,solid,boxplot prepared = {box extend=0.4, draw position = 11, lower whisker = 0.977296, lower quartile = 0.980610, median = 0.981960, upper quartile = 0.983154, upper whisker = 0.984985},]coordinates{}; 

\addplot+ [color = red!70!black,solid,boxplot prepared = {box extend=0.4, draw position = 1, lower whisker = 0.697566, lower quartile = 0.705443, median = 0.709720, upper quartile = 0.711501, upper whisker = 0.719408},]coordinates{}; 
\addplot+ [color = red!70!black,solid,boxplot prepared = {box extend=0.4, draw position = 2, lower whisker = 0.880598, lower quartile = 0.884663, median = 0.886926, upper quartile = 0.888980, upper whisker = 0.895380},]coordinates{}; 
\addplot+ [color = red!70!black,solid,boxplot prepared = {box extend=0.4, draw position = 3, lower whisker = 0.927817, lower quartile = 0.930200, median = 0.931157, upper quartile = 0.933681, upper whisker = 0.937793},]coordinates{}; 
\addplot+ [color = red!70!black,solid,boxplot prepared = {box extend=0.4, draw position = 4, lower whisker = 0.947143, lower quartile = 0.949648, median = 0.950862, upper quartile = 0.952350, upper whisker = 0.955776},]coordinates{}; 
\addplot+ [color = red!70!black,solid,boxplot prepared = {box extend=0.4, draw position = 5, lower whisker = 0.952846, lower quartile = 0.958263, median = 0.962840, upper quartile = 0.963716, upper whisker = 0.966948},]coordinates{}; 
\addplot+ [color = red!70!black,solid,boxplot prepared = {box extend=0.4, draw position = 6, lower whisker = 0.956431, lower quartile = 0.963584, median = 0.964688, upper quartile = 0.966167, upper whisker = 0.970800},]coordinates{}; 
\addplot+ [color = red!70!black,solid,boxplot prepared = {box extend=0.4, draw position = 7, lower whisker = 0.964810, lower quartile = 0.969006, median = 0.971059, upper quartile = 0.972356, upper whisker = 0.975870},]coordinates{}; 
\addplot+ [color = red!70!black,solid,boxplot prepared = {box extend=0.4, draw position = 8, lower whisker = 0.969395, lower quartile = 0.973264, median = 0.974714, upper quartile = 0.975619, upper whisker = 0.978041},]coordinates{}; 
\addplot+ [color = red!70!black,solid,boxplot prepared = {box extend=0.4, draw position = 9, lower whisker = 0.972500, lower quartile = 0.976252, median = 0.977093, upper quartile = 0.978713, upper whisker = 0.981757},]coordinates{}; 
\addplot+ [color = red!70!black,solid,boxplot prepared = {box extend=0.4, draw position = 10, lower whisker = 0.975554, lower quartile = 0.978756, median = 0.979826, upper quartile = 0.980922, upper whisker = 0.985774},]coordinates{}; 
\addplot+ [color = red!70!black,solid,boxplot prepared = {box extend=0.4, draw position = 11, lower whisker = 0.978110, lower quartile = 0.980720, median = 0.981614, upper quartile = 0.982151, upper whisker = 0.984563},]coordinates{}; 

\end{axis}
\end{tikzpicture}
}
\end{minipage}
\ref{named}
\caption{Local segregation of  $1.01$-approximate networks in the \deigame obtained by iterative best response moves for $n=1000$ over 50 runs starting from a random or segregated tree and grid.}
\label{plot:DEINCG}
\end{figure}
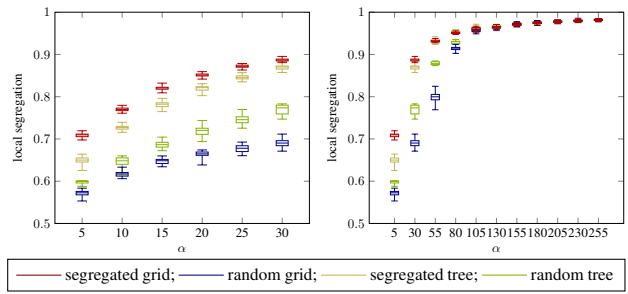

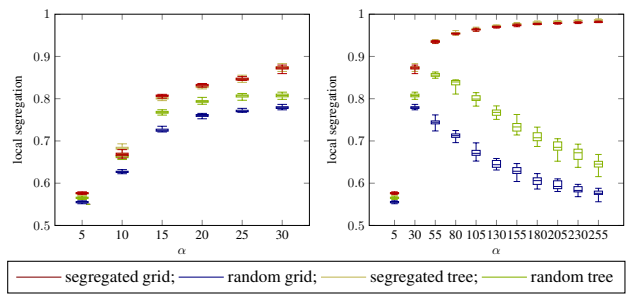
\begin{figure}[!ht]
\centering
\begin{minipage}{0.23\textwidth}
\resizebox {\textwidth} {!} {
\begin{tikzpicture}
\begin{axis}[
		legend columns=-1,
          legend entries={segregated grid;,random grid;,segregated tree;,random tree},
          legend to name=named,
		legend style={nodes={scale=0.65, transform shape}},
		xlabel= {$\alpha$},
		ylabel= {local segregation},
		boxplot/draw direction=y,
		baseline,
		xtick = {1,2,3, 4, 5, 6, 7, 8},
		xticklabels = {5, 10, 15, 20, 25, 30, 35, 40},
		ymin=0.5,
    		ymax=1
	]
\addplot[red!50!black, domain=1.1:1.2]{0.55};
\addplot[blue!50!black, domain=1.1:1.2]{0.55};
\addplot[yellow!70!black, domain=1.1:1.2]{0.55};
\addplot[lime!70!black, domain=1.1:1.2]{0.55};

\addplot+ [color = lime!70!black,solid,boxplot prepared = {box extend=0.3,  upper whisker = 0.569317, upper quartile = 0.566948, median = 0.565704, lower quartile = 0.563991, lower whisker = 0.561343},]coordinates{}; 
\addplot+ [color = lime!70!black,solid,boxplot prepared = {box extend=0.3, draw position = 2, lower whisker = 0.656054, lower quartile = 0.660461, median = 0.664114, upper quartile = 0.665906, upper whisker = 0.670238},]coordinates{}; 
\addplot+ [color = lime!70!black,solid,boxplot prepared = {box extend=0.3, draw position = 3, lower whisker = 0.761166, lower quartile = 0.765647, median = 0.767783, upper quartile = 0.770015, upper whisker = 0.773884},]coordinates{}; 
\addplot+ [color = lime!70!black,solid,boxplot prepared = {box extend=0.3, draw position = 4, lower whisker = 0.786826, lower quartile = 0.791560, median = 0.793824, upper quartile = 0.795343, upper whisker = 0.803133},]coordinates{}; 
\addplot+ [color = lime!70!black,solid,boxplot prepared = {box extend=0.3, draw position = 5, lower whisker = 0.796294, lower quartile = 0.804001, median = 0.806804, upper quartile = 0.808822, upper whisker = 0.812360},]coordinates{}; 
\addplot+ [color = lime!70!black,solid,boxplot prepared = {box extend=0.3, draw position = 6,  upper whisker = 0.815487, upper quartile = 0.810792, median = 0.807471, lower quartile = 0.804435, lower whisker = 0.798459},]coordinates{}; 

\addplot+ [color = yellow!70!black,solid,boxplot prepared = {box extend=0.3, draw position = 1, lower whisker = 0.571288, lower quartile = 0.574302, median = 0.575457, upper quartile = 0.576574, upper whisker = 0.577170},]coordinates{}; 
\addplot+ [color = yellow!70!black,solid,boxplot prepared = {box extend=0.3, draw position = 2, lower whisker = 0.669629, lower quartile = 0.680519, median = 0.682716, upper quartile = 0.685042, upper whisker = 0.693209},]coordinates{}; 
\addplot+ [color = yellow!70!black,solid,boxplot prepared = {box extend=0.3, draw position = 3, lower whisker = 0.795463, lower quartile = 0.799373, median = 0.800758, upper quartile = 0.802217, upper whisker = 0.804513},]coordinates{}; 
\addplot+ [color = yellow!70!black,solid,boxplot prepared = {box extend=0.3, draw position = 4, lower whisker = 0.823144, lower quartile = 0.826224, median = 0.828067, upper quartile = 0.830107, upper whisker = 0.833727},]coordinates{}; 
\addplot+ [color = yellow!70!black,solid,boxplot prepared = {box extend=0.3, draw position = 5, lower whisker = 0.838669, lower quartile = 0.844115, median = 0.845808, upper quartile = 0.848215, upper whisker = 0.855670},]coordinates{}; 
\addplot+ [color = yellow!70!black,solid,boxplot prepared = {box extend=0.3, draw position = 6, lower whisker = 0.864458, lower quartile = 0.870431, median = 0.873253, upper quartile = 0.875989, upper whisker = 0.882930},]coordinates{};

\addplot+ [color = red!70!black,solid,boxplot prepared = {box extend=0.3, draw position = 1, lower whisker = 0.573521, lower quartile = 0.575512, median = 0.576715, upper quartile = 0.577702, upper whisker = 0.579620},]coordinates{}; 
\addplot+ [color = red!70!black,solid,boxplot prepared = {box extend=0.3, draw position = 2, lower whisker = 0.659204, lower quartile = 0.665023, median = 0.667811, upper quartile = 0.671121, upper whisker = 0.679496},]coordinates{}; 
\addplot+ [color = red!70!black,solid,boxplot prepared = {box extend=0.3, draw position = 3, lower whisker = 0.801244, lower quartile = 0.805434, median = 0.807576, upper quartile = 0.808680, upper whisker = 0.810356},]coordinates{}; 
\addplot+ [color = red!70!black,solid,boxplot prepared = {box extend=0.3, draw position = 4, lower whisker = 0.827337, lower quartile = 0.830322, median = 0.832444, upper quartile = 0.833859, upper whisker = 0.835993},]coordinates{}; 
\addplot+ [color = red!70!black,solid,boxplot prepared = {box extend=0.3, draw position = 5, lower whisker = 0.843377, lower quartile = 0.844394, median = 0.846362, upper quartile = 0.848605, upper whisker = 0.852728},]coordinates{}; 
\addplot+ [color = red!70!black,solid,boxplot prepared = {box extend=0.3, draw position = 6, lower whisker = 0.859387, lower quartile = 0.870538, median = 0.872640, upper quartile = 0.874923, upper whisker = 0.879567},]coordinates{}; 

\addplot+ [color = blue!50!black,solid, boxplot prepared = {box extend=0.3, draw position = 1, upper whisker = 0.558332, upper quartile = 0.556623, median = 0.555612, lower quartile = 0.553640, lower whisker = 0.551390},]coordinates{}; 
\addplot+ [color = blue!50!black,solid,boxplot prepared = {box extend=0.3, draw position = 2, lower whisker = 0.622430, lower quartile = 0.625420, median = 0.626853, upper quartile = 0.628278, upper whisker = 0.632146},]coordinates{}; 
\addplot+ [color = blue!50!black,solid,boxplot prepared = {box extend=0.3, draw position = 3, lower whisker = 0.721274, lower quartile = 0.723054, median = 0.725273, upper quartile = 0.727901, upper whisker = 0.734792},]coordinates{}; 
\addplot+ [color = blue!50!black,solid,boxplot prepared = {box extend=0.3, draw position = 4, lower whisker = 0.752463, lower quartile = 0.758091, median = 0.760399, upper quartile = 0.762618, upper whisker = 0.765219},]coordinates{}; 
\addplot+ [color = blue!50!black,solid,boxplot prepared = {box extend=0.3, draw position = 5, lower whisker = 0.767886, lower quartile = 0.769114, median = 0.769849, upper quartile = 0.772662, upper whisker = 0.777153},]coordinates{}; 
\addplot+ [color = blue!50!black,solid,boxplot prepared = {box extend=0.3, draw position = 6, upper whisker = 0.786549, upper quartile = 0.781019, median = 0.778677, lower quartile = 0.776975, lower whisker = 0.773675},]coordinates{}; 

\end{axis}
\end{tikzpicture}
}
\end{minipage}
\begin{minipage}{0.23\textwidth}
\resizebox {\textwidth} {!} {
\begin{tikzpicture}
\begin{axis}[
		xlabel= {$\alpha$},
		ylabel= {local segregation},
		boxplot/draw direction=y,
		baseline,
		xtick = {1, 2, 3, ..., 11},
		xticklabels  = {5, 30, 55, 80, 105, 130, 155, 180, 205, 230, 255},
		ymin=0.5,
    		ymax=1
	]

\addplot+ [color = lime!70!black,solid,boxplot prepared = {box extend=0.4,  upper whisker = 0.569317, upper quartile = 0.566948, median = 0.565704, lower quartile = 0.563991, lower whisker = 0.561343},]coordinates{};

\addplot+ [color = lime!70!black,solid,boxplot prepared = {box extend=0.4,  upper whisker = 0.815487, upper quartile = 0.810792, median = 0.807471, lower quartile = 0.804435, lower whisker = 0.798459},]coordinates{}; 

\addplot+ [color = lime!70!black,solid,boxplot prepared = {box extend=0.4,  upper whisker =0.863297, upper quartile = 0.859855, median = 0.855878, lower quartile = 0.852971, lower whisker = 0.848393},]coordinates{}; 

\addplot+ [color = lime!70!black,solid,boxplot prepared = {box extend=0.4,  upper whisker =0.844602, upper quartile = 0.843130, median =0.839602, lower quartile = 0.832798, lower whisker = 0.810987},]coordinates{};

\addplot+ [color = lime!70!black,solid,boxplot prepared = {box extend=0.4,  upper whisker =0.814326, upper quartile = 0.807044, median = 0.799180, lower quartile = 0.795927, lower whisker = 0.782402},]coordinates{};

\addplot+ [color = lime!70!black,solid,boxplot prepared = {box extend=0.4,  upper whisker =0.782966, upper quartile = 0.773236, median = 0.767630, lower quartile = 0.760762, lower whisker = 0.751039},]coordinates{};

\addplot+ [color = lime!70!black,solid,boxplot prepared = {box extend=0.4,  upper whisker =0.762613, upper quartile = 0.741276, median = 0.733168, lower quartile = 0.723363, lower whisker = 0.714112},]coordinates{};

\addplot+ [color = lime!70!black,solid,boxplot prepared = {box extend=0.4,  upper whisker =0.732723, upper quartile = 0.719416, median = 0.707596, lower quartile = 0.700607, lower whisker = 0.686914},]coordinates{};

\addplot+ [color = lime!70!black,solid,boxplot prepared = {box extend=0.4, upper whisker =0.705276, upper quartile = 0.697771, median = 0.684921, lower quartile = 0.678178, lower whisker = 0.652036},]coordinates{};

\addplot+ [color = lime!70!black,solid,boxplot prepared = {box extend=0.4,  upper whisker =0.692432, upper quartile = 0.680535, median = 0.672003, lower quartile = 0.656411, lower whisker = 0.637509},]coordinates{};

\addplot+ [color = lime!70!black,solid,boxplot prepared = {box extend=0.4,  upper whisker =0.668041, upper quartile = 0.651777, median = 0.645223, lower quartile = 0.638358, lower whisker = 0.615563},]coordinates{};

\addplot+ [color = yellow!70!black,solid,boxplot prepared = {box extend=0.4, draw position = 1, lower whisker = 0.571288, lower quartile = 0.574302, median = 0.575457, upper quartile = 0.576574, upper whisker = 0.577170},]coordinates{}; 
\addplot+ [color = yellow!70!black,solid,boxplot prepared = {box extend=0.4, draw position = 2, lower whisker = 0.864458, lower quartile = 0.870431, median = 0.873253, upper quartile = 0.875989, upper whisker = 0.882930},]coordinates{}; 
\addplot+ [color = yellow!70!black,solid,boxplot prepared = {box extend=0.4, draw position = 3, lower whisker = 0.930440, lower quartile = 0.933701, median = 0.935299, upper quartile = 0.938208, upper whisker = 0.939951},]coordinates{}; 
\addplot+ [color = yellow!70!black,solid,boxplot prepared = {box extend=0.4, draw position = 4, lower whisker = 0.950209, lower quartile = 0.953446, median = 0.955660, upper quartile = 0.956451, upper whisker = 0.961252},]coordinates{}; 
\addplot+ [color = yellow!70!black,solid,boxplot prepared = {box extend=0.4, draw position = 5, lower whisker = 0.959553, lower quartile = 0.963507, median = 0.964870, upper quartile = 0.966410, upper whisker = 0.969313},]coordinates{}; 
\addplot+ [color = yellow!70!black,solid,boxplot prepared = {box extend=0.4, draw position = 6, lower whisker = 0.970464, lower quartile = 0.971915, median = 0.972507, upper quartile = 0.973634, upper whisker = 0.974675},]coordinates{}; 
\addplot+ [color = yellow!70!black,solid,boxplot prepared = {box extend=0.4, draw position = 7, lower whisker = 0.973479, lower quartile = 0.975701, median = 0.976698, upper quartile = 0.977827, upper whisker = 0.980390},]coordinates{}; 
\addplot+ [color = yellow!70!black,solid,boxplot prepared = {box extend=0.4, draw position = 8, lower whisker = 0.974726, lower quartile = 0.979194, median = 0.980581, upper quartile = 0.981786, upper whisker = 0.983989},]coordinates{}; 
\addplot+ [color = yellow!70!black,solid,boxplot prepared = {box extend=0.4, draw position = 9, lower whisker = 0.980286, lower quartile = 0.981267, median = 0.982301, upper quartile = 0.983633, upper whisker = 0.985189},]coordinates{}; 
\addplot+ [color = yellow!70!black,solid,boxplot prepared = {box extend=0.4, draw position = 10, lower whisker = 0.982389, lower quartile = 0.984324, median = 0.984832, upper quartile = 0.985335, upper whisker = 0.986870},]coordinates{}; 
\addplot+ [color = yellow!70!black,solid,boxplot prepared = {box extend=0.4, draw position = 11, lower whisker = 0.983601, lower quartile = 0.985197, median = 0.985953, upper quartile = 0.986788, upper whisker = 0.988917},]coordinates{}; 

\addplot+ [color = blue!50!black,solid, boxplot prepared = {box extend=0.4, draw position = 1, upper whisker = 0.558332, upper quartile = 0.556623, median = 0.555612, lower quartile = 0.553640, lower whisker = 0.551390},]coordinates{}; 
\addplot+ [color = blue!50!black,solid,boxplot prepared = {box extend=0.4, draw position = 2, upper whisker = 0.786549, upper quartile = 0.781019, median = 0.778677, lower quartile = 0.776975, lower whisker = 0.773675},]coordinates{}; 
\addplot+ [color = blue!50!black,solid,boxplot prepared = {box extend=0.4, draw position = 3, upper whisker = 0.761652, upper quartile = 0.747456, median = 0.744271, lower quartile = 0.740453, lower whisker = 0.723683},]coordinates{}; 
\addplot+ [color = blue!50!black,solid,boxplot prepared = {box extend=0.4, draw position = 4, upper whisker = 0.724669, upper quartile = 0.716603, median = 0.712547, lower quartile = 0.708569, lower whisker = 0.696026},]coordinates{}; 
\addplot+ [color = blue!50!black,solid,boxplot prepared = {box extend=0.4, draw position = 5, upper whisker = 0.695507, upper quartile = 0.677721, median = 0.670677, lower quartile = 0.666027, lower whisker = 0.652375},]coordinates{}; 
\addplot+ [color = blue!50!black,solid,boxplot prepared = {box extend=0.4, draw position = 6, upper whisker = 0.658490, upper quartile = 0.653778, median = 0.643803, lower quartile = 0.637674, lower whisker = 0.631237},]coordinates{}; 
\addplot+ [color = blue!50!black,solid,boxplot prepared = {box extend=0.4, draw position = 7, upper whisker = 0.646484, upper quartile = 0.636930, median = 0.628301, lower quartile = 0.623119, lower whisker = 0.603993},]coordinates{}; 
\addplot+ [color = blue!50!black,solid,boxplot prepared = {box extend=0.4, draw position = 8, upper whisker = 0.622749, upper quartile = 0.612773, median = 0.606101, lower quartile = 0.597162, lower whisker = 0.586179},]coordinates{}; 
\addplot+ [color = blue!50!black,solid,boxplot prepared = {box extend=0.4, draw position = 9, upper whisker = 0.610267, upper quartile = 0.605699, median = 0.592222, lower quartile = 0.587106, lower whisker = 0.580332},]coordinates{}; 
\addplot+ [color = blue!50!black,solid,boxplot prepared = {box extend=0.4, draw position = 10, upper whisker = 0.597261, upper quartile = 0.591855, median = 0.582898, lower quartile = 0.579444, lower whisker = 0.568106},]coordinates{}; 
\addplot+ [color = blue!50!black,solid,boxplot prepared = {box extend=0.4, draw position = 11, upper whisker = 0.587987, upper quartile = 0.581270, median = 0.576112, lower quartile = 0.573518, lower whisker = 0.556067},]coordinates{};

\addplot+ [color = red!70!black,solid,boxplot prepared = {box extend=0.4, draw position = 1, lower whisker = 0.573521, lower quartile = 0.575512, median = 0.576715, upper quartile = 0.577702, upper whisker = 0.579620},]coordinates{}; 
\addplot+ [color = red!70!black,solid,boxplot prepared = {box extend=0.4, draw position = 2, lower whisker = 0.859387, lower quartile = 0.870538, median = 0.872640, upper quartile = 0.874923, upper whisker = 0.879567},]coordinates{}; 
\addplot+ [color = red!70!black,solid,boxplot prepared = {box extend=0.4, draw position = 3, lower whisker = 0.931284, lower quartile = 0.933401, median = 0.935217, upper quartile = 0.936746, upper whisker = 0.938053},]coordinates{}; 
\addplot+ [color = red!70!black,solid,boxplot prepared = {box extend=0.4, draw position = 4, lower whisker = 0.951478, lower quartile = 0.953341, median = 0.954085, upper quartile = 0.954621, upper whisker = 0.955221},]coordinates{}; 
\addplot+ [color = red!70!black,solid,boxplot prepared = {box extend=0.4, draw position = 5, lower whisker = 0.959751, lower quartile = 0.962712, median = 0.963661, upper quartile = 0.963842, upper whisker = 0.966639},]coordinates{}; 
\addplot+ [color = red!70!black,solid,boxplot prepared = {box extend=0.4, draw position = 6, lower whisker = 0.967118, lower quartile = 0.969189, median = 0.969529, upper quartile = 0.970626, upper whisker = 0.971680},]coordinates{}; 
\addplot+ [color = red!70!black,solid,boxplot prepared = {box extend=0.4, draw position = 7, lower whisker = 0.970371, lower quartile = 0.972813, median = 0.973415, upper quartile = 0.974614, upper whisker = 0.976680},]coordinates{}; 
\addplot+ [color = red!70!black,solid,boxplot prepared = {box extend=0.4, draw position = 8, lower whisker = 0.975417, lower quartile = 0.975686, median = 0.976479, upper quartile = 0.977643, upper whisker = 0.978513},]coordinates{}; 
\addplot+ [color = red!70!black,solid,boxplot prepared = {box extend=0.4, draw position = 9, lower whisker = 0.976438, lower quartile = 0.977982, median = 0.978677, upper quartile = 0.979150, upper whisker = 0.980169},]coordinates{}; 
\addplot+ [color = red!70!black,solid,boxplot prepared = {box extend=0.4, draw position = 10, lower whisker = 0.977814, lower quartile = 0.979731, median = 0.980077, upper quartile = 0.981225, upper whisker = 0.981783},]coordinates{}; 
\addplot+ [color = red!70!black,solid,boxplot prepared = {box extend=0.4, draw position = 11, lower whisker = 0.980028, lower quartile = 0.981022, median = 0.981347, upper quartile = 0.981774, upper whisker = 0.983726},]coordinates{}; 

\end{axis}
\end{tikzpicture}
}
\end{minipage}
\\

\ref{named}
\caption{Local segregation of pairwise stable networks in the \deigame obtained by iterative best add-only moves for $n=1000$ over 50 runs starting from a random or segregated tree and grid.}
\label{plot:DEINCG_AddOnly}
\end{figure}

\section{Conclusion}

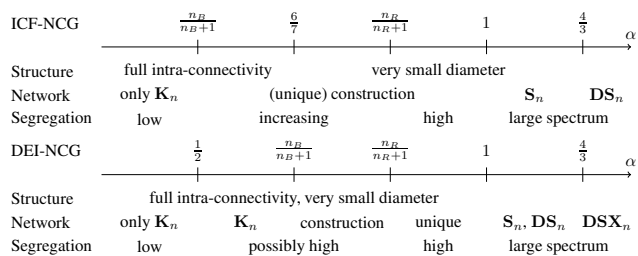
\begin{figure}
	\centering
\resizebox{\columnwidth}{!}{
	\begin{tikzpicture}
		\draw[->] (0,0) -- (11,0);
		\node (c0) at (0,.5) {};
		\node (c1) at (2,.5) {};
		\node (c2) at (4,.5) {};
		\node (c3) at (6,.5) {};
		\node (c4) at (8,.5) {};
		\node (c5) at (10,.5) {};
		\node (c6) at (11,.5) {};
		\foreach \i in {0,1,2,3,4,5,6}{
	    \coordinate (s\i) at ($(c\i) + (0,-1)$);
	    \coordinate (n\i) at ($(c\i) + (0,-1.5)$);
	    \coordinate (se\i) at ($(c\i) + (0,-2)$);}

		\foreach \i in {1,2,3,4,5}
		\draw ($(c\i) + (0,-.35)$) -- ($(c\i) + (0,-.65)$) ;

		\node at (c1) {$\frac {\numb{B}}{\numb{B}+1}$};
		\node at (c2) {$\frac 67$};
		\node at (c3) {$\frac {\numb{R}}{\numb{R}+1}$};
		\node at (c4) {$1$};
		\node at (c5) {$\frac 43$};
		
		\node at ($(c6) + (0,-.25)$) {$\alpha$};
		
		\node[anchor = west] at (-2,.5) {\icfgame};
		\node[anchor = west] at (-2,-.5) {Structure};
		\node[anchor = west] at (-2,-1) {Network};
		\node[anchor = west] at (-2,-1.5) {Segregation};
		
		\node at (barycentric cs:s0=0.5,s2=0.5) {full intra-connectivity};
		\node at (barycentric cs:s2=0.5,s5=0.5) {very small diameter};
		
		\node at (barycentric cs:n0=0.5,n1=0.5) {only $\com_n$};
		\node at (barycentric cs:n1=0.5,n4=0.5) {(unique) construction};
		\node at (barycentric cs:n4=0.5,n5=0.5) {$\sta_n$};
		\node at (barycentric cs:n6=0.5,n5=0.5) {$\dsta_n$};
		
		\node at (barycentric cs:se0=0.5,se1=0.5) {low};
		\node at (barycentric cs:se3=0.5,se1=0.5) {increasing};
		\node at (barycentric cs:se3=0.5,se4=0.5) {high};
		\node at (barycentric cs:se4=0.5,se6=0.5) {large spectrum};
	\end{tikzpicture}
} 
\resizebox{\columnwidth}{!}{
	\begin{tikzpicture}
		\draw[->] (0,0) -- (11,0);
		\node (c0) at (0,.5) {};
		\node (c1) at (2,.5) {};
		\node (c2) at (4,.5) {};
		\node (c3) at (6,.5) {};
		\node (c4) at (8,.5) {};
		\node (c5) at (10,.5) {};
		\node (c6) at (11,.5) {};
		\foreach \i in {0,1,2,3,4,5,6}{
	    \coordinate (s\i) at ($(c\i) + (0,-1)$);
	    \coordinate (n\i) at ($(c\i) + (0,-1.5)$);
	    \coordinate (se\i) at ($(c\i) + (0,-2)$);}

		\foreach \i in {1,2,3,4,5}
		\draw ($(c\i) + (0,-.35)$) -- ($(c\i) + (0,-.65)$) ;

		\node at (c1) {$\frac 12$};
		\node at (c2) {$\frac {\numb{B}}{\numb{B}+1}$};
		\node at (c3) {$\frac {\numb{R}}{\numb{R}+1}$};
		\node at (c4) {$1$};
		\node at (c5) {$\frac 43$};
		
		\node at ($(c6) + (0,-.25)$) {$\alpha$};
		
		\node[anchor = west] at (-2,.5) {\deigame};
		\node[anchor = west] at (-2,-.5) {Structure};
		\node[anchor = west] at (-2,-1) {Network};
		\node[anchor = west] at (-2,-1.5) {Segregation};
		
		\node at (barycentric cs:s0=0.5,s4=0.5) {full intra-connectivity, very small diameter};
		
		\node at (barycentric cs:n0=0.5,n1=0.5) {only $\com_n$};
		\node at (barycentric cs:n1=0.5,n2=0.5) {$\com_n$};
		\node at (barycentric cs:n2=0.5,n3=0.5) {construction};
		\node at (barycentric cs:n3=0.5,n4=0.5) {unique};
		\node at (barycentric cs:n4=0.5,n5=0.5) {$\sta_n$, $\dsta_n$};
		\node at (barycentric cs:n6=0.5,n5=0.5) {$\dss_n$};
		
		\node at (barycentric cs:se0=0.5,se1=0.5) {low};
		\node at (barycentric cs:se3=0.5,se1=0.5) {possibly high};
		\node at (barycentric cs:se3=0.5,se4=0.5) {high};
		\node at (barycentric cs:se4=0.5,se6=0.5) {large spectrum};
	\end{tikzpicture}
} 
	\caption{Overview of our theoretical results. We display structural properties of pairwise stable networks, explicit pairwise stable networks and findings about the segregation of pairwise stable networks. The two models behave surprisingly similar.}\label{fig:results}
\end{figure}

We have investigated two variants of network creation games that consider heterogeneous edge creation of agents acting according to homophily. Our main goal was to analyze segregation within reasonable networks measured by pairwise stability. Our results are summarized in \Cref{fig:results}. Even though our two game models feature two seemingly orthogonal perspectives based on a direct and an indirect consideration of homophily, their qualitative behavior is surprisingly similar. 

Clearly, stable networks are highly integrated for a very small edge cost, when agents can afford to buy all available edges. Once our cost parameter reaches the sweet spot where agents need to balance neighborhood and distance cost, there is provably high segregation, following from characterizations of stable networks. For slightly larger edge cost, our theoretical results cannot give a clear tendency of the segregation strength. In principle, both low and high segregation can be achieved by stable networks. 
Therefore, we performed an average-case analysis by running extensive simulation experiments. These experiments provide general tendencies about segregation contrasting the large theoretical spectrum for $\alpha\ge 1$.
Most importantly, we observe low segregation under integrated initial conditions if edges cannot be deleted. This yields an escape route from segregation by a high initial investment establishing permanent integration.

\section*{Acknowledgements}
Funded by the Deutsche Forschungsgemeinschaft (German
Research Foundation): BR 2312/12-1 and LE 3445/3-1.

\bibliographystyle{named}
\bibliography{beyondschelling}

\appendix

\section{Missing Proofs}

In this appendix, we provide missing proofs.

\subsection{Increasing Comfort among Friends}
For the analysis of pairwise stability, we frequently have to compute an agent's cost change after creating or severing one edge. To clarify the calculations, we gather the respective formulae in a technical lemma. 

\begin{restatable}{lemma}{OneStep}\label{lem:one_step_benefit}
	Consider a network $G=(\ag{},E)$ and an agent $u\in \ag{}$ in the \icfgame. Consider an agent $v\in \ag{\mathcal T(u)}$ of the same type and an agent $w\in \ag{}\setminus \ag{\mathcal T(u)}$ of a different type. Then, the following statements hold:
	\begin{enumerate}
		\item $\Ncost{G+uv}{u} - \Ncost{G}{u} = \alpha\left(1+\frac{\fr{G}{u}-\deg{G}{u}+1}{(\fr{G}{u}+1)(\fr{G}{u}+2)}\right)$ if $uv\notin E$ (creation of a monochromatic edge),
		\item $\Ncost{G-uv}{u} - \Ncost{G}{u} = -\alpha\left(1+\frac{\fr{G}{u}-\deg{G}{u}+1}{(\fr{G}{u}+1)\fr{G}{u}}\right)$ if $uv\in E$ (deletion of a monochromatic edge),
		\item $\Ncost{G+uw}{u} - \Ncost{G}{u} = \alpha\left(1+\frac{1}{\fr{G}{u}+1}\right)$ if $uw\notin E$ (creation of a bichromatic edge), and
		\item $\Ncost{G-uw}{u} - \Ncost{G}{u} = -\alpha\left(1+\frac{1}{\fr{G}{u}+1}\right)$ if $uw\in E$ (deletion of a bichromatic edge).
	\end{enumerate}
\end{restatable}

\begin{proof}
	We perform the calculations for each case accordingly. Let $G'$ be the network after the respective edge creation or deletion.
	\begin{enumerate}
		\item Creation of a monochromatic edge: $\Ncost{G'}{u}-\Ncost{G}{u} = (\deg{G}{u}+1)\cdot\alpha\left(1+\frac{1}{\fr{G}{u}+2}\right) - \deg{G}{u}\cdot\alpha\left(1+\frac{1}{\fr{G}{u}+1}\right) = \alpha\left(1+\frac{\fr{G}{u}-\deg{G}{u}+1}{(\fr{G}{u}+1)(\fr{G}{u}+2)}\right)$.
		\item Deletion of a monochromatic edge: $\Ncost{G'}{u}-\Ncost{G}{u} = (\deg{G}{u}-1)\cdot\alpha\left(1+\frac{1}{\fr{G}{u}}\right) - \deg{G}{u}\cdot\alpha\left(1+\frac{1}{\fr{G}{u}+1}\right) = -\alpha\left(1+\frac{\fr{G}{u}-\deg{G}{u}+1}{(\fr{G}{u}+1)\fr{G}{u}}\right)$.
		\item Creation of a bichromatic edge: $\Ncost{G'}{u}-\Ncost{G}{u} = (\deg{G}{u}+1)\cdot\alpha\left(1+\frac{1}{\fr{G}{u}+1}\right) - \deg{G}{u}\cdot\alpha\left(1+\frac{1}{\fr{G}{u}+1}\right) = \alpha\left(1+\frac{1}{\fr{G}{u}+1}\right)$.
		\item Deletion of a bichromatic edge: $\Ncost{G'}{u}-\Ncost{G}{u} = (\deg{G}{u}-1)\cdot\alpha\left(1+\frac{1}{\fr{G}{u}+1}\right) - \deg{G}{u}\cdot\alpha\left(1+\frac{1}{\fr{G}{u}+1}\right) = -\alpha\left(1+\frac{1}{\fr{G}{u}+1}\right)$. \qedhere
	\end{enumerate}
	
\end{proof}

Next, we provide proofs for the collected statements about \icfgame{s} concerning structural properties of pairwise stable networks and simple pairwise stable networks.

\ICFproperties*
\begin{proof}
	We prove the statements one after another.
	
	\begin{enumerate}
		\item 
		Let $\alpha < \frac 67$. Assume that a network $G = (\ag{},E)$ is given that is not fully intra-connected. Let $u, v\in \ag{}$ be agents of the same type with $uv\notin E$. Define $G' = G + uv$. We will show that $\Tcost{G'}{u} - \Tcost{G}{u} < 0$ (the computation for $v$ is identical). We can assume that $\deg{G}{u}\ge 1$, because otherwise agent $u$'s cost would be infinite and adding $uv$ would be beneficial. We compute the difference in the neighborhood cost, using \Cref{lem:one_step_benefit} in the first equality.
		
		\begin{align*}
		&\Ncost{G'}{u} - \Ncost{G}{u} = \alpha\left(1+\frac{\fr{G}{u}-\deg{G}{u}+1}{(\fr{G}{u}+1)(\fr{G}{u}+2)}\right)\\
		&= \alpha \left(\frac {\fr{G}{u}+3}{\fr{G}{u}+2} - \deg{G}{u}\frac 1 {(\fr{G}{u}+2)(\fr{G}{u}+1)}\right) \\
		&\le \alpha \left(\frac {\fr{G}{u}+3}{\fr{G}{u}+2} - \frac 1 {(\fr{G}{u}+2)(\fr{G}{u}+1)}\right)\text.
		\end{align*}

		Now, consider the function $f\colon \mathbb R_{\ge0}\to \mathbb R, f(x) = \frac {x+3}{x+2} - \frac 1 {(x+2)(x+1)}$. This function attains its maximum for $x = \sqrt 2$ and is monotonically increasing for $0\le x \le \sqrt 2$ and monotonically decreasing for $x\ge \sqrt 2$. Moreover, $f(1) = f(2) = \frac 76$. Hence, the maximum attained by integer values is $\frac 76$. We conclude that $\Ncost{G'}{u} - \Ncost{G}{u} \le \frac 76 \alpha < 1$. Since $\Dcost{G'}{u} - \Dcost{G}{u} \le -1$, we obtain $\Tcost{G'}{u} - \Tcost{G}{u} < 0$. Hence, creation of the edge $uv$ is beneficial for $u$.
		\item
		Let $\alpha < \frac43$ and consider a pairwise stable network $G$. In particular, $G$ is connected. Assume that there are agents $v$ and $w$ of distance at least $3$. We will show that $G' = G + vw$ is better for both of these agents, contradicting the pairwise stability of $G$.
		
		The same computations as in the proof of the first property show that the neighborhood cost increases by at most $\frac 76 \alpha$ if $vw$ is monochromatic. On the other hand, if $vw$ is bichromatic, then the neighborhood cost increases by at most $\frac 32 \alpha$. Since the distance cost decreases by at least $2$, we conclude that $\Tcost{G'}{x} - \Tcost{G}{x} < 0$ for $\alpha < \frac43$ and $x\in \{v,w\}$.
		
		The curiosity of one agent type follows from the fact that two agents from different types, which are both not curious, must have distance at least $3$.
		\item 
		Let $\alpha < 1$ and assume that $v$ is a curious agent of a network $G = (\ag{},E)$. Consider an agent $w$ of the same type such that $vw\notin E$.
		Then,
		\begin{align*}
		&\Ncost{G'}{u} - \Ncost{G}{u}\\ & = \alpha \left(\frac {\fr{G}{u}+3}{\fr{G}{u}+2} - \frac {\deg{G}{u}} {(\fr{G}{u}+2)(\fr{G}{u}+1)}\right)\\
		& \le \alpha \left(\frac {\fr{G}{u}+3}{\fr{G}{u}+2} - \frac {\fr{G}{u}+1} {(\fr{G}{u}+2)(\fr{G}{u}+1)}\right) = \alpha < 1\text.
		\end{align*}
		
		The first equality is derived by the same computations as in the proof of the first property. Consequently, $\Tcost{G'}{u} - \Tcost{G}{u} < 0$. Hence, if $v$ and $w$ are both curious agents of the same type, then the edge $vw$ must be present in any pairwise stable network. 
		\item
		We start to show that $\com_n$ is pairwise stable for $\alpha\leq\frac{\numb{B}}{\numb{B}+1}$.
		
		To this end, we show that no edge can be deleted by one of its endpoints.
		Consider a pair of agents $u, v\in \ag{}$. If they are of the same type, then severing the edge $uv$ by $u$ decreases her cost by 
		\begin{align*}
		&\Tcost{G-uv}{u} - \Tcost{G}{u} = - \alpha\left(1+\frac{\fr{G}{u}-\deg{G}{u}+1}{(\fr{G}{u}+1)\fr{G}{u}}\right)\\
		&= -\alpha\left(1+\frac{\fr{G}{u}+2-n}{\fr{G}{u}(\fr{G}{u}+1)}\right)+1\\
		&\geq  -\frac{{\numb{B}}}{{\numb{B}}+1}\cdot\frac{n^2-n+1}{(n-1)n}+1\\
		&\geq -\frac{n}{n+1}\cdot\frac{n^2-n+1}{(n-1)n}+1\geq 0.
		\end{align*}
		Hence, no agent can improve her strategy by severing an edge to an agent of the same color.
		
		If $u$ and $v$ have different colors, the cost decrease is
		\begin{align*}
		&\Tcost{G-uv}{u} - \Tcost{G}{u} = -\alpha\left(1+\frac{1}{\fr{G}{u}+1}\right)+1\\&\geq -\alpha\left(1+\frac{1}{\numb{B}}\right)+1\geq - \frac{\numb{B}}{\numb{B}+1}\cdot \frac{\numb{B}+1}{\numb{B}}+1=0.
		\end{align*}
		Therefore, there is no improving move for any agent in the network, which implies that $\com_n$ is pairwise stable.
		
		For the uniqueness, consider any pairwise stable network $G = (\ag{},E)$ and assume that $\alpha < \min\{\frac{6}{7}, \frac{\numb{B}}{\numb{B}+1}\}$. Note that $G$ is fully intra-connected according to \Cref{prop:ICF-properties}(\ref{prop:fullintraconn}). Assume for contradiction that there are two agents $u, v\in \ag{}$ with $uv\notin E$ which have a different type.
		
		Then, creating the edge $uv$ increases the neighborhood cost for each involved agent by at most $\alpha\left(1+\frac{1}{\fr{G}{u}+1}\right)\leq \alpha\left(1+\frac{1}{\numb{B}}\right) < 1$, while it decreases the distance to at least one node, a contradiction. Hence, $uv\in E$, which implies that $G$ is a clique.
		
		\item
		Consider a star graph $\sta_n$ with central node $c$. 
		To show that $\sta_n$ is pairwise stable, we need to prove that no two leaves can jointly create an edge. 
		Consider two leafs $u$ and $v$. 
		There can be a few possible situations. The first two cases cover the case that $c$ and one of $u$ and $v$ are of the same color, say $u\in \ag{\mathcal T(c)}$. If $v\in \ag{\mathcal T(c)}$, then creating $uv$ causes an increase in neighborhood cost of $\Ncost{\sta_n+uv}{u} - \Ncost{\sta_n}{u} = \alpha\left(1+\frac{1}{6}\right)=\frac{7}{6}\alpha$, while the distance cost is only decreased by $1$. Hence, for $\alpha\geq 1$, creating the edge $uv$ is not beneficial for $u$. If $v$ has a different color, then $\Ncost{\sta_n+uv}{u} - \Ncost{\sta_n}{u} = \frac{3}{2}\alpha$, and $u$ would again prevent the creation of $uv$.
		
		It remains that $u$ and $v$ both have a different color from $c$. If $v\in \ag{\mathcal T(u)}$, then creating the edge $uv$ increases the neighborhood cost by $\alpha$ and decreases the distance cost by 1 for both $u$ and $v$. Thus, since $\alpha\geq 1$, this is not beneficial.
		
		If all three nodes $u, v,$ and $c$ have different colors, then the creation of the edge $uv$ increases the neighborhood cost of $u$ by $2\alpha\ge 2$ and decreases her distance cost by only $1$.
		
		Therefore, no pair of nodes can create an edge to improve their cost. Clearly, also no edge can be unilaterally deleted. The assertion follows. \qedhere
	\end{enumerate}
\end{proof}

The proof of existence of pairwise stable networks for multiple types and an intermediate range of $\alpha$ has a similar structure as the special case of two types. In particular, the structure of the subnetwork induced by the agents in $\ag{B}\cup \ag{T}$ for any type $T\in \mathcal T$ with $T\neq B$ is essentially the same. However, dependent on $\alpha$, agents from larger communities might have an incentive to maintain further bichromatic edges. For completeness and to get acquainted with the relevant networks, we also provide the complete proof for two agent types, for which we only provided the construction in the body of the paper.

\ExPSNintermed*

\begin{proof}[Complete proof for two agent types]
	Consider an instance of the \icfgame and let $\frac{\numb{B}}{\numb{B}+1} \le \alpha < 1$. We will define a stable network for $\alpha$ dependent on the threshold $\tau = \frac {\numb{B}(\numb{B} + 1)}{\numb{B}(\numb{B}+1)+1}$. Note that $\frac{\numb{B}}{\numb{B}+1} < \tau <1$, as $\numb{B}(\numb{B} + 1) > \numb{B}$.
	
	We assume $\ag{B} = \{b_1,\dots, b_{\numb{B}}\}$ and $\ag{R} = \{r_1,\dots, r_{\numb{R}}\}$ and define the edge set of the graph $G = (\ag{},E)$ as follows:
	
	\begin{itemize}
		\item $\{x_i,x_j\}\in E$, for $x\in \{b,r\}, i,j\in \{1,\dots, \numb{B}\}$,
		\item $\{r_i,b_i\}\in E$, for $i\in \{1,\dots, \numb{B}\}$,
		\item $\{r_i, r_j\}\in E$, for $i\in \{1,\dots, \numb{B}\}$ and $j\in \{\numb{B}+1,\dots, \numb{R}\}$,
		\item if $\alpha < \tau$, then $\{r_i, r_j\}\in E$, for $i,j\in \{\numb{B}+1,\dots, \numb{R}\}$, and no further edges are in $E$,
		\item otherwise, no further edges are in $E$.
	\end{itemize} 
	The network $G$ is illustrated in \Cref{fig:PSN-exist-intermed2}.
	\begin{figure}[h]
		\centering
		\begin{tikzpicture}[scale = .83]
		\node[protovertex] (b2) at (0.5,0) {};
		\node[protovertex] (b1) at (108:1) {};
		\node[protovertex] (b3) at (252:1) {};
		
		\node (c) at (3,0){};
		
		\node[protodiam] (r4) at ($(c) + (36:1)$) {};
		\node[protodiam] (r1) at ($(c) + (108:1)$) {};
		\node[protodiam] (r2) at ($(c) + (180:1)$) {};
		\node[protodiam] (r3) at ($(c) + (252:1)$) {};
		\node[protodiam] (r5) at ($(c) + (324:1)$) {};
		
		\foreach \i in {1,2,3}
		{\draw (b\i) edge (r\i);
			\draw (r4) edge (r\i);
			\draw (r5) edge (r\i);
		}
		\foreach \x/\y in {1/2,2/3,3/1}
		{\draw  (b\x) edge (b\y);
			\draw  (r\x) edge (r\y);}
		\draw (r4) edge (r5);
		\end{tikzpicture}
		\qquad
		\begin{tikzpicture}[scale = .83]
		\node[protovertex] (b2) at (0.5,0) {};
		\node[protovertex] (b1) at (108:1) {};
		\node[protovertex] (b3) at (252:1) {};
		
		\node (c) at (3,0){};
		
		\node[protodiam] (r4) at ($(c) + (36:1)$) {};
		\node[protodiam] (r1) at ($(c) + (108:1)$) {};
		\node[protodiam] (r2) at ($(c) + (180:1)$) {};
		\node[protodiam] (r3) at ($(c) + (252:1)$) {};
		\node[protodiam] (r5) at ($(c) + (324:1)$) {};
		
		\foreach \i in {1,2,3}
		{\draw (b\i) edge (r\i);
			\draw (r4) edge (r\i);
			\draw (r5) edge (r\i);
		}
		\foreach \x/\y in {1/2,2/3,3/1}
		{\draw  (b\x) edge (b\y);
			\draw  (r\x) edge (r\y);}
		\end{tikzpicture}
		\caption{Pairwise stable networks for $\frac{\numb{B}}{\numb{B}+1} \le \alpha < \tau$ (left) and $\tau \le \alpha < 1$ (right).}
		\label{fig:PSN-exist-intermed2}
	\end{figure}
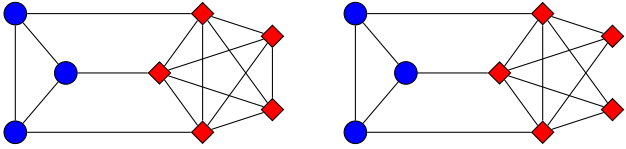
	
	\noindent We claim that $G$ is pairwise stable. First, we show that no agent can sever an edge. Let $i\in \{1,\dots, \numb{B}\}$ and $j,k\in \{\numb{B}+1,\dots, \numb{R}\}$.
	
	\begin{itemize}
		\item If agent $b_i$ severs an edge to an agent of her type, the distance cost is increased by $1$ while the neighborhood cost is decreased by $\alpha < 1$ (which can be computed using \Cref{lem:one_step_benefit}). If a connection to a red agent is severed, then the distance to this neighbor increases by $2$ while the neighborhood cost is decreased by $\alpha \left(1 + \frac 1 {\numb{B}}\right)< 2\alpha <2$.
		\item The same considerations show that agents $r_i$ cannot sever edges to red and blue agents, respectively.
		\item The red agent $r_j$ cannot sever the edge towards agent $r_i$, because this improves the neighborhood cost by less than $2$ while it increases the distance to both $r_i$ and $b_i$ by $1$ each.
		\item Finally, consider the case that $\alpha < \tau$. Then, $r_j$ cannot sever $r_jr_k$ for $k\neq j$. Indeed, this would increase the distance cost by $1$ while saving a neighborhood cost of $\alpha \left(1+ \frac 1{\numb{R}(\numb{R}-1)}\right)\le \alpha \left(1+ \frac 1{(\numb{B}+1)\numb{B}}\right)$. Here, we use that such an edge can only exist if $\numb{R}\ge \numb{B}+1$. Hence, the total increase in cost is at least $1-\alpha \left(1+ \frac 1{(\numb{B}+1)\numb{B}}\right) = 1- \alpha \frac {\numb{B}(\numb{B}+1)+1}{(\numb{B}+1)\numb{B}} > 1- \tau \frac {\numb{B}(\numb{B}+1)+1}{(\numb{B}+1)\numb{B}} = 0$.
	\end{itemize}
	
	Next, we show that it is also not possible to add edges.
	
	\begin{itemize}
		\item Let $i\in \{1,\dots, \numb{B}\}$ and $j\in \{1,\dots, \numb{R}\}$  with $i\neq j$. Then, agent $b_i$ does not benefit from creating the edge $b_ir_j$. Indeed, this decreases her distance cost by exactly $1$ while it increases her neighborhood cost by $\alpha \frac {\numb{B}+1}{\numb{B}}\ge 1$, using the lower bound on $\alpha$.
		\item It remains the case of missing edges between red agents for large edge cost. Assume therefore $\alpha\ge \tau$ and let $i,j \in \{\numb{B}+1,\dots, \numb{R}\}$. Adding the edge $r_ir_j$ decreases the distance cost for $r_i$ by $1$ while increasing her neighborhood cost by $\alpha \left(1 + \frac 1 {\numb{B}(\numb{B}+1)}\right)\ge \tau \frac {\numb{B}(\numb{B}+1)+1}{(\numb{B}+1)\numb{B}} = 1$. Hence, creating this edge is not beneficial for $r_i$.
	\end{itemize}
	Thus, we have found stable networks for $\frac{\numb{B}}{\numb{B}+1} \le \alpha < 1$.
\end{proof}

Now, we give the proof for an arbitrary number of agent types.

\begin{proof}
	Consider an instance of \icfgame and let $\frac{\numb{B}}{\numb{B}+1} \le \alpha < 1$. Assume that we have ordered the types in increasing size, i.e., $\mathcal T = \{T_1,\dots, T_k\}$ where $T_1 = B$, $T_k = R$ and $\numb{T_1}\le \dots \le \numb{T_k}$. Suppose that $\ag{T_j} = \{t^1_j,\dots, t^{\numb{T_j}}_j\}$. We will define a stable network for $\alpha$ dependent on several thresholds for $\alpha$. In particular, there is a threshold $\tau = \frac {\numb{T_{k-1}}(\numb{T_{k-1}} + 1)}{\numb{T_{k-1}}(\numb{T_{k-1}}+1)+1}$, which plays a similar role as in the case of $2$ types. However, we have to consider further threshold values. Let therefore $2\le j \le k - 1$, and define $\tau_j = \frac{\numb{T_j}}{\numb{T_j} + 1}$.
	Note that $\frac{\numb{B}}{\numb{B}+1} \le \tau_2 \le \tau_3\le \dots \le  \tau_{k-1} < \tau < 1$ as $\numb{T_{k-1}}(\numb{T_{k-1}} + 1) > \numb{T_{k-1}}$.

	\begin{figure*}
		\centering
		\begin{tikzpicture}[scale = 0.8]
		\node[protovertex] (b1) at (108:1) {};
		\node[protovertex] (b2) at (252:1) {};
		\draw  (b1) edge (b2);
		
		\node (c) at (2,0){};

		\node[protodiam] (r4) at ($(c) + (96:1)$) {};
		\node[protodiam] (r1) at ($(c) + (168:1)$) {};
		\node[protodiam] (r2) at ($(c) + (240:1)$) {};
		\node[protodiam] (r3) at ($(c) + (312:1)$) {};
		\node[protodiam] (r5) at ($(c) + (14:1)$) {};
		
		\node (k) at (2,-4){};
		
		\node[protosq] (s2) at ($(k) + (90:1)$) {};
		\node[protosq] (s1) at ($(k) + (180:1.2)$) {};
		\node[protosq] (s3) at ($(k) + (0:1.2)$) {};

		\foreach \i in {1,2}
		{\draw (b\i) edge (r\i);
		}
		\draw (b1) edge (s2);
		\draw (b2) edge (s1);
		\foreach \i in {1,2,3}
		{
			\draw (r4) edge (r\i);
			\draw (r5) edge (r\i);
		}
		\foreach \x/\y in {1/2,1/3,2/3}
		{\draw  (s\x) edge (s\y);
			\draw  (r\x) edge (r\y);}
		\foreach \i in {1,2,3}{
		\foreach \j in {1,2,3,4,5}
		\draw (s\i) edge (r\j);}
		\draw (r4) edge (r5);
		\end{tikzpicture}
		\qquad
		\begin{tikzpicture}[scale = 0.8]
		\node[protovertex] (b1) at (108:1) {};
		\node[protovertex] (b2) at (252:1) {};
		\draw  (b1) edge (b2);
		
		\node (c) at (2,0){};

		\node[protodiam] (r4) at ($(c) + (96:1)$) {};
		\node[protodiam] (r1) at ($(c) + (168:1)$) {};
		\node[protodiam] (r2) at ($(c) + (240:1)$) {};
		\node[protodiam] (r3) at ($(c) + (312:1)$) {};
		\node[protodiam] (r5) at ($(c) + (14:1)$) {};
		
		\node (k) at (2,-4){};
		
		\node[protosq] (s2) at ($(k) + (90:1)$) {};
		\node[protosq] (s1) at ($(k) + (180:1.2)$) {};
		\node[protosq] (s3) at ($(k) + (0:1.2)$) {};

		\foreach \i in {1,2}
		{\draw (b\i) edge (r\i);
		}
		\draw (b1) edge (s2);
		\draw (b2) edge (s1);
		\foreach \i in {1,2,3}
		{\draw (s\i) edge (r\i);
			\draw (r4) edge (r\i);
			\draw (r5) edge (r\i);
		}
		\foreach \x/\y in {1/2,1/3,2/3}
		{\draw  (s\x) edge (s\y);
			\draw  (r\x) edge (r\y);}
		\draw (r4) edge (r5);
		\end{tikzpicture}
		\qquad
		\begin{tikzpicture}[scale = 0.8]
		\node[protovertex] (b1) at (108:1) {};
		\node[protovertex] (b2) at (252:1) {};
		\draw  (b1) edge (b2);
		
		\node (c) at (2,0){};

		\node[protodiam] (r4) at ($(c) + (96:1)$) {};
		\node[protodiam] (r1) at ($(c) + (168:1)$) {};
		\node[protodiam] (r2) at ($(c) + (240:1)$) {};
		\node[protodiam] (r3) at ($(c) + (312:1)$) {};
		\node[protodiam] (r5) at ($(c) + (14:1)$) {};
		
		\node (k) at (2,-4){};
		
		\node[protosq] (s2) at ($(k) + (90:1)$) {};
		\node[protosq] (s1) at ($(k) + (180:1.2)$) {};
		\node[protosq] (s3) at ($(k) + (0:1.2)$) {};

		\foreach \i in {1,2}
		{\draw (b\i) edge (r\i);
		}
		\draw (b1) edge (s2);
		\draw (b2) edge (s1);
		\foreach \i in {1,2,3}
		{\draw (s\i) edge (r\i);
			\draw (r4) edge (r\i);
			\draw (r5) edge (r\i);
		}
		\foreach \x/\y in {1/2,1/3,2/3}
		{\draw  (s\x) edge (s\y);
			\draw  (r\x) edge (r\y);}
		\end{tikzpicture}
		\caption{Illustration of the proof of \Cref{prop:exist-PSN-intermed}. We consider an \icfgame with $3$ types containing $3$, $4$, and $6$ agents, respectively. Hence, we consider the parameter range $\frac{\numb{B}}{\numb{B}+1} = \frac 23 \le \alpha < 1$. The pairwise stable networks are dependent on the thresholds $\tau_2 = \frac 34$ and $\tau = \frac {12}{13}$. We then find the pairwise stable networks for $\frac{\numb{B}}{\numb{B}+1} \le \alpha < \tau_2$ (left), $\tau_2\le \alpha < \tau$ (middle), and $\tau \le \alpha < 1$ (right).}
		\label{fig:PSN-exist-intermed-k}
	\end{figure*}
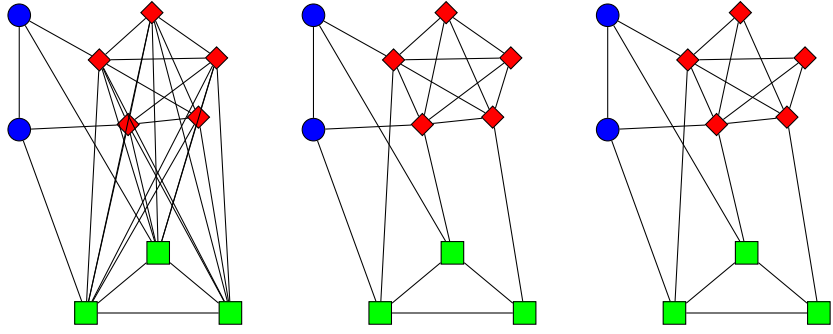

	We define the network $G = (\ag{},E)$ with edges given as follows:
	
	\begin{itemize}
		\item $\{t_j^i,t_j^l\}\in E$ for $1\le j\le k$, $1 \le i < l \le \min\{\numb{T_j},\numb{T_{k-1}}\}$,
		\item $\{t_j^i,t_l^i\}\in E$ for $1\le j < l \le k$, $1\le i\le \numb{T_j}$,
		\item $\{t_k^i, t_k^l\}\in E$ for $1\le i \le \numb{T_{k-1}}$ and $\numb{T_{k-1}}+1 \le l \le \numb{T_k}$,
		\item for each $2\le j \le k-1$, if $\alpha < \tau_j$, then $\{t_j^i,t_l^m\}\in E$ for $j < l \le k$, $1\le i\le \numb{T_j}$, and $1\le m\le \numb{T_l}$,
		\item if $\alpha < \tau$, then  $\{t_k^i, t_k^l\}\in E$ for $\numb{T_{k-1}}+1 \le i < l \le \numb{T_k}$, and
		\item no further edges are in $E$.
	\end{itemize} 
	
	The two cases for the network $G$ are illustrated in \Cref{fig:PSN-exist-intermed-k}.

	We claim that $G$ is pairwise stable. First, no agent can sever an edge. Let $1\le j \le k$, $1 \le i \le \numb{T_{k-1}}$, and $\numb{T_{k-1}}+1 \le l, m \le \numb{T_{k}}$.

	\begin{itemize}
		\item If agent $t_j^i$ severs an edge to an agent of her type, the distance cost is increased by $1$ while the neighborhood cost is decreased by $\alpha\left(1+\frac{\fr{G}{u}-\deg{G}{u}+1}{(\fr{G}{u}+1)\fr{G}{u}}\right) \le \alpha < 1$ (which can be computed with the aid of \Cref{lem:one_step_benefit}). 
		\item In the next two bullet points, we show that no agent can sever a bichromatic edge between an agent in $\ag{T_j}$ an agent of type $T_p$ for $j+1\le p \le k$. First, $t_1^i$ cannot sever a bichromatic edge, because then the distance to the adjacent neighbor increases by $2$ while the neighborhood cost is decreased by $\alpha \left(1 + \frac 1 {\numb{T_1}}\right)< 2\alpha <2$. For the same reason, the unique neighbor of $t_1^i$ in $\ag{T_p}$ for $2\le p \le k$ cannot sever the edge to $t_1^i$.
		\item Next consider the case that $2\le j\le k-1$. If $\alpha < \tau_j$, then severing an edge to a neighbor in $\ag{T_p}$ for $2\le p \le k$, because this increases the distance cost by $1$ while saving only a neighborhood cost of $\alpha \left(1 + \frac 1{\numb{T_j}}\right) < \tau_j \left(1 + \frac 1{\numb{T_j}}\right) = 1$. The neighbors in $\ag{T_p}$ have (weakly) more friends and would save even less neighborhood cost. In the case $\alpha \ge \tau_j$, there is again a unique neighbor of type $T_p$ and the case is analogous to the case for agents of type $T_1$. Thus, we have considered all bichromatic edges.
		\item The red agent $t_k^l$ cannot sever the edge towards agent $t_k^i$, because this improves the neighborhood cost by less than $2$ while it increases the distance to both $t_k^i$ and $t_1^i$ by $1$ each.
		\item Finally, consider the case that $\alpha < \tau$. Then, $t_k^l$ cannot sever $t_k^lt_k^m$ for $l\neq m$. Indeed, this would increase the distance cost by $1$ while saving a neighborhood cost of $\alpha \left(1+ \frac 1{\numb{T_k}(\numb{T_k}-1)}\right)\le \alpha \left(1+ \frac 1{(\numb{T_{k-1}}+1)\numb{T_{k-1}}}\right)$. Here, we use that such an edge can only exist if $\numb{T_k}\ge \numb{T_{k-1}}+1$. Hence, the total increase in cost is at least $1-\alpha \left(1+ \frac 1{(\numb{T_{k-1}}+1)\numb{T_{k-1}}}\right) = 1- \alpha \frac {\numb{T_{k-1}}(\numb{T_{k-1}}+1)+1}{(\numb{T_{k-1}}+1)\numb{T_{k-1}}} > 1- \tau \frac {\numb{T_{k-1}}(\numb{T_{k-1}}+1)+1}{(\numb{T_{k-1}}+1)\numb{T_{k-1}}} = 0$.
	\end{itemize}

	Next, we show that it is also not possible to add edges.

	\begin{itemize}
		\item As a first step, we show that agents cannot create bichromatic edges. Let therefore $1\le j < p \le k$ and let $1 \le i \le \numb{T_j}$ and $1\le l \le \numb{T_p}$ with $i\neq j$. Note that $t_j^it_p^l$ is present if $\alpha < \tau_j$ and $j\ge 2$. Hence, we assume that $\alpha\ge \tau_j$ if $j\ge 2$. Then, $t_j^i$ does not benefit from creating the edge $t_j^it_p^l$. Indeed, this decreases her distance cost by exactly $1$ while it increases her neighborhood cost by $\alpha \frac {\numb{T_j}+1}{\numb{T_{j}}}\ge 1$. There, we use that $\alpha\ge \frac{\numb{B}}{\numb{B}+1}$ if $j = 1$ and $\alpha\ge \tau_j$ if $j \ge 2$.
		\item It remains the case of missing edges between red agents for large edge cost. Assume therefore $\alpha\ge \tau$ and let $\numb{T_{k-1}}+1 \le i,l \in \numb{T_k}$. Adding the edge $t_k^it_k^l$ decreases the distance cost for $t_k^i$ by $1$ while increasing her neighborhood cost by $\alpha \left(1 + \frac 1 {\numb{T_{k-1}}(\numb{T_{k-1}}+1)}\right)\ge \tau \frac {\numb{T_{k-1}}(\numb{T_{k-1}}+1)+1}{(\numb{T_{k-1}}+1)\numb{T_{k-1}}} = 1$. Hence, creating this edge is not beneficial for $t_k^i$.
	\end{itemize}
	Together, we have found stable networks for $\alpha$ in the desired range.
\end{proof}
	
Now, we provide the complete proof of the uniqueness statement about pairwise stable networks for an intermediate parameter range.

\UniquePSNintermed*

\begin{proof}
	Let $\frac{\numb{R}}{\numb{R}+1} < \alpha < 1$ and assume that $G$ is pairwise stable network in the \icfgame with cost parameter $\alpha$. By \Cref{prop:ICF-properties}(\ref{prop:curioustype}), the diameter of $G$ is bounded by $2$ and there exists a curious type of agents. By \Cref{prop:ICF-properties}(\ref{prop:curiouscliques}), the curious type of agents forms a clique $C$ and the curious agents of the other type form a clique as well.
	
	Assume towards a contradiction that the bichromatic edges form no matching. Assume that there is an agent $x\in C$ that maintains bichromatic edges with two different agents $y$ and~$z$. We will show that agent $y$ has an incentive to sever the edge $xy$. Consider therefore the network $G' = G - xy$. First, the distance cost of $y$ decreases by at most $1$. Indeed, since all agents of the type of $x$ are still curious in $G'$ and since $y$ forms edges to all curious agents of her type, the distance to all these agents is $2$ in $G'$ and $1$ to agents other than $x$ to which a bichromatic edge exists in $G$. Also, since $y$ is connected to all curious agents of her type, the shortest paths to agents of her own type in $G$ cannot use $x$ and still exist after severing the edge~$xy$.
	Now, the neighborhood cost decreases by $\alpha \left(1 + \frac 1 {\fr{G}{y} + 1}\right) \ge \alpha \left(1 + \frac 1 {\numb{R}}\right) > 1 $. Hence, no agent in $C$ maintains more than one bichromatic edge.
	Now, assume that two agents $w, x\in C$ maintain a bichromatic edge to the same agent $y$. It is quickly checked that severing $xy$ increases the distance cost by $1$ for $y$ and her neighborhood cost decreases by more than $1$, as above.
	
	Together, the bichromatic edges form a matching. Hence, only a minority type can be a curious type and we can conclude that the blue agents are fully intra-connected and that the matching of bichromatic edges is of size $\numb{B}$. It remains to show that all curious red agents maintain edges with non-curious red agents. Assume that $y$ is a curious red agent forming a bichromatic edge to the blue agent $x$ and that there is no edge to a non-curious red agent $z$, i.e., $yz$ is not present in $G$. But then, $\dist{G}{x,z}\ge 3$, contradicting \Cref{prop:ICF-properties}(\ref{prop:curioustype}).
\end{proof}

We provide the computations of pairwise stability of \Cref{ex:PSN-exist-intermed-small}.

\begin{example}[\Cref{ex:PSN-exist-intermed-small} continued]\label{ex:PSN-exist-intermed-small-cont}
	
	We check that the network is pairwise stable.
	
	Blue agents cannot sever a bichromatic edge, because this increases the distance cost by $|R|\ge 2$ while saving a neighborhood cost of less than $2\alpha$. Also, they cannot sever a monochromatic edge, because this reduces the neighborhood cost only by $\alpha<1$. They would also not agree to create another edge with a red agent as this reduces the distance cost only by $1$ while it increases the neighborhood cost by $\alpha \frac{\numb{B}+1}{\numb{B}}\ge 1$.
	
	Red agents cannot sever a monochromatic edge. Using the computation in \Cref{lem:one_step_benefit}, this decreases the neighborhood cost only by $\alpha\left(1 + \frac 1{\numb{R}(\numb{R}-1)}\right)\le \alpha\left(1 + \frac 1{\numb{R}}\right)\le 1$. Finally, agent $r^*$ cannot sever a bichromatic edge as this saves her only a neighborhood cost of $\alpha \frac{\numb{R}}{\numb{R}+1}\le 1$.\hfill $\lhd$
\end{example}

The segregation of pairwise stable networks for intermediate parameter range is a direct computation based on the characterization of \Cref{thm:unique-PSN-intermed}.

\SegregationICFintermed*

\begin{proof}
	Let $\frac{\numb{R}}{\numb{R}+1} < \alpha < 1$ and assume that $G = (\ag{},E)$ is a pairwise stable network for an \icfgame with cost parameter~$\alpha$. 
	
	We start with computing the global segregation. By \Cref{thm:unique-PSN-intermed}, there are $\numb{B}$ bichromatic edges. Additionally, $|E| \ge \numb{B} + 2 {\numb{B}\choose 2} + \numb{B} (n-\numb{B}) = \numb{B} n$. Hence, $\gsm = \frac{|E|-\numb{B}}{|E|} \ge 1 - \frac 1n$.
	
	For the local segregation, we need to compute the quantity $\frac {\fr{G}{u}}{\deg{G}{u}}$ for every agent $u$. We can apply the characterization of \Cref{thm:unique-PSN-intermed} again to find

	$$\frac {\fr{G}{u}}{\deg{G}{u}} = \begin{cases}
	\frac {\numb{B} - 1}{\numb{B}} & \textnormal{if } u \textnormal{ blue,}\\
	\frac {\numb{R} - 1}{\numb{R}} & \textnormal{if } u \textnormal{ red and curious,}\\
	1 & \textnormal{otherwise.}\\
	\end{cases}$$

	Consequently,
	\begin{align*}
		\lsm(G) &= \frac 1n \left(\numb{B}\frac {\numb{B} - 1}{\numb{B}}+ \numb{B} \frac {\numb{R} - 1}{\numb{R}}+ (\numb{R}-\numb{B})\right)\\ &=  \frac 1n \left( n - 1 - \frac{\numb{B}}{\numb{R}}\right)\ge 1 - \frac 2n\text. \qedhere
	\end{align*}
\end{proof}

Next, we show that stars yield extremal segregation and double stars high segregation.

\IFCstarSeg*

\begin{proof}
	Note that in the considered parameter range, the star $\sta_n$ is pairwise stable according to \Cref{prop:ICF-properties}(\ref{prop:ICF-stars}). If there are only agents of one type, then $G = \sta_n$ fulfills $\gsm(G),\lsm(G) = 1$. On the other hand, if there are $2$ blue agents and $n-2$ red agents, consider $G' = \sta_n$ where the center agent is blue. Then $\gsm(G'),\lsm(G') = \frac 1{n-1}$.
\end{proof}

\DSstability*

\begin{proof}
	Consider the double star $\dsta_n$ and let $c_B$ and $c_R$ be the blue and red star center, respectively.
	
	Note that no agent can sever an edge, because this would disconnect the network. Also, no edge between a star center and a leaf node can be created, because it is not profitable for the center node. Indeed, consider a pair of nodes $v\in \ag{R}$ and the central node $c_B$. Adding the edge $c_Bv$ improves the distance to only one node for the agent $c_B$, while the neighborhood cost increases by 
	\begin{align*}
		&\Ncost{\dsta_n+c_Bv}{c_B} - \Ncost{\dsta_n}{c_B} \\
	&=\alpha\left((\deg{\dsta_n}{c_B}+1)\left(1+\frac{1}{\numb{B}}\right)\right.\\&\left.-\deg{\dsta_n}{c_B}\cdot \left(1+\frac{1}{\numb{B}}\right)\right)
	= \alpha\left(\frac{\numb{B} + 1}{\numb{B}}\right) \geq 1\text.
	\end{align*}
	Hence, the edge $c_Bv$ will be rejected by the agent $c_B$. 
	Analogously, a new edge between the center node $c_R$ and a node $v\in \ag{B}$ is not profitable for the center node $c_R$, because it increases the neighborhood cost by $\alpha\left(1+\frac{1}{\numb{R}}\right)\geq 1$ an decreases the distance cost by 1.
	
	Next, consider the case of creating a bichromatic edge. Then, the distance cost is decreased by $2$, while the neighborhood cost is increased by $\frac 32\alpha\ge 2$.
	
	Finally, consider the creation of an edge between two nodes $u, v$ of the same type, say type $R$. The new edge improves the distance cost by $1$ for both agents but increases the neighborhood cost by
	$\alpha\left(2\cdot\left(1+\frac{1}{2+1}\right)-1-\frac{1}{2}\right) = \frac{7\alpha}{6}\geq 1$.
	Hence, $\dsta_n$ is pairwise stable for any $\alpha\ge \frac 43$.
	
	It remains to compute the segregation measures for the double star.
	
	First, $\gsm(\dsta_n) = \frac{\numb{B}-1 + \numb{R} -1}{\numb{} -1} = 1 - \frac 1{\numb{}-1}$. Second, $\lsm(\dsta_n) = \frac 1{\numb{}} \left( \numb{B}-1 + \numb{R} -1 + \frac{\numb{B}-1}{\numb{B}} + \frac{\numb{R}-1}{\numb{R}}\right) = 1 - \frac 1{\numb{}}\left( \frac 1{\numb{B}} + \frac 1{\numb{R}}\right)\ge 1 - \frac 2 {\numb{}}$.
\end{proof}

In addition to the results in the body of the paper, we provide an example where the uniqueness of \Cref{prop:ICF-properties}(\ref{prop:ICF-cliques}) does not hold anymore.
\begin{example}\label{ex:PSN-exist-unique-small}
	Consider an \icfgame with two agent types. Let $n_B \ge 6$ and $\frac 67 \le \alpha \le \frac{\numb{R}}{\numb{R}+1}$. We fix a specific red agent $r^*\in \ag{R}$ and consider the network $G = (\ag{}, E)$ with $E = \{vw\colon v,w\in \ag{R}\}\cup \{vr^*\colon v\in \ag{}\setminus \{r^*\}\}$, i.e., the red type is fully intra-connected and there is a special agent $r^*$ to which all agents are connected. The structure of this network is depicted in \Cref{fig:PSN-exist-unique-small}. 
	\begin{figure}[h]
		\centering
		\begin{tikzpicture}
		
		\node (c) at (3,0){};
		
		\node[protodiam] (r4) at ($(c) + (60:1)$) {};
		\node[protodiam] (r1) at ($(c) + (120:1)$) {};
		\node[protodiam, label = $r^*$] (r2) at ($(c) + (180:1)$) {};
		\node[protodiam] (r3) at ($(c) + (240:1)$) {};
		\node[protodiam] (r5) at ($(c) + (300:1)$) {};
		\node[protodiam] (r6) at ($(c) + (0:1)$) {};

		\node (d) at (-1,0){};
		
		\node[protovertex] (b4) at ($(d) + (60:1)$) {};
		\node[protovertex] (b1) at ($(d) + (120:1)$) {};
		\node[protovertex] (b2) at ($(d) + (180:1)$) {};
		\node[protovertex] (b3) at ($(d) + (240:1)$) {};
		\node[protovertex] (b5) at ($(d) + (300:1)$) {};
		\node[protovertex] (b6) at ($(d) + (0:1)$) {};
		
		\foreach \i in {1,3,4,5,6}
		{
			\draw (r2) edge (b\i);
		}
		
		\foreach \i/\j in {1/2,1/3,1/4,1/5,1/6,2/3,2/4,2/5,2/6,3/4,3/5,3/6,4/5,4/6,5/6}{
			\draw (r\i) edge (r\j);}
		
		\draw[bend left] (b2) edge (r2);
		\end{tikzpicture}
		
		\caption{Pairwise stable network for $\frac 67 \le \alpha \le \frac{\numb{R}}{\numb{R}+1}$ with $\numb{B} =6$ and $\numb{R} = 6$ blue and red agents, respectively.}
		\label{fig:PSN-exist-unique-small}
	\end{figure} 
	If $\frac 67 \le \alpha \le \frac{\numb{B}}{\numb{B}+1}$, it is even possible to interchange the roles of the two agent types.
	
	Pairwise stability of this network follows by straightforward considerations, similar to the computations in \Cref{ex:PSN-exist-intermed-small}. \hfill $\lhd$  
\end{example}

\subsection{Decreasing Effort of Integration}

We start with the proofs of the statements collected in \Cref{prop:PSN-overview}.

\DEIoverview*

\begin{proof}
	We prove the statements one by one.
	\begin{enumerate}
		\item If some edge is not present, it has cost at most $2\alpha < 1$ and creating it decreases the distance cost by at least $1$.
		
		\item Creating a monochromatic edge has cost $\alpha <1$ and decreases the distance cost by at least $1$.
		
		\item Let $\alpha < 1$. Assume that there are agents $u,v\in \ag{}$ with $\dist G{u,v}\ge 3$. Then, creating $uv$ increases the neighborhood cost by at most $2\alpha < 2$, while decreasing the distance cost by at least $2$ for each of its endpoints. Hence, $G$ is not pairwise stable.
		
		\item Clearly, no monochromatic edge can be severed. Now, consider a bichromatic edge $uv$. Then, severing $uv$ increases the total cost for~$v$ by $1 - \alpha\left(1 + \frac 1{n - \numb{\mathcal T(v)}}\right)\ge 1 - \alpha\left(1 + \frac 1{n - \numb{R}}\right)\ge 1 - \frac{n - \numb{R}}{n - \numb{R} + 1}\left(1 + \frac 1{n - \numb{R}}\right) = 0$. Hence, also bichromatic edges cannot be severed.
		
		\item No edge can be severed, because these networks are trees. Due to the large distance cost, no agent favors to create an edge if this only improves the distance cost by~$1$. Hence, $\sta_n$ is stable, the two centers of $\dsta_n$ will not agree to build further edges, and leaves of $\dsta_n$ will not agree to create further monochromatic edges. Finally, the cost for creating an edge between two leaves of different types is $2\alpha\ge 2$ which does not make up for a distance improvement of $2$.
		
		\item As for $\dsta_n$, no edges can be severed, and the centers will not benefit from creating further edges. Also, leaves have no incentive to create monochromatic edges. Finally, the cost for a bichromatic edge between leaves of different types is $\frac 32 \alpha \ge 2$, but creating such an edge yields only a distance improvement of $2$. \qedhere
	\end{enumerate}
\end{proof}

\FewBichrom*

\begin{proof}
	The proof of both statements follows from a unified approach. Let $G = (\ag{}, E)$ be a fully intra-connected and pairwise stable network. Let $u\in \ag{}$. By full intra-connectivity, severing one of several bichromatic edges incident to $u$, increases the distance cost of $u$ by exactly $1$ while decreasing the neighborhood cost by $\Delta = \alpha \frac {\en{G}{u} + 1}{\en{G}{u}}$. If $\alpha > \frac{\numb{R}}{\numb{R} + 1}$, then $\Delta > 1$ and severing a bichromatic edge is beneficial for $u$. This proves the second statement. If even $\alpha > \frac{\numb{B}}{\numb{B} + 1}$ and $u$ is an agent of the majority type, then $\en{G}{u} \le \numb{B}$, and $\Delta \ge  \alpha \frac {\numb{B} + 1}{\numb{B}} > 1$.
\end{proof}

\DEIexistPSN*

\begin{proof}
	Suppose that $\mathcal T = \{T_1,\dots, T_k\}$ with $\numb{T_1}\le \dots \le \numb{T_k}$ and, for each $1\le j \le k$, $\ag{T_j} = \{t^1_j,\dots, t^{\numb{T_j}}_j\}$. By \Cref{prop:PSN-overview}, it suffices to consider the parameter range $\frac{\numb{} - \numb{R}}{\numb{} - \numb{R}+1} < \alpha < 1$. We will even provide pairwise stable networks whose parameter ranges overlap. 
	
	First, we will generalize the network of \Cref{ex:DEI-stability} to an arbitrary number of agent types. Let $j^* = \min(\{1\le j \le k\colon \numb{T_j}\ge 2\}\cup\{k\})$, i.e., the index of the smallest type of size at least~$2$ or the index of the last type if there exists exactly one agent per type. Consider the network $G = (\ag{},E)$ with edge set defined by
	
	\begin{itemize}
		\item $\{t_j^i,t_j^l\}\in E$ for $1\le j\le k$, $1 \le i < l \le \numb{T_j}$,
		\item $\{t_{j^*}^1, t_j^i\}\in E$ for $1\le j\le k$, $j\neq j^*$, $1 \le i \le \numb{T_j}$, and
		\item no further edges are in $E$.
	\end{itemize}
	
	We provide now conditions, under which the network $G$ is pairwise stable.
	\begin{lemma}\label{lem:PSN-centralagent}
	The network $G$ is pairwise stable if 
	\begin{enumerate}[label=\textit{(\roman*)}]
		\item $j^* = k$ and $\frac 23 \le \alpha \le 1$,
		\item $k = 2$, $j^* = k$, $\numb{T_k} \ge 2$ and $\frac 12 \le \alpha \le 1$, or
		\item $\frac 23 \le \alpha \le \frac {\numb{} - \numb{T_{j^*}}}{\numb{} - \numb{T_{j^*}} + 1}$.
	\end{enumerate}	
	\end{lemma}
	\begin{proof}
		\begin{enumerate}[label=\textit{(\roman*)}]
			\item If $j^* = k$ and $\frac 23 \le \alpha \le 1$, then no monochromatic edge can be severed because of $\alpha \le 1$. Bichromatic edges cannot be severed due to connectivity, and creating an edge costs $\frac 32\alpha \ge 1$ for an agent of type different to $k$ while it decreases her distance cost by exactly~$1$.
			\item Next, consider the case that $k = 2$, $j^* = k$, $\numb{T_k} \ge 2$ and $\frac 12 \le \alpha \le 1$. Then, again, no monochromatic edge can be severed because of $\alpha \le 1$. The unique bichromatic edge cannot be severed as this would disconnect the network. Also, adding another bichromatic edge must include a non-curious red agent. This agent would increase her neighborhood cost by $2\alpha\ge 1$ while only decreasing her distance cost by $1$.
			\item Now, assume that $\frac 23 \le \alpha \le \frac {\numb{} - \numb{T_{j^*}}}{\numb{} - \numb{T_{j^*}} + 1}$.
		Again, monochromatic edges cannot be severed as $\alpha < 1$. Further, bichromatic edges incident to an agent $t_j^1$ for $1\le j \le j^*-1$ cannot be severed as this would disconnect the network. Next, agent~$t_{j^*}^1$ cannot sever another bichromatic edge, because this would increase her cost by $1 - \alpha \frac {\numb{} - \numb{T_{j^*}} + 1}{\numb{} - \numb{T_{j^*}}}\ge 0$. Further, for $j^* <  j\le k$ and $1 \le i \le \numb{T_j}$, agent~$_j^i$ cannot sever $\{t_{j^*}^1, t_j^i\}$, because this increases the distance to at least $\numb{T_{j^*}}\ge 2$ agents (in $T_{j^*}$) by $1$ while decreasing the neighborhood cost by $2$.
	
		It remains to consider the creation of edges. Every agent in $\ag{}\setminus \ag{T_{j^*}}$ entertains exactly one bichromatic edge. Creating a second bichromatic edge costs $\frac 32\alpha \ge 1$ while it decreases the distance cost by exactly~$1$. Together, the network is pairwise stable. \qedhere
		\end{enumerate}
	\end{proof}
	
	Second, we generalize the network from \Cref{thm:DEIcharactPSN}. To this end, we design an algorithm that constructs pairwise stable networks. In the special case of $2$ agent types, it yields the networks encountered in \Cref{thm:DEIcharactPSN}. Note that this must already hold specifically for the parameter range where the uniqueness of the theorem applies. 
	
	Therefore, consider the network $G' = (\ag{},E')$ where the edge set $E'$ is computed according to \Cref{alg:psn-edges}.

	\begin{algorithm}
    
	  \caption{Determination of edge set for network~$G'$.\label{alg:psn-edges}}
	  \label{alg:PSN-edges}
	  \begin{flushleft}
	  \textbf{Input:} Set of agents~$\ag{}$. \\
	  \textbf{Output:} Edge set~$E'$.  	
	  \end{flushleft}

	  \begin{algorithmic}[] 
	\STATE $E' \leftarrow \{\{t_j^i,t_j^l\}\colon 1\le j\le k, 1 \le i < l \le \numb{T_j}\}$
	\WHILE {there exist $u,v\in\ag{}$ with $\dist{(\ag{},E')}{u,v}\ge 3$}
	\STATE $E' \leftarrow E'\cup \{uv\}$
	\ENDWHILE
	\RETURN $E'$
	 \end{algorithmic}
	\end{algorithm}

	The algorithm starts with the fully intra-connected network without any bichromatic edges. Then, bichromatic edges are added whenever the distance between two agents is too large. Clearly, this algorithm has to terminate by returning $E'$ after at most $\numb{}\choose 2$ executions of the while loop.
	
	\begin{lemma}\label{lem:PSN-props}
		The following properties are valid.
		\begin{itemize}
			\item The diameter of $G'$ satisfies $\diam(G') \le 2$.
			\item Every triangle\footnote{A triangle is defined as a complete subnetwork induced by three vertices.} in $G'$ consists of monochromatic edges only.
			\item Every agent incident to at most $k-1$ bichromatic edges in $G'$.
		\end{itemize}
	\end{lemma}
	
	\begin{proof}
		The first property is immediate from the definition of the while loop. We show the second property by induction. More precisely, we show that, during the whole algorithm, every triangle in $(V,E')$ contains only monochromatic edges by induction over the number $w$ of executions of the while loop. If $w = 0$, then there are no bichromatic edges and the assertion is true. Now let $w\ge 1$ and assume that no triangle in the network $(V,E')$ before adding the $w$-th edge in the while loop contains a bichromatic edge. Assume that we add the bichromatic edge $uv$ as the $w$-th bichromatic edge. By induction, only a triangle containing $uv$ can contain a bichromatic edge. But then, there exists an agent $x\in \ag{}$ with $ux, xv\in E'$, contradicting the distance condition of the while loop.
		
		For the third property, we observe that every agent can add at most one bichromatic edge to an agent of each fixed type. Once this edge is added, the distance to all agents of this type is at most $2$ due to the intra-connectivity of the network. As there are at most $k-1$ other types, the assertion follows. \qedhere 
	\end{proof}
	It is easy to deduce the pairwise stability of $G'$.
	
	\begin{lemma}\label{lem:PSN-algo}
		The network $G'$ is pairwise stable for $\frac k{k+1}\le \alpha \le 1$.
	\end{lemma}
	\begin{proof}
		As in previous networks, monochromatic edges cannot be severed because of $\alpha \le 1$. Now, consider a bichromatic edge $uv$. Then, $\dist{G-uv}{u,v} \ge 3$. Indeed, if $\dist{G-uv}{u,v} = 2$, then $uv$ is part of a triangle, contradicting the second statement in \Cref{lem:PSN-props}. Hence, severing $uv$ increases the distance cost for $uv$ by at least $2$ while saving a neighborhood cost of at most $2$.
		
		It remains to consider the creation of edges. As the network is fully intra-connected, only bichromatic edges can be created. Hence, consider the creation of a bichromatic edge~$uv$. Its creation decreases the distance cost for $u$ by exactly~$1$. Indeed, as $\diam(G') \le 2$, the distance to $v$ is decreased by exactly~$1$, and the distance to other agents is no shorter. On the other hand, as $u$ is incident to at most $k-1$ bichromatic edges, the creation of $uv$ costs at least $\alpha\left(1 + \frac 1k\right)\ge 1$. Hence, the total cost for $u$ cannot have decreased.
	\end{proof}
	
	To conclude the proof, we want to argue that we can cover the whole parameter range of $\alpha$. First, we cover the range until $\alpha = \frac 23$. According to \Cref{prop:PSN-overview}(\ref{prop:PSN-over-kn}), this is covered by $\com_n$ if $\numb{}-\numb{T_k}\ge 2$. In particular, this is the case if $k\ge 3$ or $\numb{T_1}\ge 2$. If $k = 2$ and $\numb{T_1} = 1$, we can apply case \textit{(ii)} of \Cref{lem:PSN-centralagent} if $\numb{T_k}\ge 2$. If $\numb{T_k} = 1$, then the network consisting of two agents of different types, connected by an edge, is pairwise stable.
	
	Finally, consider the parameter range $\frac 23 \le \alpha \le 1$. If $j^* = k$, then case \textit{(i)} of \Cref{lem:PSN-centralagent} applies. Otherwise, $j^* < k$, and therefore $\numb{} - \numb{T_{j^*}}\ge k$. This implies that $\frac {\numb{} - \numb{T_{j^*}}}{\numb{} - \numb{T_{j^*}} + 1}\ge \frac k{k+1}$, and the parameter range is covered by case \textit{(iii)} of \Cref{lem:PSN-centralagent} and \Cref{lem:PSN-algo}.
\end{proof}

\DEIuniquePSN*

\begin{proof}
Consider a network $G = (V,E)$ satisfying the assumptions of the corollary. We start with the global segregation measure. According to \Cref{thm:DEIcharactPSN}, there are $\numb{B}$ bichromatic edges and a total of $\numb{B} + \numb{B}(\numb{B} - 1)/2 + \numb{R}(\numb{R} - 1)/2 \ge \numb{B} + \numb{B}(\numb{B} - 1)/2 + \numb{B}(\numb{R} - 1)/2 = \numb{B}\numb{}/2$ edges. Hence,
$$\gsm(G) = \frac{|E|-\numb{B}}{|E|} = 1 - \frac{\numb{B}}{|E|}\ge 1 - \frac{\numb{B}}{\numb{B}\numb{}/2} = 1 -\frac 2n\text.$$
Using the characterization in \Cref{thm:DEIcharactPSN} once again, the computation of the local segregation measure is identical as in the proof of \Cref{cor:SegrICFintermed}.
\end{proof}

\section{Detailed Experimental Analysis}\label{appendix:plots}
In this section we provide more detailed experimental analysis complementing Section~\ref{sec:experiments}.

\subsection{Details about the Experimental Setup}
For our simulation experiments we first generated an initial network and an intitial agent distribution. Then agents are activated and compute a best possible edge addition or edge deletion. This sequential activation process is then run until no agent has an improving move and a pairwise stable network is found. We now discuss the details of this setup. 

\paragraph*{General Setup} Our experiments considered 1000 agents partitioned into two types with 500 agents each.
For each run we chose
\begin{itemize}
	\item a random spanning tree or a grid as initial network,
	\item an integrated or perfectly segregated inital agent distribution,
	\item if best response moves or if best add-only moves are performed,
	\item if the segregation strength is measured via the local segregation measure $\lsm$ or via the global segregation measure $\gsm$, and
	\item the value of $\alpha$ in 15 steps between $5$ and $255$.
\end{itemize}
In total this yielded $2^4*15 = 240$ different configurations and for every configuration we simulated 50 runs, yielding a total number of 12000 considered networks. 

\paragraph*{Generating the Initial Networks} We considered random spanning trees and grids as initial networks. We used grids of size $20\times 50$. Moreover, we sampled the random spanning trees by the following scheme: starting from a single node, we add nodes one-by-one, and each new arriving node attaches to one of the existing nodes chosen uniformly at random.

\paragraph*{Generating the Initial Agent Distribution}
We focus on two cases: perfectly segregated and integrated initial state. An integrated state is sampled by a uniformly random type assignment to each node. To generate a perfectly segregated spanning tree, we generate two one-type spanning trees of 500 nodes and join them by connecting the initial nodes of each tree. A perfectly segregated grid is sampled by assigning one type to all 500 nodes in the first ten rows and another type to the rest.

\paragraph*{Random Activation of the Agents}
We start with marking all nodes as willing to improve.
In each step of the algorithm, one agent is chosen from the set of the marked nodes uniformly at random. This active agent is searching for a best-allowed move. If no move is possible, the agent is unmarked. If the agent has an improving move, the new strategy is applied to the network, and all agents move back to the set of the marked nodes to be ready to become activated again. The algorithm stops when the last agent is unmarked.

\paragraph*{Convergence Criteria}
Figure~\ref{plot:avg_segr_timeline} shows a representative timeline of the local segregation of the obtained networks in each step of the best move dynamics starting from a random tree with a random color distribution. We observe that the segregation value quickly reaches a high value and remains in the interval $[0.8, 1]$ until the end of the execution of the dynamics. It illustrates the need for relaxation of the solution concept to avoid long calculations. Therefore, our experimental study of the best move dynamics uses  $1.01$-approximate pairwise stable states as solution concept. We say that a network is a $1.01$-\emph{approximate pairwise stable} if no agent can improve her cost by more than a factor of $1.01$. The approximation factor is chosen empirically to minimize the convergence time and the approximation gap.

Note that for the add-only move dynamics, the process naturally stops at the latest when a complete network is reached. Hence, the computation time is rather low compared to the best move dynamics and we could consider perfect pairwise stable networks.

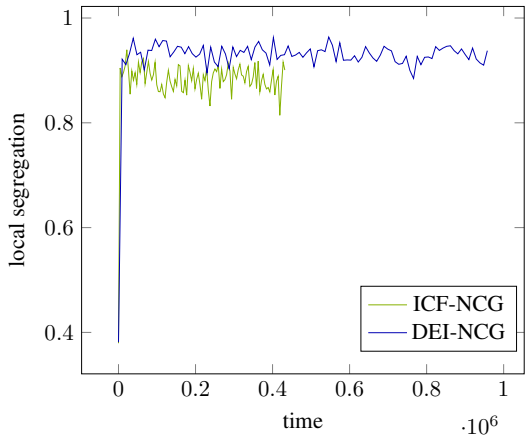
\begin{figure}[!ht]
\centering
\begin{minipage}{0.4\textwidth}
\resizebox {\textwidth} {!} {
\begin{tikzpicture}
\begin{axis}[
		legend style={
		at={(0.63,0.15)},
		anchor=west},
		xlabel= {time},
		ylabel= {local segregation},
		baseline
	]
 
\addplot+ [color = lime!70!black, mark = none] coordinates{  (0,0.379333) (4318,0.904667) (8636,0.888000) (12954,0.896000) (17272,0.908667) (21590,0.939333) (25908,0.900000) (30226,0.854667) (34544,0.899000) (38862,0.881333) (43180,0.895762) (47498,0.873333) (51816,0.916000) (56134,0.896000) (60452,0.916000) (64770,0.922667) (69088,0.875000) (73406,0.896333) (77724,0.917143) (82042,0.894333) (86360,0.881667) (90678,0.891667) (94996,0.915333) (99314,0.872667) (103632,0.859667) (107950,0.859333) (112268,0.873000) (116586,0.855000) (120904,0.848000) (125222,0.878000) (129540,0.900333) (133858,0.879810) (138176,0.869000) (142494,0.860000) (146812,0.882667) (151130,0.870476) (155448,0.911000) (159766,0.908667) (164084,0.860000) (168402,0.857667) (172720,0.881333) (177038,0.852833) (181356,0.909000) (185674,0.896333) (189992,0.880333) (194310,0.903667) (198628,0.865000) (202946,0.883000) (207264,0.894667) (211582,0.872667) (215900,0.845000) (220218,0.891667) (224536,0.880667) (228854,0.903333) (233172,0.862000) (237490,0.832000) (241808,0.878667) (246126,0.895333) (250444,0.902667) (254762,0.897667) (259080,0.910000) (263398,0.865667) (267716,0.902667) (272034,0.878667) (276352,0.882333) (280670,0.892000) (284988,0.913667) (289306,0.899000) (293624,0.845000) (297942,0.905667) (302260,0.893333) (306578,0.884000) (310896,0.901333) (315214,0.912000) (319532,0.891333) (323850,0.889000) (328168,0.876333) (332486,0.897810) (336804,0.907667) (341122,0.869667) (345440,0.874000) (349758,0.887667) (354076,0.915000) (358394,0.865667) (362712,0.917667) (367030,0.859333) (371348,0.880667) (375666,0.892333) (379984,0.873333) (384302,0.864667) (388620,0.867333) (392938,0.859000) (397256,0.881667) (401574,0.901333) (405892,0.853000) (410210,0.879333) (414528,0.888667) (418846,0.814333) (423164,0.870000) (427482,0.915667) (431800,0.900667)}; 
\addlegendentry{\icfgame};

\addplot+ [color = blue!70!black, mark = none] coordinates{  (0,0.381381) (9576,0.921333) (19152,0.911000) (28728,0.932667) (38304,0.961333) (47880,0.930000) (57456,0.934667) (67032,0.901000) (76608,0.938667) (86184,0.937952) (95760,0.959667) (105336,0.945333) (114912,0.957000) (124488,0.956000) (134064,0.926000) (143640,0.935333) (153216,0.946000) (162792,0.944333) (172368,0.932000) (181944,0.945333) (191520,0.932000) (201096,0.925333) (210672,0.931000) (220248,0.947667) (229824,0.895333) (239400,0.944333) (248976,0.916000) (258552,0.907333) (268128,0.946000) (277704,0.931333) (287280,0.903333) (296856,0.944667) (306432,0.926333) (316008,0.938333) (325584,0.934333) (335160,0.946000) (344736,0.920667) (354312,0.944667) (363888,0.955333) (373464,0.939333) (383040,0.933000) (392616,0.910333) (402192,0.962667) (411768,0.921000) (421344,0.928667) (430920,0.930000) (440496,0.946667) (450072,0.926667) (459648,0.929333) (469224,0.934333) (478800,0.925000) (488376,0.933667) (497952,0.941000) (507528,0.907000) (517104,0.937667) (526680,0.939333) (536256,0.936000) (545832,0.963667) (555408,0.948000) (564984,0.917000) (574560,0.952667) (584136,0.919333) (593712,0.920000) (603288,0.919667) (612864,0.928000) (622440,0.916000) (632016,0.925000) (641592,0.945333) (651168,0.933333) (660744,0.923667) (670320,0.917333) (679896,0.929333) (689472,0.946000) (699048,0.940667) (708624,0.937333) (718200,0.916667) (727776,0.911667) (737352,0.913333) (746928,0.927000) (756504,0.900000) (766080,0.884667) (775656,0.926667) (785232,0.911000) (794808,0.925000) (804384,0.925667) (813960,0.922667) (823536,0.945333) (833112,0.938143) (842688,0.942667) (852264,0.946000) (861840,0.947000) (871416,0.937667) (880992,0.931667) (890568,0.941000) (900144,0.930333) (909720,0.920667) (919296,0.942667) (928872,0.921667) (938448,0.914333) (948024,0.910000) (957600,0.937667)}; 
\addlegendentry{\deigame};

\end{axis}
\end{tikzpicture}
}
\end{minipage}
\caption{A timeline of the local segregation of a network obtained by the best response dynamic for $n=50$, $\alpha=15$ starting from a tree with random color distribution in the  \icfgame and \deigame.}
\label{plot:avg_segr_timeline}
\end{figure}

\paragraph*{Visualization of Our Results}
In the next section we show 
box-and-whiskers plots of the local and global segregation for the networks obtained by the best move dynamics for $n=1000$ over 50 runs.  Lower and upper whiskers are the minimal and maximal local segregation values over 50 runs of the algorithm. The middle lines are the median values, while the bottom and top of the boxes represent the first and the third quartiles. 

\subsection{Additional Experiments Regarding the Local Segregation Measure}
This section provides additional empirical results for the local segregation measure for the \icfgame and the \deigame. 

\subsubsection{Results for the \icfgame}
The following figures are box-and-whiskers plots showing the obtained local segregation in our experiments for the sequential process of the \icfgame. All plots show high segregation of stable networks can be avoided by the lower cost of the connections ($\alpha < 30$) and if the add-only process starts from a sparse low-segregated network.

\begin{figure}[!ht]
\centering
\begin{minipage}{0.23\textwidth}
\resizebox {\textwidth} {!} {
\begin{tikzpicture}
\begin{axis}[
legend columns=-1,
legend entries={segregated grid;,random grid;,segregated tree;,random tree},
legend to name=named,
legend style={nodes={scale=0.65, transform shape}},
xlabel= {$\alpha$},
ylabel= {local segregation},
boxplot/draw direction=y,
baseline,
xtick = {1,2,3, 4, 5, 6, 7, 8},
xticklabels = {5, 10, 15, 20, 25, 30, 35, 40},
ymin=0.5,
ymax=1
]
\addplot[red!50!black, domain=1.1:1.11]{0.55};
\addplot[blue!50!black, domain=1.1:1.11]{0.55};
\addplot[yellow!70!black, domain=1.1:1.11]{0.55};
\addplot[lime!70!black, domain=1.1:1.11]{0.55};

\addplot+ [color = lime!70!black,solid,boxplot prepared = {box extend=0.3, draw position = 1, lower whisker = 0.603168, lower quartile = 0.607732, median = 0.612729, upper quartile = 0.617235, upper whisker = 0.621527},]coordinates{}; 
\addplot+ [color = lime!70!black,solid,boxplot prepared = {box extend=0.3, draw position = 2, lower whisker = 0.631120, lower quartile = 0.639577, median = 0.642429, upper quartile = 0.645221, upper whisker = 0.648405},]coordinates{}; 
\addplot+ [color = lime!70!black,solid,boxplot prepared = {box extend=0.3, draw position = 3, lower whisker = 0.657872, lower quartile = 0.663968, median = 0.668497, upper quartile = 0.673094, upper whisker = 0.686780},]coordinates{}; 
\addplot+ [color = lime!70!black,solid,boxplot prepared = {box extend=0.3, draw position = 4, lower whisker = 0.678090, lower quartile = 0.683689, median = 0.686953, upper quartile = 0.692236, upper whisker = 0.707467},]coordinates{}; 
\addplot+ [color = lime!70!black,solid,boxplot prepared = {box extend=0.3, draw position = 5, lower whisker = 0.691792, lower quartile = 0.701712, median = 0.707269, upper quartile = 0.711191, upper whisker = 0.718742},]coordinates{}; 
\addplot+ [color = lime!70!black,solid,boxplot prepared = {box extend=0.3, draw position = 6, lower whisker = 0.708435, lower quartile = 0.714302, median = 0.718056, upper quartile = 0.722919, upper whisker = 0.725387},]coordinates{}; 

\addplot+ [color = yellow!70!black,solid,boxplot prepared = {box extend=0.3, draw position = 1, lower whisker = 0.614643, lower quartile = 0.625454, median = 0.629678, upper quartile = 0.633610, upper whisker = 0.642933},]coordinates{}; 
\addplot+ [color = yellow!70!black,solid,boxplot prepared = {box extend=0.3, draw position = 2, lower whisker = 0.657352, lower quartile = 0.674074, median = 0.677401, upper quartile = 0.681813, upper whisker = 0.692238},]coordinates{}; 
\addplot+ [color = yellow!70!black,solid,boxplot prepared = {box extend=0.3, draw position = 3, lower whisker = 0.704792, lower quartile = 0.709536, median = 0.716545, upper quartile = 0.719788, upper whisker = 0.731335},]coordinates{}; 
\addplot+ [color = yellow!70!black,solid,boxplot prepared = {box extend=0.3, draw position = 4, lower whisker = 0.727169, lower quartile = 0.741266, median = 0.745318, upper quartile = 0.748648, upper whisker = 0.756755},]coordinates{}; 
\addplot+ [color = yellow!70!black,solid,boxplot prepared = {box extend=0.3, draw position = 5, lower whisker = 0.746182, lower quartile = 0.761763, median = 0.766180, upper quartile = 0.769747, upper whisker = 0.781742},]coordinates{}; 
\addplot+ [color = yellow!70!black,solid,boxplot prepared = {box extend=0.3, draw position = 6, lower whisker = 0.774973, lower quartile = 0.782826, median = 0.787657, upper quartile = 0.790773, upper whisker = 0.803171},]coordinates{};

\addplot+ [color = red!70!black,solid,boxplot prepared = {box extend=0.3, draw position = 1, lower whisker = 0.680514, lower quartile = 0.686651, median = 0.690038, upper quartile = 0.693975, upper whisker = 0.701583},]coordinates{}; 
\addplot+ [color = red!70!black,solid,boxplot prepared = {box extend=0.3, draw position = 2, lower whisker = 0.705936, lower quartile = 0.714685, median = 0.718417, upper quartile = 0.720312, upper whisker = 0.729312},]coordinates{}; 
\addplot+ [color = red!70!black,solid,boxplot prepared = {box extend=0.3, draw position = 3, lower whisker = 0.738681, lower quartile = 0.746192, median = 0.750123, upper quartile = 0.753216, upper whisker = 0.759837},]coordinates{}; 
\addplot+ [color = red!70!black,solid,boxplot prepared = {box extend=0.3, draw position = 4, lower whisker = 0.767805, lower quartile = 0.775313, median = 0.777602, upper quartile = 0.780722, upper whisker = 0.785118},]coordinates{}; 
\addplot+ [color = red!70!black,solid,boxplot prepared = {box extend=0.3, draw position = 5, lower whisker = 0.786601, lower quartile = 0.795219, median = 0.797798, upper quartile = 0.799603, upper whisker = 0.805759},]coordinates{}; 
\addplot+ [color = red!70!black,solid,boxplot prepared = {box extend=0.3, draw position = 6, lower whisker = 0.805258, lower quartile = 0.809236, median = 0.811469, upper quartile = 0.813139, upper whisker = 0.819814},]coordinates{};

\addplot+ [color = blue!50!black,solid,boxplot prepared = {box extend=0.3, draw position = 1, lower whisker = 0.578120, lower quartile = 0.583523, median = 0.585371, upper quartile = 0.588368, upper whisker = 0.594768},]coordinates{}; 
\addplot+ [color = blue!50!black,solid,boxplot prepared = {box extend=0.3, draw position = 2, lower whisker = 0.621240, lower quartile = 0.631941, median = 0.634500, upper quartile = 0.636490, upper whisker = 0.638341},]coordinates{}; 
\addplot+ [color = blue!50!black,solid,boxplot prepared = {box extend=0.3, draw position = 3, lower whisker = 0.649372, lower quartile = 0.653145, median = 0.655508, upper quartile = 0.657635, upper whisker = 0.667402},]coordinates{}; 
\addplot+ [color = blue!50!black,solid,boxplot prepared = {box extend=0.3, draw position = 4, lower whisker = 0.654096, lower quartile = 0.663345, median = 0.666467, upper quartile = 0.670526, upper whisker = 0.673912},]coordinates{}; 
\addplot+ [color = blue!50!black,solid,boxplot prepared = {box extend=0.3, draw position = 5, lower whisker = 0.669777, lower quartile = 0.673172, median = 0.675110, upper quartile = 0.677894, upper whisker = 0.682192},]coordinates{}; 
\addplot+ [color = blue!50!black,solid,boxplot prepared = {box extend=0.3, draw position = 6, lower whisker = 0.674330, lower quartile = 0.682455, median = 0.685261, upper quartile = 0.688337, upper whisker = 0.693348},]coordinates{}; 

\end{axis}
\end{tikzpicture}
}
\end{minipage}
\begin{minipage}{0.23\textwidth}
\resizebox {\textwidth} {!} {
\begin{tikzpicture}
\begin{axis}[
xlabel= {$\alpha$},
ylabel= {local segregation},
boxplot/draw direction=y,
baseline,
xtick = {1, 2, 3, ..., 11},
xticklabels  = {5, 30, 55, 80, 105, 130, 155, 180, 205, 230, 255},
ymin=0.5,
ymax=1
]

\addplot+ [color = lime!70!black,solid,boxplot prepared = {box extend=0.4, draw position = 1, lower whisker = 0.603168, lower quartile = 0.607732, median = 0.612729, upper quartile = 0.617235, upper whisker = 0.621527},]coordinates{}; 
\addplot+ [color = lime!70!black,solid,boxplot prepared = {box extend=0.4, draw position = 2, lower whisker = 0.708435, lower quartile = 0.714302, median = 0.718056, upper quartile = 0.722919, upper whisker = 0.725387},]coordinates{}; 
\addplot+ [color = lime!70!black,solid,boxplot prepared = {box extend=0.4, draw position = 3, lower whisker = 0.734886, lower quartile = 0.763801, median = 0.770590, upper quartile = 0.774945, upper whisker = 0.779703},]coordinates{}; 
\addplot+ [color = lime!70!black,solid,boxplot prepared = {box extend=0.4, draw position = 4, lower whisker = 0.799903, lower quartile = 0.804931, median = 0.810123, upper quartile = 0.816841, upper whisker = 0.833158},]coordinates{}; 
\addplot+ [color = lime!70!black,solid,boxplot prepared = {box extend=0.4, draw position = 5, lower whisker = 0.835867, lower quartile = 0.861096, median = 0.877844, upper quartile = 0.903136, upper whisker = 0.909231},]coordinates{}; 
\addplot+ [color = lime!70!black,solid,boxplot prepared = {box extend=0.4, draw position = 6, lower whisker = 0.931126, lower quartile = 0.942213, median = 0.948357, upper quartile = 0.950541, upper whisker = 0.955797},]coordinates{}; 
\addplot+ [color = lime!70!black,solid,boxplot prepared = {box extend=0.4, draw position = 7, lower whisker = 0.940105, lower quartile = 0.944721, median = 0.947245, upper quartile = 0.949611, upper whisker = 0.952959},]coordinates{}; 
\addplot+ [color = lime!70!black,solid,boxplot prepared = {box extend=0.4, draw position = 8, lower whisker = 0.945450, lower quartile = 0.952702, median = 0.955511, upper quartile = 0.958615, upper whisker = 0.964269},]coordinates{}; 
\addplot+ [color = lime!70!black,solid,boxplot prepared = {box extend=0.4, draw position = 9, lower whisker = 0.950111, lower quartile = 0.956789, median = 0.959013, upper quartile = 0.959888, upper whisker = 0.961478},]coordinates{}; 
\addplot+ [color = lime!70!black,solid,boxplot prepared = {box extend=0.4, draw position = 10, lower whisker = 0.958757, lower quartile = 0.960711, median = 0.962101, upper quartile = 0.965926, upper whisker = 0.967783},]coordinates{}; 
\addplot+ [color = lime!70!black,solid,boxplot prepared = {box extend=0.4, draw position = 11, lower whisker = 0.959326, lower quartile = 0.964103, median = 0.966723, upper quartile = 0.969318, upper whisker = 0.973596},]coordinates{}; 

\addplot+ [color = yellow!70!black,solid,boxplot prepared = {box extend=0.4, draw position = 1, lower whisker = 0.614643, lower quartile = 0.625454, median = 0.629678, upper quartile = 0.633610, upper whisker = 0.642933},]coordinates{}; 
\addplot+ [color = yellow!70!black,solid,boxplot prepared = {box extend=0.4, draw position = 2, lower whisker = 0.774973, lower quartile = 0.782826, median = 0.787657, upper quartile = 0.790773, upper whisker = 0.803171},]coordinates{}; 
\addplot+ [color = yellow!70!black,solid,boxplot prepared = {box extend=0.4, draw position = 3, lower whisker = 0.851263, lower quartile = 0.857680, median = 0.862585, upper quartile = 0.865482, upper whisker = 0.877369},]coordinates{}; 
\addplot+ [color = yellow!70!black,solid,boxplot prepared = {box extend=0.4, draw position = 4, lower whisker = 0.890544, lower quartile = 0.902372, median = 0.905525, upper quartile = 0.911209, upper whisker = 0.917202},]coordinates{}; 
\addplot+ [color = yellow!70!black,solid,boxplot prepared = {box extend=0.4, draw position = 5, lower whisker = 0.918113, lower quartile = 0.926994, median = 0.930621, upper quartile = 0.932897, upper whisker = 0.941251},]coordinates{}; 
\addplot+ [color = yellow!70!black,solid,boxplot prepared = {box extend=0.4, draw position = 6, lower whisker = 0.926186, lower quartile = 0.939059, median = 0.944117, upper quartile = 0.949088, upper whisker = 0.955444},]coordinates{}; 
\addplot+ [color = yellow!70!black,solid,boxplot prepared = {box extend=0.4, draw position = 7, lower whisker = 0.938230, lower quartile = 0.945303, median = 0.947669, upper quartile = 0.950496, upper whisker = 0.958903},]coordinates{}; 
\addplot+ [color = yellow!70!black,solid,boxplot prepared = {box extend=0.4, draw position = 8, lower whisker = 0.946412, lower quartile = 0.952007, median = 0.955525, upper quartile = 0.958435, upper whisker = 0.962932},]coordinates{}; 
\addplot+ [color = yellow!70!black,solid,boxplot prepared = {box extend=0.4, draw position = 9, lower whisker = 0.949517, lower quartile = 0.955718, median = 0.959361, upper quartile = 0.962694, upper whisker = 0.969487},]coordinates{}; 
\addplot+ [color = yellow!70!black,solid,boxplot prepared = {box extend=0.4, draw position = 10, lower whisker = 0.956616, lower quartile = 0.960896, median = 0.962680, upper quartile = 0.964558, upper whisker = 0.972669},]coordinates{}; 
\addplot+ [color = yellow!70!black,solid,boxplot prepared = {box extend=0.4, draw position = 11, lower whisker = 0.956660, lower quartile = 0.964738, median = 0.966981, upper quartile = 0.968442, upper whisker = 0.972266},]coordinates{}; 

\addplot+ [color = red!70!black,solid,boxplot prepared = {box extend=0.4, draw position = 1, lower whisker = 0.680514, lower quartile = 0.686651, median = 0.690038, upper quartile = 0.693975, upper whisker = 0.701583},]coordinates{}; 
\addplot+ [color = red!70!black,solid,boxplot prepared = {box extend=0.4, draw position = 2, lower whisker = 0.805258, lower quartile = 0.809236, median = 0.811469, upper quartile = 0.813139, upper whisker = 0.819814},]coordinates{}; 
\addplot+ [color = red!70!black,solid,boxplot prepared = {box extend=0.4, draw position = 3, lower whisker = 0.852233, lower quartile = 0.856354, median = 0.858914, upper quartile = 0.861127, upper whisker = 0.868929},]coordinates{}; 
\addplot+ [color = red!70!black,solid,boxplot prepared = {box extend=0.4, draw position = 4, lower whisker = 0.881934, lower quartile = 0.886633, median = 0.889181, upper quartile = 0.892313, upper whisker = 0.896038},]coordinates{}; 
\addplot+ [color = red!70!black,solid,boxplot prepared = {box extend=0.4, draw position = 5, lower whisker = 0.908956, lower quartile = 0.916114, median = 0.919376, upper quartile = 0.921264, upper whisker = 0.931096},]coordinates{}; 
\addplot+ [color = red!70!black,solid,boxplot prepared = {box extend=0.4, draw position = 6, lower whisker = 0.928296, lower quartile = 0.935104, median = 0.942069, upper quartile = 0.946054, upper whisker = 0.956155},]coordinates{}; 
\addplot+ [color = red!70!black,solid,boxplot prepared = {box extend=0.4, draw position = 7, lower whisker = 0.937665, lower quartile = 0.942019, median = 0.945583, upper quartile = 0.949181, upper whisker = 0.955683},]coordinates{}; 
\addplot+ [color = red!70!black,solid,boxplot prepared = {box extend=0.4, draw position = 8, lower whisker = 0.945851, lower quartile = 0.950943, median = 0.954033, upper quartile = 0.956352, upper whisker = 0.962383},]coordinates{}; 
\addplot+ [color = red!70!black,solid,boxplot prepared = {box extend=0.4, draw position = 9, lower whisker = 0.946687, lower quartile = 0.955077, median = 0.958788, upper quartile = 0.961060, upper whisker = 0.969218},]coordinates{}; 
\addplot+ [color = red!70!black,solid,boxplot prepared = {box extend=0.4, draw position = 10, lower whisker = 0.954956, lower quartile = 0.959738, median = 0.962722, upper quartile = 0.965188, upper whisker = 0.973058},]coordinates{}; 
\addplot+ [color = red!70!black,solid,boxplot prepared = {box extend=0.4, draw position = 11, lower whisker = 0.959119, lower quartile = 0.963684, median = 0.965516, upper quartile = 0.967531, upper whisker = 0.972349},]coordinates{}; 

\addplot+ [color = blue!50!black,solid,boxplot prepared = {box extend=0.4, draw position = 1, lower whisker = 0.578120, lower quartile = 0.583523, median = 0.585371, upper quartile = 0.588368, upper whisker = 0.594768},]coordinates{}; 
\addplot+ [color = blue!50!black,solid,boxplot prepared = {box extend=0.4, draw position = 2, lower whisker = 0.674330, lower quartile = 0.682455, median = 0.685261, upper quartile = 0.688337, upper whisker = 0.693348},]coordinates{}; 
\addplot+ [color = blue!50!black,solid,boxplot prepared = {box extend=0.4, draw position = 3, lower whisker = 0.721993, lower quartile = 0.729992, median = 0.733861, upper quartile = 0.737676, upper whisker = 0.758905},]coordinates{}; 
\addplot+ [color = blue!50!black,solid,boxplot prepared = {box extend=0.4, draw position = 4, lower whisker = 0.775265, lower quartile = 0.786859, median = 0.791750, upper quartile = 0.797196, upper whisker = 0.809965},]coordinates{}; 
\addplot+ [color = blue!50!black,solid,boxplot prepared = {box extend=0.4, draw position = 5, lower whisker = 0.815139, lower quartile = 0.835095, median = 0.845527, upper quartile = 0.853985, upper whisker = 0.884983},]coordinates{}; 
\addplot+ [color = blue!50!black,solid,boxplot prepared = {box extend=0.4, draw position = 6, lower whisker = 0.928179, lower quartile = 0.938846, median = 0.945421, upper quartile = 0.948944, upper whisker = 0.955840},]coordinates{}; 
\addplot+ [color = blue!50!black,solid,boxplot prepared = {box extend=0.4, draw position = 7, lower whisker = 0.931994, lower quartile = 0.945407, median = 0.948281, upper quartile = 0.952247, upper whisker = 0.957554},]coordinates{}; 
\addplot+ [color = blue!50!black,solid,boxplot prepared = {box extend=0.4, draw position = 8, lower whisker = 0.941164, lower quartile = 0.952538, median = 0.954399, upper quartile = 0.955750, upper whisker = 0.963303},]coordinates{}; 
\addplot+ [color = blue!50!black,solid,boxplot prepared = {box extend=0.4, draw position = 9, lower whisker = 0.953434, lower quartile = 0.958500, median = 0.960338, upper quartile = 0.961651, upper whisker = 0.964993},]coordinates{}; 
\addplot+ [color = blue!50!black,solid,boxplot prepared = {box extend=0.4, draw position = 10, lower whisker = 0.952112, lower quartile = 0.959760, median = 0.962246, upper quartile = 0.966463, upper whisker = 0.969811},]coordinates{}; 
\addplot+ [color = blue!50!black,solid,boxplot prepared = {box extend=0.4, draw position = 11, lower whisker = 0.959592, lower quartile = 0.965390, median = 0.967149, upper quartile = 0.968333, upper whisker = 0.972602},]coordinates{}; 

\end{axis}
\end{tikzpicture}
}
\end{minipage}
\ref{named}
\caption{Local segregation of  $1.01$-approximate networks in the \icfgame obtained by the best move dynamic for $n=1000$ over 50 runs starting from a random or segregated tree and grid.}
\label{plot:ICFNCG}
\end{figure}
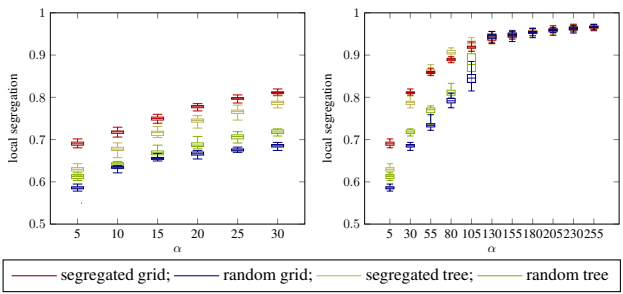

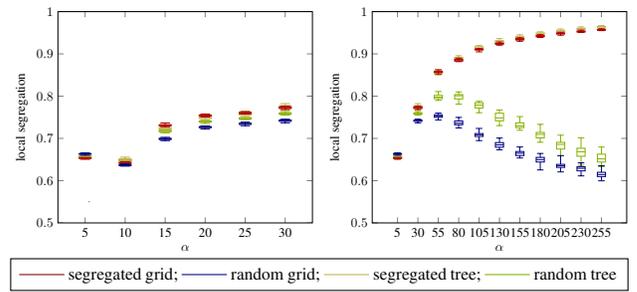
\begin{figure}[!ht]
\centering
\begin{minipage}{0.23\textwidth}
\resizebox {\textwidth} {!} {
\begin{tikzpicture}
\begin{axis}[
legend columns=-1,
legend entries={segregated grid;,random grid;,segregated tree;,random tree},
legend to name=named,
legend style={nodes={scale=0.65, transform shape}},
xlabel= {$\alpha$},
ylabel= {local segregation},
boxplot/draw direction=y,
baseline,
xtick = {1,2,3, 4, 5, 6, 7, 8},
xticklabels = {5, 10, 15, 20, 25, 30, 35, 40},
ymin=0.5,
ymax=1
]
\addplot[red!50!black, domain=1.1:1.11]{0.55};
\addplot[blue!50!black, domain=1.1:1.11]{0.55};
\addplot[yellow!70!black, domain=1.1:1.11]{0.55};
\addplot[lime!70!black, domain=1.1:1.11]{0.55};

\addplot+ [color = lime!70!black,solid,boxplot prepared = {box extend=0.3, upper whisker = 0.659879, upper quartile = 0.658397, median = 0.656753, lower quartile = 0.655413, lower whisker = 0.653990},]coordinates{}; 
\addplot+ [color = lime!70!black,solid,boxplot prepared = {box extend=0.3, upper whisker = 0.650925, upper quartile = 0.649870, median = 0.649138, lower quartile = 0.647020, lower whisker = 0.644806},]coordinates{}; 
\addplot+ [color = lime!70!black,solid,boxplot prepared = {box extend=0.3, upper whisker = 0.720179, upper quartile = 0.719365, median = 0.717668, lower quartile = 0.715807, lower whisker = 0.713312},]coordinates{}; 
\addplot+ [color = lime!70!black,solid,boxplot prepared = {box extend=0.3, upper whisker = 0.743078, upper quartile = 0.741271, median = 0.739565, lower quartile = 0.738483, lower whisker = 0.735386},]coordinates{}; 
\addplot+ [color = lime!70!black,solid,boxplot prepared = {box extend=0.3, upper whisker = 0.751736, upper quartile = 0.748995, median = 0.746719, lower quartile = 0.745677, lower whisker = 0.744522},]coordinates{}; 
\addplot+ [color = lime!70!black,solid,boxplot prepared = {box extend=0.3, upper whisker = 0.763583, upper quartile = 0.760466, median = 0.758692, lower quartile = 0.757232, lower whisker = 0.755161},]coordinates{}; 

\addplot+ [color = yellow!70!black,solid,boxplot prepared = {box extend=0.3, draw position = 1, lower whisker = 0.652161, lower quartile = 0.653350, median = 0.654066, upper quartile = 0.654914, upper whisker = 0.655529},]coordinates{}; 
\addplot+ [color = yellow!70!black,solid,boxplot prepared = {box extend=0.3, draw position = 2, lower whisker = 0.643005, lower quartile = 0.648401, median = 0.649864, upper quartile = 0.651919, upper whisker = 0.655596},]coordinates{}; 
\addplot+ [color = yellow!70!black,solid,boxplot prepared = {box extend=0.3, draw position = 3, lower whisker = 0.718538, lower quartile = 0.722290, median = 0.724476, upper quartile = 0.725825, upper whisker = 0.731051},]coordinates{}; 
\addplot+ [color = yellow!70!black,solid,boxplot prepared = {box extend=0.3, draw position = 4, lower whisker = 0.742324, lower quartile = 0.747015, median = 0.748845, upper quartile = 0.749996, upper whisker = 0.752475},]coordinates{}; 
\addplot+ [color = yellow!70!black,solid,boxplot prepared = {box extend=0.3, draw position = 5, lower whisker = 0.749928, lower quartile = 0.756005, median = 0.757036, upper quartile = 0.758665, upper whisker = 0.764659},]coordinates{}; 
\addplot+ [color = yellow!70!black,solid,boxplot prepared = {box extend=0.3, draw position = 6, lower whisker = 0.767564, lower quartile = 0.772087, median = 0.775743, upper quartile = 0.777537, upper whisker = 0.782297},]coordinates{}; 

\addplot+ [color = red!70!black,solid,boxplot prepared = {box extend=0.3, draw position = 1, lower whisker = 0.650482, lower quartile = 0.652644, median = 0.653581, upper quartile = 0.654087, upper whisker = 0.654632},]coordinates{}; 
\addplot+ [color = red!70!black,solid,boxplot prepared = {box extend=0.3, draw position = 2, lower whisker = 0.638210, lower quartile = 0.640454, median = 0.641485, upper quartile = 0.643782, upper whisker = 0.644536},]coordinates{}; 
\addplot+ [color = red!70!black,solid,boxplot prepared = {box extend=0.3, draw position = 3, lower whisker = 0.727639, lower quartile = 0.729764, median = 0.731083, upper quartile = 0.732452, upper whisker = 0.736530},]coordinates{}; 
\addplot+ [color = red!70!black,solid,boxplot prepared = {box extend=0.3, draw position = 4, lower whisker = 0.750945, lower quartile = 0.752831, median = 0.754032, upper quartile = 0.755560, upper whisker = 0.757696},]coordinates{}; 
\addplot+ [color = red!70!black,solid,boxplot prepared = {box extend=0.3, draw position = 5, lower whisker = 0.757637, lower quartile = 0.758648, median = 0.759543, upper quartile = 0.761889, upper whisker = 0.763217},]coordinates{}; 
\addplot+ [color = red!70!black,solid,boxplot prepared = {box extend=0.3, draw position = 6, lower whisker = 0.768352, lower quartile = 0.770891, median = 0.773063, upper quartile = 0.774361, upper whisker = 0.776514},]coordinates{}; 

\addplot+ [color = blue!50!black,solid,boxplot prepared = {box extend=0.3, draw position = 1, lower whisker = 0.660295, lower quartile = 0.662472, median = 0.663409, upper quartile = 0.664214, upper whisker = 0.665714},]coordinates{}; 
\addplot+ [color = blue!50!black,solid,boxplot prepared = {box extend=0.3, draw position = 2, lower whisker = 0.634319, lower quartile = 0.635633, median = 0.636894, upper quartile = 0.637588, upper whisker = 0.639591},]coordinates{}; 
\addplot+ [color = blue!50!black,solid,boxplot prepared = {box extend=0.3, draw position = 3, lower whisker = 0.694702, lower quartile = 0.697318, median = 0.698636, upper quartile = 0.700405, upper whisker = 0.702590},]coordinates{}; 
\addplot+ [color = blue!50!black,solid,boxplot prepared = {box extend=0.3, draw position = 4, lower whisker = 0.722026, lower quartile = 0.725352, median = 0.726957, upper quartile = 0.728018, upper whisker = 0.729656},]coordinates{}; 
\addplot+ [color = blue!50!black,solid,boxplot prepared = {box extend=0.3, draw position = 5, lower whisker = 0.729907, lower quartile = 0.733217, median = 0.734471, upper quartile = 0.735810, upper whisker = 0.738647},]coordinates{}; 
\addplot+ [color = blue!50!black,solid,boxplot prepared = {box extend=0.3, draw position = 6, lower whisker = 0.736867, lower quartile = 0.740440, median = 0.742355, upper quartile = 0.743409, upper whisker = 0.745493},]coordinates{}; 

\end{axis}
\end{tikzpicture}
}
\end{minipage}
\begin{minipage}{0.23\textwidth}
\resizebox {\textwidth} {!} {
\begin{tikzpicture}
\begin{axis}[
xlabel= {$\alpha$},
ylabel= {local segregation},
boxplot/draw direction=y,
baseline,
xtick = {1, 2, 3, ..., 11},
xticklabels  = {5, 30, 55, 80, 105, 130, 155, 180, 205, 230, 255},
ymin=0.5,
ymax=1
]

\addplot+ [color = lime!70!black,solid,boxplot prepared = {box extend=0.4, upper whisker = 0.659879, upper quartile = 0.658397, median = 0.656753, lower quartile = 0.655413, lower whisker = 0.653990},]coordinates{}; 
\addplot+ [color = lime!70!black,solid,boxplot prepared = {box extend=0.4, upper whisker = 0.763583, upper quartile = 0.760466, median = 0.758692, lower quartile = 0.757232, lower whisker = 0.755161},]coordinates{}; 
\addplot+ [color = lime!70!black,solid,boxplot prepared = {box extend=0.4, upper whisker = 0.810918, upper quartile = 0.801887, median = 0.797691, lower quartile = 0.794103, lower whisker = 0.789915},]coordinates{}; 
\addplot+ [color = lime!70!black,solid,boxplot prepared = {box extend=0.4, upper whisker = 0.810056, upper quartile = 0.803677, median = 0.800853, lower quartile = 0.793454, lower whisker = 0.781104},]coordinates{}; 
\addplot+ [color = lime!70!black,solid,boxplot prepared = {box extend=0.4, upper whisker = 0.787627, upper quartile = 0.784401, median = 0.778971, lower quartile = 0.772142, lower whisker = 0.760741},]coordinates{}; 
\addplot+ [color = lime!70!black,solid,boxplot prepared = {box extend=0.4, upper whisker = 0.767364, upper quartile = 0.759520, median = 0.748505, lower quartile = 0.741629, lower whisker = 0.730269},]coordinates{}; 
\addplot+ [color = lime!70!black,solid,boxplot prepared = {box extend=0.4, upper whisker = 0.751625, upper quartile = 0.736903, median = 0.729563, lower quartile = 0.725385, lower whisker = 0.719247},]coordinates{}; 
\addplot+ [color = lime!70!black,solid,boxplot prepared = {box extend=0.4, upper whisker = 0.733305, upper quartile = 0.714098, median = 0.709601, lower quartile = 0.701306, lower whisker = 0.690854},]coordinates{}; 
\addplot+ [color = lime!70!black,solid,boxplot prepared = {box extend=0.4, upper whisker = 0.707582, upper quartile = 0.690663, median = 0.685736, lower quartile = 0.674783, lower whisker = 0.658763},]coordinates{}; 
\addplot+ [color = lime!70!black,solid,boxplot prepared = {box extend=0.4, upper whisker = 0.701504, upper quartile = 0.676710, median = 0.668350, lower quartile = 0.657748, lower whisker = 0.636315},]coordinates{}; 
\addplot+ [color = lime!70!black,solid,boxplot prepared = {box extend=0.4, upper whisker = 0.679541, upper quartile = 0.662636, median = 0.651606, lower quartile = 0.644054, lower whisker = 0.633262},]coordinates{}; 

\addplot+ [color = yellow!70!black,solid,boxplot prepared = {box extend=0.4, draw position = 1, lower whisker = 0.652161, lower quartile = 0.653350, median = 0.654066, upper quartile = 0.654914, upper whisker = 0.655529},]coordinates{}; 
\addplot+ [color = yellow!70!black,solid,boxplot prepared = {box extend=0.4, draw position = 2, lower whisker = 0.767564, lower quartile = 0.772087, median = 0.775743, upper quartile = 0.777537, upper whisker = 0.782297},]coordinates{}; 
\addplot+ [color = yellow!70!black,solid,boxplot prepared = {box extend=0.4, draw position = 3, lower whisker = 0.850005, lower quartile = 0.852788, median = 0.854899, upper quartile = 0.857658, upper whisker = 0.862474},]coordinates{}; 
\addplot+ [color = yellow!70!black,solid,boxplot prepared = {box extend=0.4, draw position = 4, lower whisker = 0.883507, lower quartile = 0.887282, median = 0.889788, upper quartile = 0.891954, upper whisker = 0.895682},]coordinates{}; 
\addplot+ [color = yellow!70!black,solid,boxplot prepared = {box extend=0.4, draw position = 5, lower whisker = 0.909074, lower quartile = 0.910714, median = 0.913204, upper quartile = 0.917234, upper whisker = 0.919856},]coordinates{}; 
\addplot+ [color = yellow!70!black,solid,boxplot prepared = {box extend=0.4, draw position = 6, lower whisker = 0.923144, lower quartile = 0.925099, median = 0.927663, upper quartile = 0.931488, upper whisker = 0.936320},]coordinates{}; 
\addplot+ [color = yellow!70!black,solid,boxplot prepared = {box extend=0.4, draw position = 7, lower whisker = 0.932835, lower quartile = 0.936333, median = 0.939496, upper quartile = 0.941652, upper whisker = 0.945190},]coordinates{}; 
\addplot+ [color = yellow!70!black,solid,boxplot prepared = {box extend=0.4, draw position = 8, lower whisker = 0.940433, lower quartile = 0.944929, median = 0.947588, upper quartile = 0.949665, upper whisker = 0.952202},]coordinates{}; 
\addplot+ [color = yellow!70!black,solid,boxplot prepared = {box extend=0.4, draw position = 9, lower whisker = 0.948553, lower quartile = 0.951383, median = 0.953600, upper quartile = 0.955831, upper whisker = 0.958901},]coordinates{}; 
\addplot+ [color = yellow!70!black,solid,boxplot prepared = {box extend=0.4, draw position = 10, lower whisker = 0.953756, lower quartile = 0.957942, median = 0.959456, upper quartile = 0.960781, upper whisker = 0.963900},]coordinates{}; 
\addplot+ [color = yellow!70!black,solid,boxplot prepared = {box extend=0.4, draw position = 11, lower whisker = 0.957611, lower quartile = 0.961093, median = 0.963655, upper quartile = 0.964826, upper whisker = 0.966842},]coordinates{}; 

\addplot+ [color = blue!50!black,solid,boxplot prepared = {box extend=0.4, draw position = 1, lower whisker = 0.660295, lower quartile = 0.662472, median = 0.663409, upper quartile = 0.664214, upper whisker = 0.665714},]coordinates{}; 
\addplot+ [color = blue!50!black,solid,boxplot prepared = {box extend=0.4, draw position = 2, lower whisker = 0.736867, lower quartile = 0.740440, median = 0.742355, upper quartile = 0.743409, upper whisker = 0.745493},]coordinates{}; 
\addplot+ [color = blue!50!black,solid,boxplot prepared = {box extend=0.4, draw position = 3, lower whisker = 0.743763, lower quartile = 0.749882, median = 0.752613, upper quartile = 0.755885, upper whisker = 0.759802},]coordinates{}; 
\addplot+ [color = blue!50!black,solid,boxplot prepared = {box extend=0.4, draw position = 4, lower whisker = 0.724693, lower quartile = 0.732339, median = 0.736061, upper quartile = 0.741292, upper whisker = 0.749750},]coordinates{}; 
\addplot+ [color = blue!50!black,solid,boxplot prepared = {box extend=0.4, draw position = 5, lower whisker = 0.694786, lower quartile = 0.704600, median = 0.708528, upper quartile = 0.711441, upper whisker = 0.723636},]coordinates{}; 
\addplot+ [color = blue!50!black,solid,boxplot prepared = {box extend=0.4, draw position = 6, lower whisker = 0.673080, lower quartile = 0.679199, median = 0.684756, upper quartile = 0.690056, upper whisker = 0.700850},]coordinates{}; 
\addplot+ [color = blue!50!black,solid,boxplot prepared = {box extend=0.4, draw position = 7, lower whisker = 0.653669, lower quartile = 0.659405, median = 0.664931, upper quartile = 0.668912, upper whisker = 0.679791},]coordinates{}; 
\addplot+ [color = blue!50!black,solid,boxplot prepared = {box extend=0.4, draw position = 8, lower whisker = 0.625365, lower quartile = 0.644113, median = 0.649811, upper quartile = 0.655027, upper whisker = 0.664394},]coordinates{}; 
\addplot+ [color = blue!50!black,solid,boxplot prepared = {box extend=0.4, draw position = 9, lower whisker = 0.620683, lower quartile = 0.631375, median = 0.635065, upper quartile = 0.638523, upper whisker = 0.659369},]coordinates{}; 
\addplot+ [color = blue!50!black,solid,boxplot prepared = {box extend=0.4, draw position = 10, lower whisker = 0.611280, lower quartile = 0.623212, median = 0.629153, upper quartile = 0.632750, upper whisker = 0.642166},]coordinates{}; 
\addplot+ [color = blue!50!black,solid,boxplot prepared = {box extend=0.4, draw position = 11, lower whisker = 0.599656, lower quartile = 0.609701, median = 0.614495, upper quartile = 0.619795, upper whisker = 0.635213},]coordinates{}; 

\addplot+ [color = red!70!black,solid,boxplot prepared = {box extend=0.4, draw position = 1, lower whisker = 0.650482, lower quartile = 0.652644, median = 0.653581, upper quartile = 0.654087, upper whisker = 0.654632},]coordinates{}; 
\addplot+ [color = red!70!black,solid,boxplot prepared = {box extend=0.4, draw position = 2, lower whisker = 0.768352, lower quartile = 0.770891, median = 0.773063, upper quartile = 0.774361, upper whisker = 0.776514},]coordinates{}; 
\addplot+ [color = red!70!black,solid,boxplot prepared = {box extend=0.4, draw position = 3, lower whisker = 0.852445, lower quartile = 0.856174, median = 0.857146, upper quartile = 0.859243, upper whisker = 0.862548},]coordinates{}; 
\addplot+ [color = red!70!black,solid,boxplot prepared = {box extend=0.4, draw position = 4, lower whisker = 0.881620, lower quartile = 0.882974, median = 0.885742, upper quartile = 0.886685, upper whisker = 0.889856},]coordinates{}; 
\addplot+ [color = red!70!black,solid,boxplot prepared = {box extend=0.4, draw position = 5, lower whisker = 0.904777, lower quartile = 0.909308, median = 0.910435, upper quartile = 0.911631, upper whisker = 0.913487},]coordinates{}; 
\addplot+ [color = red!70!black,solid,boxplot prepared = {box extend=0.4, draw position = 6, lower whisker = 0.920637, lower quartile = 0.922934, median = 0.923998, upper quartile = 0.925277, upper whisker = 0.928442},]coordinates{}; 
\addplot+ [color = red!70!black,solid,boxplot prepared = {box extend=0.4, draw position = 7, lower whisker = 0.929886, lower quartile = 0.932895, median = 0.934411, upper quartile = 0.936067, upper whisker = 0.938366},]coordinates{}; 
\addplot+ [color = red!70!black,solid,boxplot prepared = {box extend=0.4, draw position = 8, lower whisker = 0.938868, lower quartile = 0.940714, median = 0.942132, upper quartile = 0.943119, upper whisker = 0.945638},]coordinates{}; 
\addplot+ [color = red!70!black,solid,boxplot prepared = {box extend=0.4, draw position = 9, lower whisker = 0.944206, lower quartile = 0.946983, median = 0.948622, upper quartile = 0.949469, upper whisker = 0.951244},]coordinates{}; 
\addplot+ [color = red!70!black,solid,boxplot prepared = {box extend=0.4, draw position = 10, lower whisker = 0.950679, lower quartile = 0.951785, median = 0.952754, upper quartile = 0.954063, upper whisker = 0.956140},]coordinates{}; 
\addplot+ [color = red!70!black,solid,boxplot prepared = {box extend=0.4, draw position = 11, lower whisker = 0.955072, lower quartile = 0.956095, median = 0.956832, upper quartile = 0.957558, upper whisker = 0.959508},]coordinates{}; 

\end{axis}
\end{tikzpicture}
}
\end{minipage}
\ref{named}
\caption{Local segregation of pairwise stable networks in the add-only \icfgame obtained by the best move dynamic for $n=1000$ over 50 runs starting from a random or segregated tree and grid.}
\label{plot:ICFNCG_AddOnly}
\end{figure}

\subsubsection*{Results for the \deigame}
For sake of comparison, we include the results for the local segregation measure for the \deigame again. The following two plots are identical to the respective plots in the main body of the paper.  
This shows clearly the similarities of the respective results. In particular, the tendency of decreasing segregation in case of the add-only version of the dynamics with random initial networks is observed for both games.

\begin{figure}[!ht]
\centering
\begin{minipage}{0.23\textwidth}
\resizebox {\textwidth} {!} {
\begin{tikzpicture}
\begin{axis}[
legend columns=-1,
legend entries={segregated grid;,random grid;,segregated tree;,random tree},
legend to name=named,
legend style={nodes={scale=0.65, transform shape}},
xlabel= {$\alpha$},
ylabel= {local segregation},
boxplot/draw direction=y,
baseline,
xtick = {1,2,3, 4, 5, 6, 7, 8},
xticklabels = {5, 10, 15, 20, 25, 30, 35, 40},
ymin=0.5,
ymax=1
]
\addplot[red!50!black, domain=1.1:1.11]{0.55};
\addplot[blue!50!black, domain=1.1:1.11]{0.55};
\addplot[yellow!70!black, domain=1.1:1.11]{0.55};
\addplot[lime!70!black, domain=1.1:1.11]{0.55};

\addplot+ [color = lime!70!black,solid,boxplot prepared = {box extend=0.3, draw position = 1, lower whisker = 0.587877, lower quartile = 0.595792, median = 0.599611, upper quartile = 0.600736, upper whisker = 0.600876},]coordinates{}; 
\addplot+ [color = lime!70!black,solid,boxplot prepared = {box extend=0.3, draw position = 2, lower whisker = 0.614817, lower quartile = 0.639993, median = 0.648220, upper quartile = 0.655185, upper whisker = 0.659965},]coordinates{}; 
\addplot+ [color = lime!70!black,solid,boxplot prepared = {box extend=0.3, draw position = 3, lower whisker = 0.672204, lower quartile = 0.680865, median = 0.686433, upper quartile = 0.691970, upper whisker = 0.704313},]coordinates{}; 
\addplot+ [color = lime!70!black,solid,boxplot prepared = {box extend=0.3, draw position = 4, lower whisker = 0.693762, lower quartile = 0.710872, median = 0.720265, upper quartile = 0.725355, upper whisker = 0.743811},]coordinates{}; 
\addplot+ [color = lime!70!black,solid,boxplot prepared = {box extend=0.3, draw position = 5, lower whisker = 0.725430, lower quartile = 0.738375, median = 0.745606, upper quartile = 0.752434, upper whisker = 0.769720},]coordinates{}; 
\addplot+ [color = lime!70!black,solid,boxplot prepared = {box extend=0.3, draw position = 6, lower whisker = 0.746930, lower quartile = 0.759186, median = 0.773552, upper quartile = 0.779921, upper whisker = 0.783034},]coordinates{}; 

\addplot+ [color = yellow!70!black,solid,boxplot prepared = {box extend=0.3, draw position = 1, lower whisker = 0.625312, lower quartile = 0.645360, median = 0.649989, upper quartile = 0.654888, upper whisker = 0.664035},]coordinates{}; 
\addplot+ [color = yellow!70!black,solid,boxplot prepared = {box extend=0.3, draw position = 2, lower whisker = 0.715601, lower quartile = 0.723320, median = 0.726892, upper quartile = 0.729812, upper whisker = 0.739498},]coordinates{}; 
\addplot+ [color = yellow!70!black,solid,boxplot prepared = {box extend=0.3, draw position = 3, lower whisker = 0.764764, lower quartile = 0.776444, median = 0.782217, upper quartile = 0.785451, upper whisker = 0.796037},]coordinates{}; 
\addplot+ [color = yellow!70!black,solid,boxplot prepared = {box extend=0.3, draw position = 4, lower whisker = 0.802477, lower quartile = 0.815618, median = 0.821320, upper quartile = 0.823649, upper whisker = 0.830700},]coordinates{}; 
\addplot+ [color = yellow!70!black,solid,boxplot prepared = {box extend=0.3, draw position = 5, lower whisker = 0.835204, lower quartile = 0.842199, median = 0.845808, upper quartile = 0.849384, upper whisker = 0.856944},]coordinates{}; 
\addplot+ [color = yellow!70!black,solid,boxplot prepared = {box extend=0.3, draw position = 6, lower whisker = 0.857273, lower quartile = 0.865692, median = 0.870443, upper quartile = 0.872591, upper whisker = 0.883314},]coordinates{}; 

\addplot+ [color = red!70!black,solid,boxplot prepared = {box extend=0.3, draw position = 1, lower whisker = 0.697566, lower quartile = 0.705443, median = 0.709720, upper quartile = 0.711501, upper whisker = 0.719408},]coordinates{}; 
\addplot+ [color = red!70!black,solid,boxplot prepared = {box extend=0.3, draw position = 2, lower whisker = 0.760929, lower quartile = 0.767087, median = 0.769896, upper quartile = 0.772985, upper whisker = 0.779684},]coordinates{}; 
\addplot+ [color = red!70!black,solid,boxplot prepared = {box extend=0.3, draw position = 3, lower whisker = 0.809113, lower quartile = 0.817723, median = 0.820637, upper quartile = 0.822316, upper whisker = 0.832187},]coordinates{}; 
\addplot+ [color = red!70!black,solid,boxplot prepared = {box extend=0.3, draw position = 4, lower whisker = 0.841635, lower quartile = 0.849097, median = 0.851167, upper quartile = 0.853647, upper whisker = 0.859318},]coordinates{}; 
\addplot+ [color = red!70!black,solid,boxplot prepared = {box extend=0.3, draw position = 5, lower whisker = 0.863108, lower quartile = 0.869798, median = 0.871972, upper quartile = 0.874686, upper whisker = 0.878572},]coordinates{}; 
\addplot+ [color = red!70!black,solid,boxplot prepared = {box extend=0.3, draw position = 6, lower whisker = 0.880598, lower quartile = 0.884663, median = 0.886926, upper quartile = 0.888980, upper whisker = 0.895380},]coordinates{}; 

\addplot+ [color = blue!50!black,solid,boxplot prepared = {box extend=0.3, draw position = 1, lower whisker = 0.553215, lower quartile = 0.567747, median = 0.572520, upper quartile = 0.575852, upper whisker = 0.583499},]coordinates{}; 
\addplot+ [color = blue!50!black,solid,boxplot prepared = {box extend=0.3, draw position = 2, lower whisker = 0.605987, lower quartile = 0.612051, median = 0.617258, upper quartile = 0.620476, upper whisker = 0.633279},]coordinates{}; 
\addplot+ [color = blue!50!black,solid,boxplot prepared = {box extend=0.3, draw position = 3, lower whisker = 0.634162, lower quartile = 0.641935, median = 0.648416, upper quartile = 0.651590, upper whisker = 0.659823},]coordinates{}; 
\addplot+ [color = blue!50!black,solid,boxplot prepared = {box extend=0.3, draw position = 4, lower whisker = 0.638457, lower quartile = 0.661246, median = 0.664685, upper quartile = 0.669107, upper whisker = 0.674004},]coordinates{}; 
\addplot+ [color = blue!50!black,solid,boxplot prepared = {box extend=0.3, draw position = 5, lower whisker = 0.660377, lower quartile = 0.670119, median = 0.677976, upper quartile = 0.683920, upper whisker = 0.692616},]coordinates{}; 
\addplot+ [color = blue!50!black,solid,boxplot prepared = {box extend=0.3, draw position = 6, lower whisker = 0.671119, lower quartile = 0.684051, median = 0.690288, upper quartile = 0.695877, upper whisker = 0.711555},]coordinates{}; 

\end{axis}
\end{tikzpicture}
}
\end{minipage}
\begin{minipage}{0.23\textwidth}
\resizebox {\textwidth} {!} {
\begin{tikzpicture}
\begin{axis}[
xlabel= {$\alpha$},
ylabel= {local segregation},
boxplot/draw direction=y,
baseline,
xtick = {1, 2, 3, ..., 11},
xticklabels  = {5, 30, 55, 80, 105, 130, 155, 180, 205, 230, 255},
ymin=0.5,
ymax=1
]

\addplot+ [color = lime!70!black,solid,boxplot prepared = {box extend=0.4, draw position = 1, lower whisker = 0.587877, lower quartile = 0.595792, median = 0.599611, upper quartile = 0.600736, upper whisker = 0.600876},]coordinates{}; 
\addplot+ [color = lime!70!black,solid,boxplot prepared = {box extend=0.4, draw position = 2, lower whisker = 0.746930, lower quartile = 0.759186, median = 0.773552, upper quartile = 0.779921, upper whisker = 0.783034},]coordinates{}; 
\addplot+ [color = lime!70!black,solid,boxplot prepared = {box extend=0.4, draw position = 3, lower whisker = 0.873544, lower quartile = 0.875830, median = 0.878203, upper quartile = 0.882616, upper whisker = 0.884611},]coordinates{}; 
\addplot+ [color = lime!70!black,solid,boxplot prepared = {box extend=0.4, draw position = 4, lower whisker = 0.923530, lower quartile = 0.925024, median = 0.930053, upper quartile = 0.931133, upper whisker = 0.935467},]coordinates{}; 
\addplot+ [color = lime!70!black,solid,boxplot prepared = {box extend=0.4, draw position = 5, lower whisker = 0.957246, lower quartile = 0.959246, median = 0.959575, upper quartile = 0.961969, upper whisker = 0.962539},]coordinates{}; 
\addplot+ [color = lime!70!black,solid,boxplot prepared = {box extend=0.4, draw position = 6, lower whisker = 0.961707, lower quartile = 0.961947, median = 0.964808, upper quartile = 0.967901, upper whisker = 0.970280},]coordinates{}; 
\addplot+ [color = lime!70!black,solid,boxplot prepared = {box extend=0.4, draw position = 7, lower whisker = 0.968187, lower quartile = 0.968823, median = 0.972421, upper quartile = 0.974994, upper whisker = 0.975260},]coordinates{}; 
\addplot+ [color = lime!70!black,solid,boxplot prepared = {box extend=0.4, draw position = 8, lower whisker = 0.969983, lower quartile = 0.973441, median = 0.975758, upper quartile = 0.977213, upper whisker = 0.978670},]coordinates{}; 
\addplot+ [color = lime!70!black,solid,boxplot prepared = {box extend=0.4, draw position = 9, lower whisker = 0.973462, lower quartile = 0.976978, median = 0.978195, upper quartile = 0.978944, upper whisker = 0.978955},]coordinates{}; 
\addplot+ [color = lime!70!black,solid,boxplot prepared = {box extend=0.4, draw position = 10, lower whisker = 0.977007, lower quartile = 0.977624, median = 0.978862, upper quartile = 0.980195, upper whisker = 0.981108},]coordinates{}; 
\addplot+ [color = lime!70!black,solid,boxplot prepared = {box extend=0.4, draw position = 11, lower whisker = 0.978110, lower quartile = 0.980106, median = 0.982154, upper quartile = 0.982968, upper whisker = 0.983385},]coordinates{}; 

\addplot+ [color = yellow!70!black,solid,boxplot prepared = {box extend=0.4, draw position = 1, lower whisker = 0.625312, lower quartile = 0.645360, median = 0.649989, upper quartile = 0.654888, upper whisker = 0.664035},]coordinates{}; 
\addplot+ [color = yellow!70!black,solid,boxplot prepared = {box extend=0.4, draw position = 2, lower whisker = 0.857273, lower quartile = 0.865692, median = 0.870443, upper quartile = 0.872591, upper whisker = 0.883314},]coordinates{}; 
\addplot+ [color = yellow!70!black,solid,boxplot prepared = {box extend=0.4, draw position = 3, lower whisker = 0.923923, lower quartile = 0.926922, median = 0.930204, upper quartile = 0.933215, upper whisker = 0.941737},]coordinates{}; 
\addplot+ [color = yellow!70!black,solid,boxplot prepared = {box extend=0.4, draw position = 4, lower whisker = 0.949148, lower quartile = 0.951784, median = 0.954276, upper quartile = 0.955444, upper whisker = 0.959144},]coordinates{}; 
\addplot+ [color = yellow!70!black,solid,boxplot prepared = {box extend=0.4, draw position = 5, lower whisker = 0.952666, lower quartile = 0.958743, median = 0.962441, upper quartile = 0.964854, upper whisker = 0.971552},]coordinates{}; 
\addplot+ [color = yellow!70!black,solid,boxplot prepared = {box extend=0.4, draw position = 6, lower whisker = 0.959560, lower quartile = 0.963026, median = 0.964291, upper quartile = 0.966124, upper whisker = 0.970617},]coordinates{}; 
\addplot+ [color = yellow!70!black,solid,boxplot prepared = {box extend=0.4, draw position = 7, lower whisker = 0.966569, lower quartile = 0.968616, median = 0.970596, upper quartile = 0.972148, upper whisker = 0.975831},]coordinates{}; 
\addplot+ [color = yellow!70!black,solid,boxplot prepared = {box extend=0.4, draw position = 8, lower whisker = 0.969493, lower quartile = 0.973342, median = 0.974507, upper quartile = 0.975862, upper whisker = 0.979846},]coordinates{}; 
\addplot+ [color = yellow!70!black,solid,boxplot prepared = {box extend=0.4, draw position = 9, lower whisker = 0.973960, lower quartile = 0.975806, median = 0.977811, upper quartile = 0.978961, upper whisker = 0.982225},]coordinates{}; 
\addplot+ [color = yellow!70!black,solid,boxplot prepared = {box extend=0.4, draw position = 10, lower whisker = 0.974930, lower quartile = 0.978464, median = 0.979766, upper quartile = 0.980507, upper whisker = 0.983006},]coordinates{}; 
\addplot+ [color = yellow!70!black,solid,boxplot prepared = {box extend=0.4, draw position = 11, lower whisker = 0.977565, lower quartile = 0.980369, median = 0.982019, upper quartile = 0.983316, upper whisker = 0.985345},]coordinates{}; 

\addplot+ [color = blue!50!black,solid,boxplot prepared = {box extend=0.4, draw position = 1, lower whisker = 0.553215, lower quartile = 0.567747, median = 0.572520, upper quartile = 0.575852, upper whisker = 0.583499},]coordinates{}; 
\addplot+ [color = blue!50!black,solid,boxplot prepared = {box extend=0.4, draw position = 2, lower whisker = 0.671119, lower quartile = 0.684051, median = 0.690288, upper quartile = 0.695877, upper whisker = 0.711555},]coordinates{}; 
\addplot+ [color = blue!50!black,solid,boxplot prepared = {box extend=0.4, draw position = 3, lower whisker = 0.769120, lower quartile = 0.792679, median = 0.799753, upper quartile = 0.805371, upper whisker = 0.824674},]coordinates{}; 
\addplot+ [color = blue!50!black,solid,boxplot prepared = {box extend=0.4, draw position = 4, lower whisker = 0.902491, lower quartile = 0.911603, median = 0.915007, upper quartile = 0.916627, upper whisker = 0.923923},]coordinates{}; 
\addplot+ [color = blue!50!black,solid,boxplot prepared = {box extend=0.4, draw position = 5, lower whisker = 0.948760, lower quartile = 0.953995, median = 0.957272, upper quartile = 0.959069, upper whisker = 0.964313},]coordinates{}; 
\addplot+ [color = blue!50!black,solid,boxplot prepared = {box extend=0.4, draw position = 6, lower whisker = 0.959348, lower quartile = 0.962614, median = 0.965023, upper quartile = 0.966140, upper whisker = 0.971082},]coordinates{}; 
\addplot+ [color = blue!50!black,solid,boxplot prepared = {box extend=0.4, draw position = 7, lower whisker = 0.966164, lower quartile = 0.969773, median = 0.971402, upper quartile = 0.973260, upper whisker = 0.977559},]coordinates{}; 
\addplot+ [color = blue!50!black,solid,boxplot prepared = {box extend=0.4, draw position = 8, lower whisker = 0.968674, lower quartile = 0.974483, median = 0.975671, upper quartile = 0.976556, upper whisker = 0.980299},]coordinates{}; 
\addplot+ [color = blue!50!black,solid,boxplot prepared = {box extend=0.4, draw position = 9, lower whisker = 0.974160, lower quartile = 0.976849, median = 0.978065, upper quartile = 0.978941, upper whisker = 0.981402},]coordinates{}; 
\addplot+ [color = blue!50!black,solid,boxplot prepared = {box extend=0.4, draw position = 10, lower whisker = 0.976702, lower quartile = 0.978616, median = 0.979923, upper quartile = 0.980969, upper whisker = 0.983412},]coordinates{}; 
\addplot+ [color = blue!50!black,solid,boxplot prepared = {box extend=0.4, draw position = 11, lower whisker = 0.977296, lower quartile = 0.980610, median = 0.981960, upper quartile = 0.983154, upper whisker = 0.984985},]coordinates{}; 

\addplot+ [color = red!70!black,solid,boxplot prepared = {box extend=0.4, draw position = 1, lower whisker = 0.697566, lower quartile = 0.705443, median = 0.709720, upper quartile = 0.711501, upper whisker = 0.719408},]coordinates{}; 
\addplot+ [color = red!70!black,solid,boxplot prepared = {box extend=0.4, draw position = 2, lower whisker = 0.880598, lower quartile = 0.884663, median = 0.886926, upper quartile = 0.888980, upper whisker = 0.895380},]coordinates{}; 
\addplot+ [color = red!70!black,solid,boxplot prepared = {box extend=0.4, draw position = 3, lower whisker = 0.927817, lower quartile = 0.930200, median = 0.931157, upper quartile = 0.933681, upper whisker = 0.937793},]coordinates{}; 
\addplot+ [color = red!70!black,solid,boxplot prepared = {box extend=0.4, draw position = 4, lower whisker = 0.947143, lower quartile = 0.949648, median = 0.950862, upper quartile = 0.952350, upper whisker = 0.955776},]coordinates{}; 
\addplot+ [color = red!70!black,solid,boxplot prepared = {box extend=0.4, draw position = 5, lower whisker = 0.952846, lower quartile = 0.958263, median = 0.962840, upper quartile = 0.963716, upper whisker = 0.966948},]coordinates{}; 
\addplot+ [color = red!70!black,solid,boxplot prepared = {box extend=0.4, draw position = 6, lower whisker = 0.956431, lower quartile = 0.963584, median = 0.964688, upper quartile = 0.966167, upper whisker = 0.970800},]coordinates{}; 
\addplot+ [color = red!70!black,solid,boxplot prepared = {box extend=0.4, draw position = 7, lower whisker = 0.964810, lower quartile = 0.969006, median = 0.971059, upper quartile = 0.972356, upper whisker = 0.975870},]coordinates{}; 
\addplot+ [color = red!70!black,solid,boxplot prepared = {box extend=0.4, draw position = 8, lower whisker = 0.969395, lower quartile = 0.973264, median = 0.974714, upper quartile = 0.975619, upper whisker = 0.978041},]coordinates{}; 
\addplot+ [color = red!70!black,solid,boxplot prepared = {box extend=0.4, draw position = 9, lower whisker = 0.972500, lower quartile = 0.976252, median = 0.977093, upper quartile = 0.978713, upper whisker = 0.981757},]coordinates{}; 
\addplot+ [color = red!70!black,solid,boxplot prepared = {box extend=0.4, draw position = 10, lower whisker = 0.975554, lower quartile = 0.978756, median = 0.979826, upper quartile = 0.980922, upper whisker = 0.985774},]coordinates{}; 
\addplot+ [color = red!70!black,solid,boxplot prepared = {box extend=0.4, draw position = 11, lower whisker = 0.978110, lower quartile = 0.980720, median = 0.981614, upper quartile = 0.982151, upper whisker = 0.984563},]coordinates{}; 

\end{axis}
\end{tikzpicture}
}
\end{minipage}
\ref{named}
\caption{Local segregation of  $1.01$-approximate networks in the \deigame obtained by the best move dynamic for $n=1000$ over 50 runs starting from a random or segregated tree and grid.}
\label{plot:DEINCG_appendix}
\end{figure}
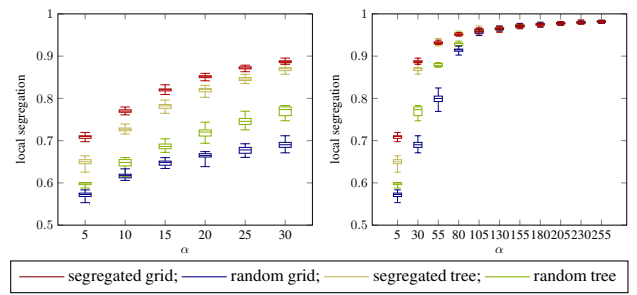

\begin{figure}[!ht]
\centering
\begin{minipage}{0.23\textwidth}
\resizebox {\textwidth} {!} {
\begin{tikzpicture}
\begin{axis}[
		legend columns=-1,
          legend entries={segregated grid;,random grid;,segregated tree;,random tree},
          legend to name=named,
		legend style={nodes={scale=0.65, transform shape}},
		xlabel= {$\alpha$},
		ylabel= {local segregation},
		boxplot/draw direction=y,
		baseline,
		xtick = {1,2,3, 4, 5, 6, 7, 8},
		xticklabels = {5, 10, 15, 20, 25, 30, 35, 40},
		ymin=0.5,
    		ymax=1
	]
\addplot[red!50!black, domain=1.1:1.2]{0.55};
\addplot[blue!50!black, domain=1.1:1.2]{0.55};
\addplot[yellow!70!black, domain=1.1:1.2]{0.55};
\addplot[lime!70!black, domain=1.1:1.2]{0.55};

\addplot+ [color = lime!70!black,solid,boxplot prepared = {box extend=0.3,  upper whisker = 0.569317, upper quartile = 0.566948, median = 0.565704, lower quartile = 0.563991, lower whisker = 0.561343},]coordinates{}; 
\addplot+ [color = lime!70!black,solid,boxplot prepared = {box extend=0.3, draw position = 2, lower whisker = 0.656054, lower quartile = 0.660461, median = 0.664114, upper quartile = 0.665906, upper whisker = 0.670238},]coordinates{}; 
\addplot+ [color = lime!70!black,solid,boxplot prepared = {box extend=0.3, draw position = 3, lower whisker = 0.761166, lower quartile = 0.765647, median = 0.767783, upper quartile = 0.770015, upper whisker = 0.773884},]coordinates{}; 
\addplot+ [color = lime!70!black,solid,boxplot prepared = {box extend=0.3, draw position = 4, lower whisker = 0.786826, lower quartile = 0.791560, median = 0.793824, upper quartile = 0.795343, upper whisker = 0.803133},]coordinates{}; 
\addplot+ [color = lime!70!black,solid,boxplot prepared = {box extend=0.3, draw position = 5, lower whisker = 0.796294, lower quartile = 0.804001, median = 0.806804, upper quartile = 0.808822, upper whisker = 0.812360},]coordinates{}; 
\addplot+ [color = lime!70!black,solid,boxplot prepared = {box extend=0.3, draw position = 6,  upper whisker = 0.815487, upper quartile = 0.810792, median = 0.807471, lower quartile = 0.804435, lower whisker = 0.798459},]coordinates{}; 

\addplot+ [color = yellow!70!black,solid,boxplot prepared = {box extend=0.3, draw position = 1, lower whisker = 0.571288, lower quartile = 0.574302, median = 0.575457, upper quartile = 0.576574, upper whisker = 0.577170},]coordinates{}; 
\addplot+ [color = yellow!70!black,solid,boxplot prepared = {box extend=0.3, draw position = 2, lower whisker = 0.669629, lower quartile = 0.680519, median = 0.682716, upper quartile = 0.685042, upper whisker = 0.693209},]coordinates{}; 
\addplot+ [color = yellow!70!black,solid,boxplot prepared = {box extend=0.3, draw position = 3, lower whisker = 0.795463, lower quartile = 0.799373, median = 0.800758, upper quartile = 0.802217, upper whisker = 0.804513},]coordinates{}; 
\addplot+ [color = yellow!70!black,solid,boxplot prepared = {box extend=0.3, draw position = 4, lower whisker = 0.823144, lower quartile = 0.826224, median = 0.828067, upper quartile = 0.830107, upper whisker = 0.833727},]coordinates{}; 
\addplot+ [color = yellow!70!black,solid,boxplot prepared = {box extend=0.3, draw position = 5, lower whisker = 0.838669, lower quartile = 0.844115, median = 0.845808, upper quartile = 0.848215, upper whisker = 0.855670},]coordinates{}; 
\addplot+ [color = yellow!70!black,solid,boxplot prepared = {box extend=0.3, draw position = 6, lower whisker = 0.864458, lower quartile = 0.870431, median = 0.873253, upper quartile = 0.875989, upper whisker = 0.882930},]coordinates{};

\addplot+ [color = red!70!black,solid,boxplot prepared = {box extend=0.3, draw position = 1, lower whisker = 0.573521, lower quartile = 0.575512, median = 0.576715, upper quartile = 0.577702, upper whisker = 0.579620},]coordinates{}; 
\addplot+ [color = red!70!black,solid,boxplot prepared = {box extend=0.3, draw position = 2, lower whisker = 0.659204, lower quartile = 0.665023, median = 0.667811, upper quartile = 0.671121, upper whisker = 0.679496},]coordinates{}; 
\addplot+ [color = red!70!black,solid,boxplot prepared = {box extend=0.3, draw position = 3, lower whisker = 0.801244, lower quartile = 0.805434, median = 0.807576, upper quartile = 0.808680, upper whisker = 0.810356},]coordinates{}; 
\addplot+ [color = red!70!black,solid,boxplot prepared = {box extend=0.3, draw position = 4, lower whisker = 0.827337, lower quartile = 0.830322, median = 0.832444, upper quartile = 0.833859, upper whisker = 0.835993},]coordinates{}; 
\addplot+ [color = red!70!black,solid,boxplot prepared = {box extend=0.3, draw position = 5, lower whisker = 0.843377, lower quartile = 0.844394, median = 0.846362, upper quartile = 0.848605, upper whisker = 0.852728},]coordinates{}; 
\addplot+ [color = red!70!black,solid,boxplot prepared = {box extend=0.3, draw position = 6, lower whisker = 0.859387, lower quartile = 0.870538, median = 0.872640, upper quartile = 0.874923, upper whisker = 0.879567},]coordinates{}; 

\addplot+ [color = blue!50!black,solid, boxplot prepared = {box extend=0.3, draw position = 1, upper whisker = 0.558332, upper quartile = 0.556623, median = 0.555612, lower quartile = 0.553640, lower whisker = 0.551390},]coordinates{}; 
\addplot+ [color = blue!50!black,solid,boxplot prepared = {box extend=0.3, draw position = 2, lower whisker = 0.622430, lower quartile = 0.625420, median = 0.626853, upper quartile = 0.628278, upper whisker = 0.632146},]coordinates{}; 
\addplot+ [color = blue!50!black,solid,boxplot prepared = {box extend=0.3, draw position = 3, lower whisker = 0.721274, lower quartile = 0.723054, median = 0.725273, upper quartile = 0.727901, upper whisker = 0.734792},]coordinates{}; 
\addplot+ [color = blue!50!black,solid,boxplot prepared = {box extend=0.3, draw position = 4, lower whisker = 0.752463, lower quartile = 0.758091, median = 0.760399, upper quartile = 0.762618, upper whisker = 0.765219},]coordinates{}; 
\addplot+ [color = blue!50!black,solid,boxplot prepared = {box extend=0.3, draw position = 5, lower whisker = 0.767886, lower quartile = 0.769114, median = 0.769849, upper quartile = 0.772662, upper whisker = 0.777153},]coordinates{}; 
\addplot+ [color = blue!50!black,solid,boxplot prepared = {box extend=0.3, draw position = 6, upper whisker = 0.786549, upper quartile = 0.781019, median = 0.778677, lower quartile = 0.776975, lower whisker = 0.773675},]coordinates{}; 

\end{axis}
\end{tikzpicture}
}
\end{minipage}
\begin{minipage}{0.23\textwidth}
\resizebox {\textwidth} {!} {
\begin{tikzpicture}
\begin{axis}[
		xlabel= {$\alpha$},
		ylabel= {local segregation},
		boxplot/draw direction=y,
		baseline,
		xtick = {1, 2, 3, ..., 11},
		xticklabels  = {5, 30, 55, 80, 105, 130, 155, 180, 205, 230, 255},
		ymin=0.5,
    		ymax=1
	]

\addplot+ [color = lime!70!black,solid,boxplot prepared = {box extend=0.4,  upper whisker = 0.569317, upper quartile = 0.566948, median = 0.565704, lower quartile = 0.563991, lower whisker = 0.561343},]coordinates{};

\addplot+ [color = lime!70!black,solid,boxplot prepared = {box extend=0.4,  upper whisker = 0.815487, upper quartile = 0.810792, median = 0.807471, lower quartile = 0.804435, lower whisker = 0.798459},]coordinates{}; 

\addplot+ [color = lime!70!black,solid,boxplot prepared = {box extend=0.4,  upper whisker =0.863297, upper quartile = 0.859855, median = 0.855878, lower quartile = 0.852971, lower whisker = 0.848393},]coordinates{}; 

\addplot+ [color = lime!70!black,solid,boxplot prepared = {box extend=0.4,  upper whisker =0.844602, upper quartile = 0.843130, median =0.839602, lower quartile = 0.832798, lower whisker = 0.810987},]coordinates{};

\addplot+ [color = lime!70!black,solid,boxplot prepared = {box extend=0.4,  upper whisker =0.814326, upper quartile = 0.807044, median = 0.799180, lower quartile = 0.795927, lower whisker = 0.782402},]coordinates{};

\addplot+ [color = lime!70!black,solid,boxplot prepared = {box extend=0.4,  upper whisker =0.782966, upper quartile = 0.773236, median = 0.767630, lower quartile = 0.760762, lower whisker = 0.751039},]coordinates{};

\addplot+ [color = lime!70!black,solid,boxplot prepared = {box extend=0.4,  upper whisker =0.762613, upper quartile = 0.741276, median = 0.733168, lower quartile = 0.723363, lower whisker = 0.714112},]coordinates{};

\addplot+ [color = lime!70!black,solid,boxplot prepared = {box extend=0.4,  upper whisker =0.732723, upper quartile = 0.719416, median = 0.707596, lower quartile = 0.700607, lower whisker = 0.686914},]coordinates{};

\addplot+ [color = lime!70!black,solid,boxplot prepared = {box extend=0.4, upper whisker =0.705276, upper quartile = 0.697771, median = 0.684921, lower quartile = 0.678178, lower whisker = 0.652036},]coordinates{};

\addplot+ [color = lime!70!black,solid,boxplot prepared = {box extend=0.4,  upper whisker =0.692432, upper quartile = 0.680535, median = 0.672003, lower quartile = 0.656411, lower whisker = 0.637509},]coordinates{};

\addplot+ [color = lime!70!black,solid,boxplot prepared = {box extend=0.4,  upper whisker =0.668041, upper quartile = 0.651777, median = 0.645223, lower quartile = 0.638358, lower whisker = 0.615563},]coordinates{};

\addplot+ [color = yellow!70!black,solid,boxplot prepared = {box extend=0.4, draw position = 1, lower whisker = 0.571288, lower quartile = 0.574302, median = 0.575457, upper quartile = 0.576574, upper whisker = 0.577170},]coordinates{}; 
\addplot+ [color = yellow!70!black,solid,boxplot prepared = {box extend=0.4, draw position = 2, lower whisker = 0.864458, lower quartile = 0.870431, median = 0.873253, upper quartile = 0.875989, upper whisker = 0.882930},]coordinates{}; 
\addplot+ [color = yellow!70!black,solid,boxplot prepared = {box extend=0.4, draw position = 3, lower whisker = 0.930440, lower quartile = 0.933701, median = 0.935299, upper quartile = 0.938208, upper whisker = 0.939951},]coordinates{}; 
\addplot+ [color = yellow!70!black,solid,boxplot prepared = {box extend=0.4, draw position = 4, lower whisker = 0.950209, lower quartile = 0.953446, median = 0.955660, upper quartile = 0.956451, upper whisker = 0.961252},]coordinates{}; 
\addplot+ [color = yellow!70!black,solid,boxplot prepared = {box extend=0.4, draw position = 5, lower whisker = 0.959553, lower quartile = 0.963507, median = 0.964870, upper quartile = 0.966410, upper whisker = 0.969313},]coordinates{}; 
\addplot+ [color = yellow!70!black,solid,boxplot prepared = {box extend=0.4, draw position = 6, lower whisker = 0.970464, lower quartile = 0.971915, median = 0.972507, upper quartile = 0.973634, upper whisker = 0.974675},]coordinates{}; 
\addplot+ [color = yellow!70!black,solid,boxplot prepared = {box extend=0.4, draw position = 7, lower whisker = 0.973479, lower quartile = 0.975701, median = 0.976698, upper quartile = 0.977827, upper whisker = 0.980390},]coordinates{}; 
\addplot+ [color = yellow!70!black,solid,boxplot prepared = {box extend=0.4, draw position = 8, lower whisker = 0.974726, lower quartile = 0.979194, median = 0.980581, upper quartile = 0.981786, upper whisker = 0.983989},]coordinates{}; 
\addplot+ [color = yellow!70!black,solid,boxplot prepared = {box extend=0.4, draw position = 9, lower whisker = 0.980286, lower quartile = 0.981267, median = 0.982301, upper quartile = 0.983633, upper whisker = 0.985189},]coordinates{}; 
\addplot+ [color = yellow!70!black,solid,boxplot prepared = {box extend=0.4, draw position = 10, lower whisker = 0.982389, lower quartile = 0.984324, median = 0.984832, upper quartile = 0.985335, upper whisker = 0.986870},]coordinates{}; 
\addplot+ [color = yellow!70!black,solid,boxplot prepared = {box extend=0.4, draw position = 11, lower whisker = 0.983601, lower quartile = 0.985197, median = 0.985953, upper quartile = 0.986788, upper whisker = 0.988917},]coordinates{}; 

\addplot+ [color = blue!50!black,solid, boxplot prepared = {box extend=0.4, draw position = 1, upper whisker = 0.558332, upper quartile = 0.556623, median = 0.555612, lower quartile = 0.553640, lower whisker = 0.551390},]coordinates{}; 
\addplot+ [color = blue!50!black,solid,boxplot prepared = {box extend=0.4, draw position = 2, upper whisker = 0.786549, upper quartile = 0.781019, median = 0.778677, lower quartile = 0.776975, lower whisker = 0.773675},]coordinates{}; 
\addplot+ [color = blue!50!black,solid,boxplot prepared = {box extend=0.4, draw position = 3, upper whisker = 0.761652, upper quartile = 0.747456, median = 0.744271, lower quartile = 0.740453, lower whisker = 0.723683},]coordinates{}; 
\addplot+ [color = blue!50!black,solid,boxplot prepared = {box extend=0.4, draw position = 4, upper whisker = 0.724669, upper quartile = 0.716603, median = 0.712547, lower quartile = 0.708569, lower whisker = 0.696026},]coordinates{}; 
\addplot+ [color = blue!50!black,solid,boxplot prepared = {box extend=0.4, draw position = 5, upper whisker = 0.695507, upper quartile = 0.677721, median = 0.670677, lower quartile = 0.666027, lower whisker = 0.652375},]coordinates{}; 
\addplot+ [color = blue!50!black,solid,boxplot prepared = {box extend=0.4, draw position = 6, upper whisker = 0.658490, upper quartile = 0.653778, median = 0.643803, lower quartile = 0.637674, lower whisker = 0.631237},]coordinates{}; 
\addplot+ [color = blue!50!black,solid,boxplot prepared = {box extend=0.4, draw position = 7, upper whisker = 0.646484, upper quartile = 0.636930, median = 0.628301, lower quartile = 0.623119, lower whisker = 0.603993},]coordinates{}; 
\addplot+ [color = blue!50!black,solid,boxplot prepared = {box extend=0.4, draw position = 8, upper whisker = 0.622749, upper quartile = 0.612773, median = 0.606101, lower quartile = 0.597162, lower whisker = 0.586179},]coordinates{}; 
\addplot+ [color = blue!50!black,solid,boxplot prepared = {box extend=0.4, draw position = 9, upper whisker = 0.610267, upper quartile = 0.605699, median = 0.592222, lower quartile = 0.587106, lower whisker = 0.580332},]coordinates{}; 
\addplot+ [color = blue!50!black,solid,boxplot prepared = {box extend=0.4, draw position = 10, upper whisker = 0.597261, upper quartile = 0.591855, median = 0.582898, lower quartile = 0.579444, lower whisker = 0.568106},]coordinates{}; 
\addplot+ [color = blue!50!black,solid,boxplot prepared = {box extend=0.4, draw position = 11, upper whisker = 0.587987, upper quartile = 0.581270, median = 0.576112, lower quartile = 0.573518, lower whisker = 0.556067},]coordinates{};

\addplot+ [color = red!70!black,solid,boxplot prepared = {box extend=0.4, draw position = 1, lower whisker = 0.573521, lower quartile = 0.575512, median = 0.576715, upper quartile = 0.577702, upper whisker = 0.579620},]coordinates{}; 
\addplot+ [color = red!70!black,solid,boxplot prepared = {box extend=0.4, draw position = 2, lower whisker = 0.859387, lower quartile = 0.870538, median = 0.872640, upper quartile = 0.874923, upper whisker = 0.879567},]coordinates{}; 
\addplot+ [color = red!70!black,solid,boxplot prepared = {box extend=0.4, draw position = 3, lower whisker = 0.931284, lower quartile = 0.933401, median = 0.935217, upper quartile = 0.936746, upper whisker = 0.938053},]coordinates{}; 
\addplot+ [color = red!70!black,solid,boxplot prepared = {box extend=0.4, draw position = 4, lower whisker = 0.951478, lower quartile = 0.953341, median = 0.954085, upper quartile = 0.954621, upper whisker = 0.955221},]coordinates{}; 
\addplot+ [color = red!70!black,solid,boxplot prepared = {box extend=0.4, draw position = 5, lower whisker = 0.959751, lower quartile = 0.962712, median = 0.963661, upper quartile = 0.963842, upper whisker = 0.966639},]coordinates{}; 
\addplot+ [color = red!70!black,solid,boxplot prepared = {box extend=0.4, draw position = 6, lower whisker = 0.967118, lower quartile = 0.969189, median = 0.969529, upper quartile = 0.970626, upper whisker = 0.971680},]coordinates{}; 
\addplot+ [color = red!70!black,solid,boxplot prepared = {box extend=0.4, draw position = 7, lower whisker = 0.970371, lower quartile = 0.972813, median = 0.973415, upper quartile = 0.974614, upper whisker = 0.976680},]coordinates{}; 
\addplot+ [color = red!70!black,solid,boxplot prepared = {box extend=0.4, draw position = 8, lower whisker = 0.975417, lower quartile = 0.975686, median = 0.976479, upper quartile = 0.977643, upper whisker = 0.978513},]coordinates{}; 
\addplot+ [color = red!70!black,solid,boxplot prepared = {box extend=0.4, draw position = 9, lower whisker = 0.976438, lower quartile = 0.977982, median = 0.978677, upper quartile = 0.979150, upper whisker = 0.980169},]coordinates{}; 
\addplot+ [color = red!70!black,solid,boxplot prepared = {box extend=0.4, draw position = 10, lower whisker = 0.977814, lower quartile = 0.979731, median = 0.980077, upper quartile = 0.981225, upper whisker = 0.981783},]coordinates{}; 
\addplot+ [color = red!70!black,solid,boxplot prepared = {box extend=0.4, draw position = 11, lower whisker = 0.980028, lower quartile = 0.981022, median = 0.981347, upper quartile = 0.981774, upper whisker = 0.983726},]coordinates{}; 

\end{axis}
\end{tikzpicture}
}
\end{minipage}
\\

\ref{named}
\caption{Local segregation of pairwise stable networks in the add-only \deigame obtained by the best move dynamic for $n=1000$ over 50 runs starting from a random or segregated tree and grid.}
\label{plot:ICFNCG_AddOnly_appendix}
\end{figure}
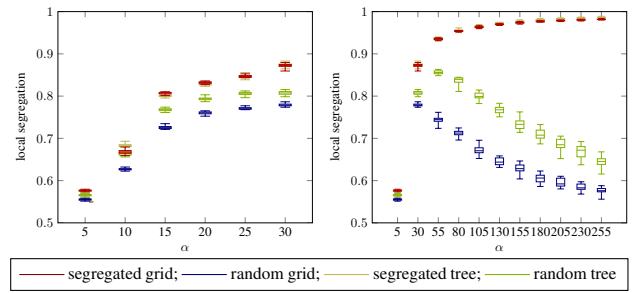

\subsection{Experiments Regarding the Global Segregation Measure}
This section illustrates the dependence of the global segregation measure on the parameter $\alpha$ and the initial state in the \deigame and \icfgame. The observations are similar as for the local segregation measure, showing the robustness of our results.

\subsubsection{Results for the \icfgame}
The results for the global segregation measure for $1.01$-approximate networks in the \icfgame and pairwise stable networks in the add-only \icfgame are presented in \Cref{plot:apx_ICFNCG_global} and \Cref{plot:ICF_AddOnly_global}.
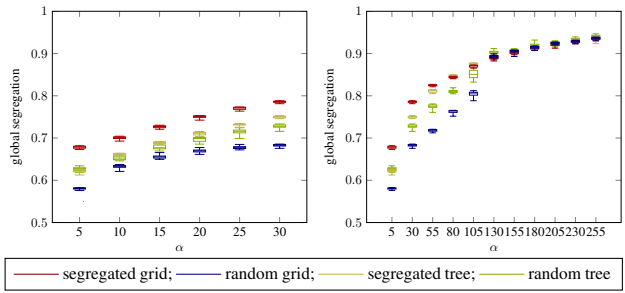
\begin{figure}[!ht]
\centering
\begin{minipage}{0.23\textwidth}
\resizebox {\textwidth} {!} {
\begin{tikzpicture}
\begin{axis}[
legend columns=-1,
legend entries={segregated grid;,random grid;,segregated tree;,random tree},
legend to name=named,
legend style={nodes={scale=0.65, transform shape}},
xlabel= {$\alpha$},
ylabel= {global segregation},
boxplot/draw direction=y,
baseline,
xtick = {1,2,3, 4, 5, 6, 7, 8},
xticklabels = {5, 10, 15, 20, 25, 30, 35, 40},
ymin=0.5,
ymax=1
]
\addplot[red!50!black, domain=1.1:1.11]{0.55};
\addplot[blue!50!black, domain=1.1:1.11]{0.55};
\addplot[yellow!70!black, domain=1.1:1.11]{0.55};
\addplot[lime!70!black, domain=1.1:1.11]{0.55};

\addplot+ [color = lime!70!black,solid,boxplot prepared = {box extend=0.3, draw position = 1, lower whisker = 0.618929, lower quartile = 0.623155, median = 0.627830, upper quartile = 0.629695, upper whisker = 0.634229},]coordinates{}; 
\addplot+ [color = lime!70!black,solid,boxplot prepared = {box extend=0.300000, draw position = 2, lower whisker = 0.646098, lower quartile = 0.649019, median = 0.652645, upper quartile = 0.657220, upper whisker = 0.660137},]coordinates{}; 
\addplot+ [color = lime!70!black,solid,boxplot prepared = {box extend=0.300000, draw position = 3, lower whisker = 0.669281, lower quartile = 0.674502, median = 0.677092, upper quartile = 0.683413, upper whisker = 0.690321},]coordinates{}; 
\addplot+ [color = lime!70!black,solid,boxplot prepared = {box extend=0.300000, draw position = 4, lower whisker = 0.685567, lower quartile = 0.692263, median = 0.698912, upper quartile = 0.701361, upper whisker = 0.714845},]coordinates{}; 
\addplot+ [color = lime!70!black,solid,boxplot prepared = {box extend=0.300000, draw position = 5, lower whisker = 0.699136, lower quartile = 0.712274, median = 0.715928, upper quartile = 0.718830, upper whisker = 0.724455},]coordinates{}; 
\addplot+ [color = lime!70!black,solid,boxplot prepared = {box extend=0.3, draw position = 6, lower whisker = 0.715860, lower quartile = 0.725669, median = 0.727815, upper quartile = 0.731089, upper whisker = 0.733295},]coordinates{}; 

\addplot+ [color = yellow!70!black,solid,boxplot prepared = {box extend=0.300000, draw position = 1, lower whisker = 0.612479, lower quartile = 0.620354, median = 0.623600, upper quartile = 0.625415, upper whisker = 0.626462},]coordinates{}; 
\addplot+ [color = yellow!70!black,solid,boxplot prepared = {box extend=0.300000, draw position = 2, lower whisker = 0.647863, lower quartile = 0.658040, median = 0.662370, upper quartile = 0.662997, upper whisker = 0.664131},]coordinates{}; 
\addplot+ [color = yellow!70!black,solid,boxplot prepared = {box extend=0.300000, draw position = 3, lower whisker = 0.685244, lower quartile = 0.687451, median = 0.688737, upper quartile = 0.691096, upper whisker = 0.691676},]coordinates{}; 
\addplot+ [color = yellow!70!black,solid,boxplot prepared = {box extend=0.300000, draw position = 4, lower whisker = 0.703676, lower quartile = 0.707377, median = 0.711440, upper quartile = 0.712958, upper whisker = 0.714674},]coordinates{}; 
\addplot+ [color = yellow!70!black,solid,boxplot prepared = {box extend=0.300000, draw position = 5, lower whisker = 0.725547, lower quartile = 0.730981, median = 0.732274, upper quartile = 0.732673, upper whisker = 0.734626},]coordinates{}; 
\addplot+ [color = yellow!70!black,solid,boxplot prepared = {box extend=0.300000, draw position = 6, lower whisker = 0.745875, lower quartile = 0.747814, median = 0.750736, upper quartile = 0.751221, upper whisker = 0.753105},]coordinates{};

\addplot+ [color = red!70!black,solid,boxplot prepared = {box extend=0.300000, draw position = 1, lower whisker = 0.672942, lower quartile = 0.676011, median = 0.678449, upper quartile = 0.680490, upper whisker = 0.681873},]coordinates{}; 
\addplot+ [color = red!70!black,solid,boxplot prepared = {box extend=0.300000, draw position = 2, lower whisker = 0.692710, lower quartile = 0.698926, median = 0.701025, upper quartile = 0.702816, upper whisker = 0.704298},]coordinates{}; 
\addplot+ [color = red!70!black,solid,boxplot prepared = {box extend=0.300000, draw position = 3, lower whisker = 0.720882, lower quartile = 0.724619, median = 0.727212, upper quartile = 0.728980, upper whisker = 0.730315},]coordinates{}; 
\addplot+ [color = red!70!black,solid,boxplot prepared = {box extend=0.300000, draw position = 4, lower whisker = 0.742350, lower quartile = 0.748929, median = 0.751017, upper quartile = 0.752624, upper whisker = 0.753410},]coordinates{}; 
\addplot+ [color = red!70!black,solid,boxplot prepared = {box extend=0.300000, draw position = 5, lower whisker = 0.763720, lower quartile = 0.767944, median = 0.771108, upper quartile = 0.772991, upper whisker = 0.773432},]coordinates{}; 
\addplot+ [color = red!70!black,solid,boxplot prepared = {box extend=0.300000, draw position = 6, lower whisker = 0.781304, lower quartile = 0.784482, median = 0.786057, upper quartile = 0.787100, upper whisker = 0.788396},]coordinates{};

\addplot+ [color = blue!50!black,solid,boxplot prepared = {box extend=0.3, draw position = 1, lower whisker = 0.575970, lower quartile = 0.579377, median = 0.580970, upper quartile = 0.582007, upper whisker = 0.582597},]coordinates{}; 
\addplot+ [color = blue!50!black,solid,boxplot prepared = {box extend=0.300000, draw position = 2, lower whisker = 0.620996, lower quartile = 0.630794, median = 0.633166, upper quartile = 0.634829, upper whisker = 0.637280},]coordinates{}; 
\addplot+ [color = blue!50!black,solid,boxplot prepared = {box extend=0.300000, draw position = 3, lower whisker = 0.649546, lower quartile = 0.652589, median = 0.655450, upper quartile = 0.657389, upper whisker = 0.665612},]coordinates{}; 
\addplot+ [color = blue!50!black,solid,boxplot prepared = {box extend=0.300000, draw position = 4, lower whisker = 0.661351, lower quartile = 0.667158, median = 0.669622, upper quartile = 0.671579, upper whisker = 0.677560},]coordinates{}; 
\addplot+ [color = blue!50!black,solid,boxplot prepared = {box extend=0.300000, draw position = 5, lower whisker = 0.671956, lower quartile = 0.675246, median = 0.677191, upper quartile = 0.679612, upper whisker = 0.684882},]coordinates{}; 
\addplot+ [color = blue!50!black,solid,boxplot prepared = {box extend=0.3, draw position = 6, lower whisker = 0.675625, lower quartile = 0.681141, median = 0.682848, upper quartile = 0.684797, upper whisker = 0.685680},]coordinates{}; 

\end{axis}
\end{tikzpicture}
}
\end{minipage}
\begin{minipage}{0.23\textwidth}
\resizebox {\textwidth} {!} {
\begin{tikzpicture}
\begin{axis}[
xlabel= {$\alpha$},
ylabel= {global segregation},
boxplot/draw direction=y,
baseline,
xtick = {1, 2, 3, ..., 11},
xticklabels  = {5, 30, 55, 80, 105, 130, 155, 180, 205, 230, 255},
ymin=0.5,
ymax=1
]

\addplot+ [color = lime!70!black,solid,boxplot prepared = {box extend=0.400000, draw position = 1, lower whisker = 0.618929, lower quartile = 0.623155, median = 0.627830, upper quartile = 0.629695, upper whisker = 0.634229},]coordinates{}; 
\addplot+ [color = lime!70!black,solid,boxplot prepared = {box extend=0.400000, draw position = 2, lower whisker = 0.715860, lower quartile = 0.725669, median = 0.727815, upper quartile = 0.731089, upper whisker = 0.733295},]coordinates{}; 
\addplot+ [color = lime!70!black,solid,boxplot prepared = {box extend=0.400000, draw position = 3, lower whisker = 0.760687, lower quartile = 0.772744, median = 0.775262, upper quartile = 0.778594, upper whisker = 0.780765},]coordinates{}; 
\addplot+ [color = lime!70!black,solid,boxplot prepared = {box extend=0.400000, draw position = 4, lower whisker = 0.804795, lower quartile = 0.807985, median = 0.810271, upper quartile = 0.812639, upper whisker = 0.819005},]coordinates{}; 
\addplot+ [color = lime!70!black,solid,boxplot prepared = {box extend=0.400000, draw position = 5, lower whisker = 0.832432, lower quartile = 0.843004, median = 0.850719, upper quartile = 0.860242, upper whisker = 0.866781},]coordinates{}; 
\addplot+ [color = lime!70!black,solid,boxplot prepared = {box extend=0.400000, draw position = 6, lower whisker = 0.888378, lower quartile = 0.900210, median = 0.903063, upper quartile = 0.905212, upper whisker = 0.911964},]coordinates{}; 
\addplot+ [color = lime!70!black,solid,boxplot prepared = {box extend=0.400000, draw position = 7, lower whisker = 0.900694, lower quartile = 0.903075, median = 0.907448, upper quartile = 0.909446, upper whisker = 0.911598},]coordinates{}; 
\addplot+ [color = lime!70!black,solid,boxplot prepared = {box extend=0.400000, draw position = 8, lower whisker = 0.910292, lower quartile = 0.917283, median = 0.919758, upper quartile = 0.922127, upper whisker = 0.931925},]coordinates{}; 
\addplot+ [color = lime!70!black,solid,boxplot prepared = {box extend=0.400000, draw position = 9, lower whisker = 0.916048, lower quartile = 0.923636, median = 0.925598, upper quartile = 0.927527, upper whisker = 0.930233},]coordinates{}; 
\addplot+ [color = lime!70!black,solid,boxplot prepared = {box extend=0.400000, draw position = 10, lower whisker = 0.927795, lower quartile = 0.930515, median = 0.932499, upper quartile = 0.936042, upper whisker = 0.939444},]coordinates{}; 
\addplot+ [color = lime!70!black,solid,boxplot prepared = {box extend=0.400000, draw position = 11, lower whisker = 0.930124, lower quartile = 0.935451, median = 0.939348, upper quartile = 0.943102, upper whisker = 0.946578},]coordinates{}; 

\addplot+ [color = yellow!70!black,solid,boxplot prepared = {box extend=0.400000, draw position = 1, lower whisker = 0.612479, lower quartile = 0.620354, median = 0.623600, upper quartile = 0.625415, upper whisker = 0.626462},]coordinates{}; 
\addplot+ [color = yellow!70!black,solid,boxplot prepared = {box extend=0.400000, draw position = 2, lower whisker = 0.745875, lower quartile = 0.747814, median = 0.750736, upper quartile = 0.751221, upper whisker = 0.753105},]coordinates{}; 
\addplot+ [color = yellow!70!black,solid,boxplot prepared = {box extend=0.400000, draw position = 3, lower whisker = 0.805130, lower quartile = 0.808382, median = 0.812320, upper quartile = 0.813887, upper whisker = 0.814978},]coordinates{}; 
\addplot+ [color = yellow!70!black,solid,boxplot prepared = {box extend=0.400000, draw position = 4, lower whisker = 0.839957, lower quartile = 0.846873, median = 0.848337, upper quartile = 0.849223, upper whisker = 0.849765},]coordinates{}; 
\addplot+ [color = yellow!70!black,solid,boxplot prepared = {box extend=0.400000, draw position = 5, lower whisker = 0.872236, lower quartile = 0.874235, median = 0.875391, upper quartile = 0.877220, upper whisker = 0.877882},]coordinates{}; 
\addplot+ [color = yellow!70!black,solid,boxplot prepared = {box extend=0.400000, draw position = 6, lower whisker = 0.882093, lower quartile = 0.887814, median = 0.890936, upper quartile = 0.894172, upper whisker = 0.895877},]coordinates{}; 
\addplot+ [color = yellow!70!black,solid,boxplot prepared = {box extend=0.400000, draw position = 7, lower whisker = 0.897558, lower quartile = 0.901108, median = 0.903828, upper quartile = 0.905611, upper whisker = 0.906386},]coordinates{}; 
\addplot+ [color = yellow!70!black,solid,boxplot prepared = {box extend=0.400000, draw position = 8, lower whisker = 0.909292, lower quartile = 0.913115, median = 0.914551, upper quartile = 0.917214, upper whisker = 0.918382},]coordinates{}; 
\addplot+ [color = yellow!70!black,solid,boxplot prepared = {box extend=0.400000, draw position = 9, lower whisker = 0.913462, lower quartile = 0.919304, median = 0.920338, upper quartile = 0.923104, upper whisker = 0.924288},]coordinates{}; 
\addplot+ [color = yellow!70!black,solid,boxplot prepared = {box extend=0.400000, draw position = 10, lower whisker = 0.926923, lower quartile = 0.929136, median = 0.930512, upper quartile = 0.931265, upper whisker = 0.931748},]coordinates{}; 
\addplot+ [color = yellow!70!black,solid,boxplot prepared = {box extend=0.400000, draw position = 11, lower whisker = 0.923567, lower quartile = 0.933620, median = 0.935762, upper quartile = 0.937801, upper whisker = 0.938312},]coordinates{}; 

\addplot+ [color = red!70!black,solid,boxplot prepared = {box extend=0.400000, draw position = 1, lower whisker = 0.672942, lower quartile = 0.676011, median = 0.678449, upper quartile = 0.680490, upper whisker = 0.681873},]coordinates{}; 
\addplot+ [color = red!70!black,solid,boxplot prepared = {box extend=0.400000, draw position = 2, lower whisker = 0.781304, lower quartile = 0.784482, median = 0.786057, upper quartile = 0.787100, upper whisker = 0.788396},]coordinates{}; 
\addplot+ [color = red!70!black,solid,boxplot prepared = {box extend=0.400000, draw position = 3, lower whisker = 0.821429, lower quartile = 0.823840, median = 0.824950, upper quartile = 0.826120, upper whisker = 0.826979},]coordinates{}; 
\addplot+ [color = red!70!black,solid,boxplot prepared = {box extend=0.400000, draw position = 4, lower whisker = 0.840796, lower quartile = 0.842715, median = 0.844414, upper quartile = 0.845490, upper whisker = 0.846120},]coordinates{}; 
\addplot+ [color = red!70!black,solid,boxplot prepared = {box extend=0.400000, draw position = 5, lower whisker = 0.865273, lower quartile = 0.868058, median = 0.871349, upper quartile = 0.871787, upper whisker = 0.872941},]coordinates{}; 
\addplot+ [color = red!70!black,solid,boxplot prepared = {box extend=0.400000, draw position = 6, lower whisker = 0.882663, lower quartile = 0.887560, median = 0.889064, upper quartile = 0.892256, upper whisker = 0.896095},]coordinates{}; 
\addplot+ [color = red!70!black,solid,boxplot prepared = {box extend=0.400000, draw position = 7, lower whisker = 0.897108, lower quartile = 0.899793, median = 0.902625, upper quartile = 0.903433, upper whisker = 0.903915},]coordinates{}; 
\addplot+ [color = red!70!black,solid,boxplot prepared = {box extend=0.400000, draw position = 8, lower whisker = 0.907089, lower quartile = 0.910662, median = 0.913250, upper quartile = 0.914894, upper whisker = 0.915888},]coordinates{}; 
\addplot+ [color = red!70!black,solid,boxplot prepared = {box extend=0.400000, draw position = 9, lower whisker = 0.912600, lower quartile = 0.919229, median = 0.921096, upper quartile = 0.922249, upper whisker = 0.923552},]coordinates{}; 
\addplot+ [color = red!70!black,solid,boxplot prepared = {box extend=0.400000, draw position = 10, lower whisker = 0.923722, lower quartile = 0.927025, median = 0.928799, upper quartile = 0.929939, upper whisker = 0.931973},]coordinates{}; 
\addplot+ [color = red!70!black,solid,boxplot prepared = {box extend=0.400000, draw position = 11, lower whisker = 0.929852, lower quartile = 0.933410, median = 0.934697, upper quartile = 0.936489, upper whisker = 0.937549},]coordinates{}; 

\addplot+ [color = blue!50!black,solid,boxplot prepared = {box extend=0.400000, draw position = 1, lower whisker = 0.575970, lower quartile = 0.579377, median = 0.580970, upper quartile = 0.582007, upper whisker = 0.582597},]coordinates{}; 
\addplot+ [color = blue!50!black,solid,boxplot prepared = {box extend=0.400000, draw position = 2, lower whisker = 0.675625, lower quartile = 0.681141, median = 0.682848, upper quartile = 0.684797, upper whisker = 0.685680},]coordinates{}; 
\addplot+ [color = blue!50!black,solid,boxplot prepared = {box extend=0.400000, draw position = 3, lower whisker = 0.712288, lower quartile = 0.716263, median = 0.718814, upper quartile = 0.719896, upper whisker = 0.720797},]coordinates{}; 
\addplot+ [color = blue!50!black,solid,boxplot prepared = {box extend=0.400000, draw position = 4, lower whisker = 0.751730, lower quartile = 0.760741, median = 0.764120, upper quartile = 0.765041, upper whisker = 0.765606},]coordinates{}; 
\addplot+ [color = blue!50!black,solid,boxplot prepared = {box extend=0.400000, draw position = 5, lower whisker = 0.788302, lower quartile = 0.801479, median = 0.806613, upper quartile = 0.808405, upper whisker = 0.812225},]coordinates{}; 
\addplot+ [color = blue!50!black,solid,boxplot prepared = {box extend=0.400000, draw position = 6, lower whisker = 0.886864, lower quartile = 0.890807, median = 0.892674, upper quartile = 0.895277, upper whisker = 0.899286},]coordinates{}; 
\addplot+ [color = blue!50!black,solid,boxplot prepared = {box extend=0.400000, draw position = 7, lower whisker = 0.892928, lower quartile = 0.903969, median = 0.905549, upper quartile = 0.907083, upper whisker = 0.908397},]coordinates{}; 
\addplot+ [color = blue!50!black,solid,boxplot prepared = {box extend=0.400000, draw position = 8, lower whisker = 0.907089, lower quartile = 0.912794, median = 0.915557, upper quartile = 0.916729, upper whisker = 0.918107},]coordinates{}; 
\addplot+ [color = blue!50!black,solid,boxplot prepared = {box extend=0.400000, draw position = 9, lower whisker = 0.920952, lower quartile = 0.923696, median = 0.924789, upper quartile = 0.926446, upper whisker = 0.927438},]coordinates{}; 
\addplot+ [color = blue!50!black,solid,boxplot prepared = {box extend=0.400000, draw position = 10, lower whisker = 0.922615, lower quartile = 0.926818, median = 0.929079, upper quartile = 0.930162, upper whisker = 0.931061},]coordinates{}; 
\addplot+ [color = blue!50!black,solid,boxplot prepared = {box extend=0.400000, draw position = 11, lower whisker = 0.931411, lower quartile = 0.935251, median = 0.937298, upper quartile = 0.938241, upper whisker = 0.939151},]coordinates{}; 

\end{axis}
\end{tikzpicture}
}
\end{minipage}
\ref{named}
\caption{Global segregation of  $1.01$-approximate networks in the \icfgame obtained by the best move dynamic for $n=1000$ over 50 runs starting from a random or segregated tree and grid.}
\label{plot:apx_ICFNCG_global}
\end{figure}

\begin{figure}[!ht]
\centering
\begin{minipage}{0.23\textwidth}
\resizebox {\textwidth} {!} {
\begin{tikzpicture}
\begin{axis}[
legend columns=-1,
legend entries={segregated grid;,random grid;,segregated tree;,random tree},
legend to name=named,
legend style={nodes={scale=0.65, transform shape}},
xlabel= {$\alpha$},
ylabel= {global segregation},
boxplot/draw direction=y,
baseline,
xtick = {1,2,3, 4, 5, 6, 7, 8},
xticklabels = {5, 10, 15, 20, 25, 30, 35, 40},
ymin=0.5,
ymax=1
]
\addplot[red!50!black, domain=1.1:1.11]{0.55};
\addplot[blue!50!black, domain=1.1:1.11]{0.55};
\addplot[yellow!70!black, domain=1.1:1.11]{0.55};
\addplot[lime!70!black, domain=1.1:1.11]{0.55};

\addplot+ [color = lime!70!black,solid,boxplot prepared = {box extend=0.300000, draw position = 1, lower whisker = 0.648007, lower quartile = 0.648604, median = 0.649320, upper quartile = 0.649623, upper whisker = 0.649807},]coordinates{}; 
\addplot+ [color = lime!70!black,solid,boxplot prepared = {box extend=0.300000, draw position = 2, lower whisker = 0.632181, lower quartile = 0.633513, median = 0.634361, upper quartile = 0.635223, upper whisker = 0.635909},]coordinates{}; 
\addplot+ [color = lime!70!black,solid,boxplot prepared = {box extend=0.300000, draw position = 3, lower whisker = 0.679879, lower quartile = 0.681863, median = 0.683223, upper quartile = 0.683817, upper whisker = 0.685149},]coordinates{}; 
\addplot+ [color = lime!70!black,solid,boxplot prepared = {box extend=0.300000, draw position = 4, lower whisker = 0.714886, lower quartile = 0.715664, median = 0.716619, upper quartile = 0.717405, upper whisker = 0.718170},]coordinates{}; 
\addplot+ [color = lime!70!black,solid,boxplot prepared = {box extend=0.300000, draw position = 5, lower whisker = 0.728128, lower quartile = 0.729569, median = 0.730808, upper quartile = 0.732004, upper whisker = 0.733315},]coordinates{}; 
\addplot+ [color = lime!70!black,solid,boxplot prepared = {box extend=0.300000, draw position = 6, lower whisker = 0.737440, lower quartile = 0.739493, median = 0.740480, upper quartile = 0.741059, upper whisker = 0.741290},]coordinates{}; 

\addplot+ [color = yellow!70!black,solid,boxplot prepared = {box extend=0.300000, draw position = 1, lower whisker = 0.643976, lower quartile = 0.645205, median = 0.645717, upper quartile = 0.646838, upper whisker = 0.647937},]coordinates{}; 
\addplot+ [color = yellow!70!black,solid,boxplot prepared = {box extend=0.300000, draw position = 2, lower whisker = 0.630822, lower quartile = 0.632524, median = 0.633139, upper quartile = 0.633528, upper whisker = 0.633663},]coordinates{}; 
\addplot+ [color = yellow!70!black,solid,boxplot prepared = {box extend=0.300000, draw position = 3, lower whisker = 0.681770, lower quartile = 0.684107, median = 0.684907, upper quartile = 0.685582, upper whisker = 0.685988},]coordinates{}; 
\addplot+ [color = yellow!70!black,solid,boxplot prepared = {box extend=0.300000, draw position = 4, lower whisker = 0.716240, lower quartile = 0.717477, median = 0.718607, upper quartile = 0.719288, upper whisker = 0.719723},]coordinates{}; 
\addplot+ [color = yellow!70!black,solid,boxplot prepared = {box extend=0.300000, draw position = 5, lower whisker = 0.728806, lower quartile = 0.730788, median = 0.731329, upper quartile = 0.731726, upper whisker = 0.732125},]coordinates{}; 
\addplot+ [color = yellow!70!black,solid,boxplot prepared = {box extend=0.300000, draw position = 6, lower whisker = 0.742579, lower quartile = 0.744773, median = 0.747167, upper quartile = 0.747985, upper whisker = 0.750081},]coordinates{}; 

\addplot+ [color = red!70!black,solid,boxplot prepared = {box extend=0.300000, draw position = 1, lower whisker = 0.642875, lower quartile = 0.645809, median = 0.646281, upper quartile = 0.647352, upper whisker = 0.649161},]coordinates{}; 
\addplot+ [color = red!70!black,solid,boxplot prepared = {box extend=0.300000, draw position = 2, lower whisker = 0.630746, lower quartile = 0.631111, median = 0.632170, upper quartile = 0.632837, upper whisker = 0.633751},]coordinates{}; 
\addplot+ [color = red!70!black,solid,boxplot prepared = {box extend=0.300000, draw position = 3, lower whisker = 0.690603, lower quartile = 0.692035, median = 0.693586, upper quartile = 0.694292, upper whisker = 0.697453},]coordinates{}; 
\addplot+ [color = red!70!black,solid,boxplot prepared = {box extend=0.300000, draw position = 4, lower whisker = 0.726248, lower quartile = 0.728068, median = 0.729734, upper quartile = 0.730238, upper whisker = 0.731425},]coordinates{}; 
\addplot+ [color = red!70!black,solid,boxplot prepared = {box extend=0.300000, draw position = 5, lower whisker = 0.738079, lower quartile = 0.739128, median = 0.739926, upper quartile = 0.741113, upper whisker = 0.743522},]coordinates{}; 
\addplot+ [color = red!70!black,solid,boxplot prepared = {box extend=0.300000, draw position = 6, lower whisker = 0.749704, lower quartile = 0.751004, median = 0.753062, upper quartile = 0.753901, upper whisker = 0.756270},]coordinates{}; 

\addplot+ [color = blue!50!black,solid,boxplot prepared = {box extend=0.300000, draw position = 1, lower whisker = 0.658300, lower quartile = 0.660030, median = 0.660389, upper quartile = 0.660936, upper whisker = 0.661601},]coordinates{}; 
\addplot+ [color = blue!50!black,solid,boxplot prepared = {box extend=0.300000, draw position = 2, lower whisker = 0.629599, lower quartile = 0.631421, median = 0.632037, upper quartile = 0.632408, upper whisker = 0.633867},]coordinates{}; 
\addplot+ [color = blue!50!black,solid,boxplot prepared = {box extend=0.300000, draw position = 3, lower whisker = 0.675195, lower quartile = 0.677056, median = 0.679030, upper quartile = 0.679880, upper whisker = 0.681814},]coordinates{}; 
\addplot+ [color = blue!50!black,solid,boxplot prepared = {box extend=0.300000, draw position = 4, lower whisker = 0.710738, lower quartile = 0.714578, median = 0.715147, upper quartile = 0.716649, upper whisker = 0.717636},]coordinates{}; 
\addplot+ [color = blue!50!black,solid,boxplot prepared = {box extend=0.300000, draw position = 5, lower whisker = 0.721108, lower quartile = 0.724819, median = 0.725819, upper quartile = 0.727084, upper whisker = 0.729588},]coordinates{}; 
\addplot+ [color = blue!50!black,solid,boxplot prepared = {box extend=0.300000, draw position = 6, lower whisker = 0.730211, lower quartile = 0.732211, median = 0.733372, upper quartile = 0.734354, upper whisker = 0.734898},]coordinates{}; 

\end{axis}
\end{tikzpicture}
}
\end{minipage}
\begin{minipage}{0.23\textwidth}
\resizebox {\textwidth} {!} {
\begin{tikzpicture}
\begin{axis}[
xlabel= {$\alpha$},
ylabel= {global segregation},
boxplot/draw direction=y,
baseline,
xtick = {1, 2, 3, ..., 11},
xticklabels  = {5, 30, 55, 80, 105, 130, 155, 180, 205, 230, 255},
ymin=0.5,
ymax=1
]

\addplot+ [color = lime!70!black,solid,boxplot prepared = {box extend=0.400000, draw position = 1, lower whisker = 0.648007, lower quartile = 0.648604, median = 0.649320, upper quartile = 0.649623, upper whisker = 0.649807},]coordinates{}; 
\addplot+ [color = lime!70!black,solid,boxplot prepared = {box extend=0.400000, draw position = 2, lower whisker = 0.737440, lower quartile = 0.739493, median = 0.740480, upper quartile = 0.741059, upper whisker = 0.741290},]coordinates{}; 
\addplot+ [color = lime!70!black,solid,boxplot prepared = {box extend=0.400000, draw position = 3, lower whisker = 0.778626, lower quartile = 0.782714, median = 0.783492, upper quartile = 0.784204, upper whisker = 0.784753},]coordinates{}; 
\addplot+ [color = lime!70!black,solid,boxplot prepared = {box extend=0.400000, draw position = 4, lower whisker = 0.783639, lower quartile = 0.787297, median = 0.789242, upper quartile = 0.790266, upper whisker = 0.791085},]coordinates{}; 
\addplot+ [color = lime!70!black,solid,boxplot prepared = {box extend=0.400000, draw position = 5, lower whisker = 0.765206, lower quartile = 0.773098, median = 0.777313, upper quartile = 0.777779, upper whisker = 0.778775},]coordinates{}; 
\addplot+ [color = lime!70!black,solid,boxplot prepared = {box extend=0.400000, draw position = 6, lower whisker = 0.745188, lower quartile = 0.753110, median = 0.755892, upper quartile = 0.757775, upper whisker = 0.758525},]coordinates{}; 
\addplot+ [color = lime!70!black,solid,boxplot prepared = {box extend=0.400000, draw position = 7, lower whisker = 0.717182, lower quartile = 0.724742, median = 0.731864, upper quartile = 0.734665, upper whisker = 0.737000},]coordinates{}; 
\addplot+ [color = lime!70!black,solid,boxplot prepared = {box extend=0.400000, draw position = 8, lower whisker = 0.695897, lower quartile = 0.703981, median = 0.710274, upper quartile = 0.712241, upper whisker = 0.715324},]coordinates{}; 
\addplot+ [color = lime!70!black,solid,boxplot prepared = {box extend=0.400000, draw position = 9, lower whisker = 0.680276, lower quartile = 0.689853, median = 0.693562, upper quartile = 0.698360, upper whisker = 0.700844},]coordinates{}; 
\addplot+ [color = lime!70!black,solid,boxplot prepared = {box extend=0.400000, draw position = 10, lower whisker = 0.661774, lower quartile = 0.670861, median = 0.674266, upper quartile = 0.677608, upper whisker = 0.679503},]coordinates{}; 
\addplot+ [color = lime!70!black,solid,boxplot prepared = {box extend=0.400000, draw position = 11, lower whisker = 0.638889, lower quartile = 0.652936, median = 0.657060, upper quartile = 0.662937, upper whisker = 0.663799},]coordinates{}; 

\addplot+ [color = yellow!70!black,solid,boxplot prepared = {box extend=0.400000, draw position = 1, lower whisker = 0.643976, lower quartile = 0.645205, median = 0.645717, upper quartile = 0.646838, upper whisker = 0.647937},]coordinates{}; 
\addplot+ [color = yellow!70!black,solid,boxplot prepared = {box extend=0.400000, draw position = 2, lower whisker = 0.742579, lower quartile = 0.744773, median = 0.747167, upper quartile = 0.747985, upper whisker = 0.750081},]coordinates{}; 
\addplot+ [color = yellow!70!black,solid,boxplot prepared = {box extend=0.400000, draw position = 3, lower whisker = 0.805945, lower quartile = 0.807923, median = 0.809153, upper quartile = 0.810766, upper whisker = 0.814514},]coordinates{}; 
\addplot+ [color = yellow!70!black,solid,boxplot prepared = {box extend=0.400000, draw position = 4, lower whisker = 0.839542, lower quartile = 0.842693, median = 0.844190, upper quartile = 0.846097, upper whisker = 0.847535},]coordinates{}; 
\addplot+ [color = yellow!70!black,solid,boxplot prepared = {box extend=0.400000, draw position = 5, lower whisker = 0.864472, lower quartile = 0.866667, median = 0.869385, upper quartile = 0.872506, upper whisker = 0.874357},]coordinates{}; 
\addplot+ [color = yellow!70!black,solid,boxplot prepared = {box extend=0.400000, draw position = 6, lower whisker = 0.881146, lower quartile = 0.883085, median = 0.885910, upper quartile = 0.889085, upper whisker = 0.892354},]coordinates{}; 
\addplot+ [color = yellow!70!black,solid,boxplot prepared = {box extend=0.400000, draw position = 7, lower whisker = 0.893617, lower quartile = 0.898765, median = 0.900872, upper quartile = 0.902588, upper whisker = 0.906937},]coordinates{}; 
\addplot+ [color = yellow!70!black,solid,boxplot prepared = {box extend=0.400000, draw position = 8, lower whisker = 0.905193, lower quartile = 0.908902, median = 0.911418, upper quartile = 0.913350, upper whisker = 0.916012},]coordinates{}; 
\addplot+ [color = yellow!70!black,solid,boxplot prepared = {box extend=0.400000, draw position = 9, lower whisker = 0.914340, lower quartile = 0.918233, median = 0.919370, upper quartile = 0.922330, upper whisker = 0.924852},]coordinates{}; 
\addplot+ [color = yellow!70!black,solid,boxplot prepared = {box extend=0.400000, draw position = 10, lower whisker = 0.924252, lower quartile = 0.926420, median = 0.928838, upper quartile = 0.930113, upper whisker = 0.931897},]coordinates{}; 
\addplot+ [color = yellow!70!black,solid,boxplot prepared = {box extend=0.400000, draw position = 11, lower whisker = 0.928068, lower quartile = 0.931467, median = 0.934067, upper quartile = 0.935679, upper whisker = 0.937820},]coordinates{}; 

\addplot+ [color = blue!50!black,solid,boxplot prepared = {box extend=0.400000, draw position = 1, lower whisker = 0.658300, lower quartile = 0.660030, median = 0.660389, upper quartile = 0.660936, upper whisker = 0.661601},]coordinates{}; 
\addplot+ [color = blue!50!black,solid,boxplot prepared = {box extend=0.400000, draw position = 2, lower whisker = 0.730211, lower quartile = 0.732211, median = 0.733372, upper quartile = 0.734354, upper whisker = 0.734898},]coordinates{}; 
\addplot+ [color = blue!50!black,solid,boxplot prepared = {box extend=0.400000, draw position = 3, lower whisker = 0.745502, lower quartile = 0.746995, median = 0.749004, upper quartile = 0.749945, upper whisker = 0.750923},]coordinates{}; 
\addplot+ [color = blue!50!black,solid,boxplot prepared = {box extend=0.400000, draw position = 4, lower whisker = 0.724851, lower quartile = 0.729116, median = 0.730055, upper quartile = 0.731054, upper whisker = 0.733386},]coordinates{}; 
\addplot+ [color = blue!50!black,solid,boxplot prepared = {box extend=0.400000, draw position = 5, lower whisker = 0.695086, lower quartile = 0.699329, median = 0.701764, upper quartile = 0.702635, upper whisker = 0.703967},]coordinates{}; 
\addplot+ [color = blue!50!black,solid,boxplot prepared = {box extend=0.400000, draw position = 6, lower whisker = 0.672099, lower quartile = 0.673938, median = 0.675082, upper quartile = 0.676207, upper whisker = 0.679817},]coordinates{}; 
\addplot+ [color = blue!50!black,solid,boxplot prepared = {box extend=0.400000, draw position = 7, lower whisker = 0.649025, lower quartile = 0.651817, median = 0.655501, upper quartile = 0.656119, upper whisker = 0.657968},]coordinates{}; 
\addplot+ [color = blue!50!black,solid,boxplot prepared = {box extend=0.400000, draw position = 8, lower whisker = 0.619187, lower quartile = 0.636746, median = 0.638899, upper quartile = 0.640006, upper whisker = 0.642056},]coordinates{}; 
\addplot+ [color = blue!50!black,solid,boxplot prepared = {box extend=0.400000, draw position = 9, lower whisker = 0.615180, lower quartile = 0.621558, median = 0.624240, upper quartile = 0.626280, upper whisker = 0.627676},]coordinates{}; 
\addplot+ [color = blue!50!black,solid,boxplot prepared = {box extend=0.400000, draw position = 10, lower whisker = 0.606474, lower quartile = 0.613042, median = 0.615595, upper quartile = 0.617444, upper whisker = 0.620499},]coordinates{}; 
\addplot+ [color = blue!50!black,solid,boxplot prepared = {box extend=0.400000, draw position = 11, lower whisker = 0.591147, lower quartile = 0.599237, median = 0.602815, upper quartile = 0.604557, upper whisker = 0.605939},]coordinates{}; 

\addplot+ [color = red!70!black,solid,boxplot prepared = {box extend=0.400000, draw position = 1, lower whisker = 0.642875, lower quartile = 0.645809, median = 0.646281, upper quartile = 0.647352, upper whisker = 0.649161},]coordinates{}; 
\addplot+ [color = red!70!black,solid,boxplot prepared = {box extend=0.400000, draw position = 2, lower whisker = 0.749704, lower quartile = 0.751004, median = 0.753062, upper quartile = 0.753901, upper whisker = 0.756270},]coordinates{}; 
\addplot+ [color = red!70!black,solid,boxplot prepared = {box extend=0.400000, draw position = 3, lower whisker = 0.820424, lower quartile = 0.822718, median = 0.824999, upper quartile = 0.826011, upper whisker = 0.829050},]coordinates{}; 
\addplot+ [color = red!70!black,solid,boxplot prepared = {box extend=0.400000, draw position = 4, lower whisker = 0.855568, lower quartile = 0.857484, median = 0.859322, upper quartile = 0.860402, upper whisker = 0.861912},]coordinates{}; 
\addplot+ [color = red!70!black,solid,boxplot prepared = {box extend=0.400000, draw position = 5, lower whisker = 0.882391, lower quartile = 0.885412, median = 0.886294, upper quartile = 0.887031, upper whisker = 0.890188},]coordinates{}; 
\addplot+ [color = red!70!black,solid,boxplot prepared = {box extend=0.400000, draw position = 6, lower whisker = 0.900917, lower quartile = 0.902327, median = 0.903112, upper quartile = 0.904913, upper whisker = 0.906582},]coordinates{}; 
\addplot+ [color = red!70!black,solid,boxplot prepared = {box extend=0.400000, draw position = 7, lower whisker = 0.911567, lower quartile = 0.914737, median = 0.915226, upper quartile = 0.917392, upper whisker = 0.920304},]coordinates{}; 
\addplot+ [color = red!70!black,solid,boxplot prepared = {box extend=0.400000, draw position = 8, lower whisker = 0.922836, lower quartile = 0.924759, median = 0.926064, upper quartile = 0.927169, upper whisker = 0.928882},]coordinates{}; 
\addplot+ [color = red!70!black,solid,boxplot prepared = {box extend=0.400000, draw position = 9, lower whisker = 0.930488, lower quartile = 0.932708, median = 0.934355, upper quartile = 0.935062, upper whisker = 0.937192},]coordinates{}; 
\addplot+ [color = red!70!black,solid,boxplot prepared = {box extend=0.400000, draw position = 10, lower whisker = 0.937316, lower quartile = 0.939170, median = 0.940223, upper quartile = 0.941636, upper whisker = 0.943404},]coordinates{}; 
\addplot+ [color = red!70!black,solid,boxplot prepared = {box extend=0.400000, draw position = 11, lower whisker = 0.942820, lower quartile = 0.944290, median = 0.945321, upper quartile = 0.946525, upper whisker = 0.948148},]coordinates{}; 

\end{axis}
\end{tikzpicture}
}
\end{minipage}
\ref{named}
\caption{Global segregation of pairwise stable networks in the add-only \icfgame obtained by the best move dynamic for $n=1000$ over 50 runs starting from a random or segregated tree and grid.}
\label{plot:ICF_AddOnly_global}
\end{figure}
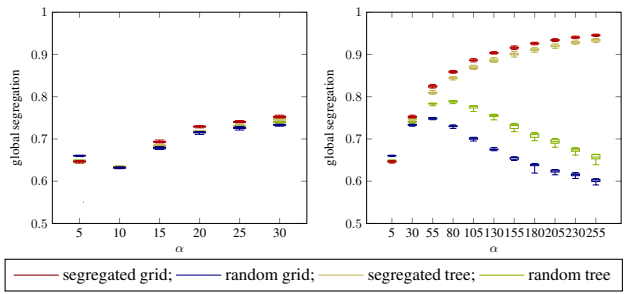

\subsubsection{Results for the \deigame}
The results for the global segregation measure for $1.01$-approximate networks in the \deigame and pairwise stable networks in the add-only \deigame are presented in \Cref{plot:DEINCG_global} and \Cref{plot:DEINCG_AddOnly_global}.
\begin{figure}[!ht]
\centering
\begin{minipage}{0.23\textwidth}
\resizebox {\textwidth} {!} {
\begin{tikzpicture}
\begin{axis}[
legend columns=-1,
legend entries={segregated grid;,random grid;,segregated tree;,random tree},
legend to name=named,
legend style={nodes={scale=0.65, transform shape}},
xlabel= {$\alpha$},
ylabel= {global segregation},
boxplot/draw direction=y,
baseline,
xtick = {1,2,3, 4, 5, 6, 7, 8},
xticklabels = {5, 10, 15, 20, 25, 30, 35, 40},
ymin=0.5,
ymax=1
]
\addplot[red!50!black, domain=1.1:1.11]{0.55};
\addplot[blue!50!black, domain=1.1:1.11]{0.55};
\addplot[yellow!70!black, domain=1.1:1.11]{0.55};
\addplot[lime!70!black, domain=1.1:1.11]{0.55};

\addplot+ [color = lime!70!black,solid,boxplot prepared = {box extend=0.300000, draw position = 1, lower whisker = 0.598642, lower quartile = 0.605847, median = 0.610995, upper quartile = 0.618524, upper whisker = 0.625982},]coordinates{}; 
\addplot+ [color = lime!70!black,solid,boxplot prepared = {box extend=0.300000, draw position = 2, lower whisker = 0.637378, lower quartile = 0.649328, median = 0.655052, upper quartile = 0.659940, upper whisker = 0.667246},]coordinates{}; 
\addplot+ [color = lime!70!black,solid,boxplot prepared = {box extend=0.300000, draw position = 3, lower whisker = 0.680788, lower quartile = 0.684353, median = 0.689790, upper quartile = 0.695897, upper whisker = 0.705000},]coordinates{}; 
\addplot+ [color = lime!70!black,solid,boxplot prepared = {box extend=0.300000, draw position = 4, lower whisker = 0.697016, lower quartile = 0.713847, median = 0.720553, upper quartile = 0.723195, upper whisker = 0.735770},]coordinates{}; 
\addplot+ [color = lime!70!black,solid,boxplot prepared = {box extend=0.300000, draw position = 5, lower whisker = 0.728161, lower quartile = 0.737786, median = 0.747042, upper quartile = 0.752733, upper whisker = 0.762135},]coordinates{}; 
\addplot+ [color = lime!70!black,solid,boxplot prepared = {box extend=0.300000, draw position = 6, lower whisker = 0.750520, lower quartile = 0.764808, median = 0.774535, upper quartile = 0.781290, upper whisker = 0.785038},]coordinates{}; 
 
\addplot+ [color = yellow!70!black,solid,boxplot prepared = {box extend=0.300000, draw position = 1, lower whisker = 0.619507, lower quartile = 0.634352, median = 0.637027, upper quartile = 0.638455, upper whisker = 0.640711},]coordinates{}; 
\addplot+ [color = yellow!70!black,solid,boxplot prepared = {box extend=0.300000, draw position = 2, lower whisker = 0.686895, lower quartile = 0.693440, median = 0.695082, upper quartile = 0.696847, upper whisker = 0.697896},]coordinates{}; 
\addplot+ [color = yellow!70!black,solid,boxplot prepared = {box extend=0.300000, draw position = 3, lower whisker = 0.733468, lower quartile = 0.737979, median = 0.739228, upper quartile = 0.740204, upper whisker = 0.742288},]coordinates{}; 
\addplot+ [color = yellow!70!black,solid,boxplot prepared = {box extend=0.300000, draw position = 4, lower whisker = 0.760627, lower quartile = 0.769251, median = 0.774887, upper quartile = 0.777047, upper whisker = 0.778423},]coordinates{}; 
\addplot+ [color = yellow!70!black,solid,boxplot prepared = {box extend=0.300000, draw position = 5, lower whisker = 0.795031, lower quartile = 0.799918, median = 0.801567, upper quartile = 0.803217, upper whisker = 0.804052},]coordinates{}; 
\addplot+ [color = yellow!70!black,solid,boxplot prepared = {box extend=0.300000, draw position = 6, lower whisker = 0.818865, lower quartile = 0.821736, median = 0.823875, upper quartile = 0.826220, upper whisker = 0.826646},]coordinates{}; 

\addplot+ [color = red!70!black,solid,boxplot prepared = {box extend=0.300000, draw position = 1, lower whisker = 0.685543, lower quartile = 0.689188, median = 0.690133, upper quartile = 0.692806, upper whisker = 0.694358},]coordinates{}; 
\addplot+ [color = red!70!black,solid,boxplot prepared = {box extend=0.300000, draw position = 2, lower whisker = 0.731772, lower quartile = 0.737377, median = 0.739848, upper quartile = 0.741032, upper whisker = 0.742491},]coordinates{}; 
\addplot+ [color = red!70!black,solid,boxplot prepared = {box extend=0.300000, draw position = 3, lower whisker = 0.773682, lower quartile = 0.778501, median = 0.781317, upper quartile = 0.782543, upper whisker = 0.783164},]coordinates{}; 
\addplot+ [color = red!70!black,solid,boxplot prepared = {box extend=0.300000, draw position = 4, lower whisker = 0.805967, lower quartile = 0.810447, median = 0.812542, upper quartile = 0.813686, upper whisker = 0.813953},]coordinates{}; 
\addplot+ [color = red!70!black,solid,boxplot prepared = {box extend=0.300000, draw position = 5, lower whisker = 0.829143, lower quartile = 0.831546, median = 0.835151, upper quartile = 0.836440, upper whisker = 0.837135},]coordinates{}; 
\addplot+ [color = red!70!black,solid,boxplot prepared = {box extend=0.300000, draw position = 6, lower whisker = 0.848888, lower quartile = 0.852472, median = 0.853565, upper quartile = 0.854623, upper whisker = 0.855541},]coordinates{}; 

\addplot+ [color = blue!50!black,solid,boxplot prepared = {box extend=0.300000, draw position = 1, lower whisker = 0.554730, lower quartile = 0.562611, median = 0.565151, upper quartile = 0.568360, upper whisker = 0.569819},]coordinates{}; 
\addplot+ [color = blue!50!black,solid,boxplot prepared = {box extend=0.300000, draw position = 2, lower whisker = 0.605972, lower quartile = 0.613832, median = 0.618300, upper quartile = 0.621822, upper whisker = 0.633006},]coordinates{}; 
\addplot+ [color = blue!50!black,solid,boxplot prepared = {box extend=0.300000, draw position = 3, lower whisker = 0.643008, lower quartile = 0.647145, median = 0.653228, upper quartile = 0.657115, upper whisker = 0.665526},]coordinates{}; 
\addplot+ [color = blue!50!black,solid,boxplot prepared = {box extend=0.300000, draw position = 4, lower whisker = 0.652384, lower quartile = 0.669450, median = 0.672480, upper quartile = 0.678194, upper whisker = 0.682833},]coordinates{}; 
\addplot+ [color = blue!50!black,solid,boxplot prepared = {box extend=0.300000, draw position = 5, lower whisker = 0.670041, lower quartile = 0.681806, median = 0.686957, upper quartile = 0.692654, upper whisker = 0.700480},]coordinates{}; 
\addplot+ [color = blue!50!black,solid,boxplot prepared = {box extend=0.300000, draw position = 6, lower whisker = 0.678233, lower quartile = 0.685877, median = 0.690606, upper quartile = 0.693104, upper whisker = 0.695684},]coordinates{}; 

\end{axis}
\end{tikzpicture}
}
\end{minipage}
\begin{minipage}{0.23\textwidth}
\resizebox {\textwidth} {!} {
\begin{tikzpicture}
\begin{axis}[
xlabel= {$\alpha$},
ylabel= {global segregation},
boxplot/draw direction=y,
baseline,
xtick = {1, 2, 3, ..., 11},
xticklabels  = {5, 30, 55, 80, 105, 130, 155, 180, 205, 230, 255},
ymin=0.5,
ymax=1
]

\addplot+ [color = lime!70!black,solid,boxplot prepared = {box extend=0.400000, draw position = 1, lower whisker = 0.598642, lower quartile = 0.605847, median = 0.610995, upper quartile = 0.618524, upper whisker = 0.625982},]coordinates{}; 
\addplot+ [color = lime!70!black,solid,boxplot prepared = {box extend=0.400000, draw position = 2, lower whisker = 0.750520, lower quartile = 0.764808, median = 0.774535, upper quartile = 0.781290, upper whisker = 0.785038},]coordinates{}; 
\addplot+ [color = lime!70!black,solid,boxplot prepared = {box extend=0.400000, draw position = 3, lower whisker = 0.847980, lower quartile = 0.853699, median = 0.859185, upper quartile = 0.861315, upper whisker = 0.865099},]coordinates{}; 
\addplot+ [color = lime!70!black,solid,boxplot prepared = {box extend=0.400000, draw position = 4, lower whisker = 0.896596, lower quartile = 0.900281, median = 0.903299, upper quartile = 0.906011, upper whisker = 0.910558},]coordinates{}; 
\addplot+ [color = lime!70!black,solid,boxplot prepared = {box extend=0.400000, draw position = 5, lower whisker = 0.930101, lower quartile = 0.933798, median = 0.935349, upper quartile = 0.936779, upper whisker = 0.940709},]coordinates{}; 
\addplot+ [color = lime!70!black,solid,boxplot prepared = {box extend=0.400000, draw position = 6, lower whisker = 0.942177, lower quartile = 0.943682, median = 0.945138, upper quartile = 0.947977, upper whisker = 0.952569},]coordinates{}; 
\addplot+ [color = lime!70!black,solid,boxplot prepared = {box extend=0.400000, draw position = 7, lower whisker = 0.949937, lower quartile = 0.951663, median = 0.953044, upper quartile = 0.955647, upper whisker = 0.957823},]coordinates{}; 
\addplot+ [color = lime!70!black,solid,boxplot prepared = {box extend=0.400000, draw position = 8, lower whisker = 0.955369, lower quartile = 0.957673, median = 0.959377, upper quartile = 0.961269, upper whisker = 0.964485},]coordinates{}; 
\addplot+ [color = lime!70!black,solid,boxplot prepared = {box extend=0.400000, draw position = 9, lower whisker = 0.960403, lower quartile = 0.962871, median = 0.964132, upper quartile = 0.965491, upper whisker = 0.966019},]coordinates{}; 
\addplot+ [color = lime!70!black,solid,boxplot prepared = {box extend=0.400000, draw position = 10, lower whisker = 0.964646, lower quartile = 0.965617, median = 0.967047, upper quartile = 0.968198, upper whisker = 0.971111},]coordinates{}; 
\addplot+ [color = lime!70!black,solid,boxplot prepared = {box extend=0.400000, draw position = 11, lower whisker = 0.966381, lower quartile = 0.968445, median = 0.970324, upper quartile = 0.971471, upper whisker = 0.972384},]coordinates{}; 

\addplot+ [color = yellow!70!black,solid,boxplot prepared = {box extend=0.400000, draw position = 1, lower whisker = 0.619507, lower quartile = 0.634352, median = 0.637027, upper quartile = 0.638455, upper whisker = 0.640711},]coordinates{}; 
\addplot+ [color = yellow!70!black,solid,boxplot prepared = {box extend=0.400000, draw position = 2, lower whisker = 0.818865, lower quartile = 0.821736, median = 0.823875, upper quartile = 0.826220, upper whisker = 0.826646},]coordinates{}; 
\addplot+ [color = yellow!70!black,solid,boxplot prepared = {box extend=0.400000, draw position = 3, lower whisker = 0.885703, lower quartile = 0.887951, median = 0.890556, upper quartile = 0.892070, upper whisker = 0.892499},]coordinates{}; 
\addplot+ [color = yellow!70!black,solid,boxplot prepared = {box extend=0.400000, draw position = 4, lower whisker = 0.915436, lower quartile = 0.917107, median = 0.918577, upper quartile = 0.919428, upper whisker = 0.920283},]coordinates{}; 
\addplot+ [color = yellow!70!black,solid,boxplot prepared = {box extend=0.400000, draw position = 5, lower whisker = 0.929860, lower quartile = 0.932720, median = 0.934273, upper quartile = 0.935248, upper whisker = 0.935817},]coordinates{}; 
\addplot+ [color = yellow!70!black,solid,boxplot prepared = {box extend=0.400000, draw position = 6, lower whisker = 0.941176, lower quartile = 0.942745, median = 0.943958, upper quartile = 0.944542, upper whisker = 0.945074},]coordinates{}; 
\addplot+ [color = yellow!70!black,solid,boxplot prepared = {box extend=0.400000, draw position = 7, lower whisker = 0.950372, lower quartile = 0.951226, median = 0.952006, upper quartile = 0.952753, upper whisker = 0.953144},]coordinates{}; 
\addplot+ [color = yellow!70!black,solid,boxplot prepared = {box extend=0.400000, draw position = 8, lower whisker = 0.955598, lower quartile = 0.957521, median = 0.958578, upper quartile = 0.958894, upper whisker = 0.959264},]coordinates{}; 
\addplot+ [color = yellow!70!black,solid,boxplot prepared = {box extend=0.400000, draw position = 9, lower whisker = 0.960000, lower quartile = 0.961774, median = 0.962584, upper quartile = 0.963112, upper whisker = 0.963710},]coordinates{}; 
\addplot+ [color = yellow!70!black,solid,boxplot prepared = {box extend=0.400000, draw position = 10, lower whisker = 0.962396, lower quartile = 0.964714, median = 0.965958, upper quartile = 0.966857, upper whisker = 0.967468},]coordinates{}; 
\addplot+ [color = yellow!70!black,solid,boxplot prepared = {box extend=0.400000, draw position = 11, lower whisker = 0.965543, lower quartile = 0.967869, median = 0.968587, upper quartile = 0.969272, upper whisker = 0.970117},]coordinates{}; 


\addplot+ [color = blue!50!black,solid,boxplot prepared = {box extend=0.400000, draw position = 1, lower whisker = 0.554730, lower quartile = 0.562611, median = 0.565151, upper quartile = 0.568360, upper whisker = 0.569819},]coordinates{}; 
\addplot+ [color = blue!50!black,solid,boxplot prepared = {box extend=0.400000, draw position = 2, lower whisker = 0.678233, lower quartile = 0.685877, median = 0.690606, upper quartile = 0.693104, upper whisker = 0.695684},]coordinates{}; 
\addplot+ [color = blue!50!black,solid,boxplot prepared = {box extend=0.400000, draw position = 3, lower whisker = 0.762827, lower quartile = 0.779407, median = 0.782863, upper quartile = 0.784582, upper whisker = 0.788092},]coordinates{}; 
\addplot+ [color = blue!50!black,solid,boxplot prepared = {box extend=0.400000, draw position = 4, lower whisker = 0.878512, lower quartile = 0.882069, median = 0.886451, upper quartile = 0.887102, upper whisker = 0.887805},]coordinates{}; 
\addplot+ [color = blue!50!black,solid,boxplot prepared = {box extend=0.400000, draw position = 5, lower whisker = 0.923282, lower quartile = 0.927321, median = 0.928754, upper quartile = 0.930149, upper whisker = 0.930960},]coordinates{}; 
\addplot+ [color = blue!50!black,solid,boxplot prepared = {box extend=0.400000, draw position = 6, lower whisker = 0.941377, lower quartile = 0.943547, median = 0.944292, upper quartile = 0.944840, upper whisker = 0.945334},]coordinates{}; 
\addplot+ [color = blue!50!black,solid,boxplot prepared = {box extend=0.400000, draw position = 7, lower whisker = 0.949666, lower quartile = 0.952703, median = 0.953205, upper quartile = 0.953931, upper whisker = 0.954128},]coordinates{}; 
\addplot+ [color = blue!50!black,solid,boxplot prepared = {box extend=0.400000, draw position = 8, lower whisker = 0.954516, lower quartile = 0.958530, median = 0.959369, upper quartile = 0.960040, upper whisker = 0.960464},]coordinates{}; 
\addplot+ [color = blue!50!black,solid,boxplot prepared = {box extend=0.400000, draw position = 9, lower whisker = 0.961349, lower quartile = 0.962718, median = 0.963147, upper quartile = 0.963485, upper whisker = 0.963804},]coordinates{}; 
\addplot+ [color = blue!50!black,solid,boxplot prepared = {box extend=0.400000, draw position = 10, lower whisker = 0.964260, lower quartile = 0.965877, median = 0.966212, upper quartile = 0.966737, upper whisker = 0.967119},]coordinates{}; 
\addplot+ [color = blue!50!black,solid,boxplot prepared = {box extend=0.400000, draw position = 11, lower whisker = 0.966102, lower quartile = 0.967905, median = 0.969191, upper quartile = 0.970144, upper whisker = 0.970268},]coordinates{}; 

\addplot+ [color = red!70!black,solid,boxplot prepared = {box extend=0.400000, draw position = 1, lower whisker = 0.685543, lower quartile = 0.689188, median = 0.690133, upper quartile = 0.692806, upper whisker = 0.694358},]coordinates{}; 
\addplot+ [color = red!70!black,solid,boxplot prepared = {box extend=0.400000, draw position = 2, lower whisker = 0.848888, lower quartile = 0.852472, median = 0.853565, upper quartile = 0.854623, upper whisker = 0.855541},]coordinates{}; 
\addplot+ [color = red!70!black,solid,boxplot prepared = {box extend=0.400000, draw position = 3, lower whisker = 0.895161, lower quartile = 0.899094, median = 0.900034, upper quartile = 0.900383, upper whisker = 0.901413},]coordinates{}; 
\addplot+ [color = red!70!black,solid,boxplot prepared = {box extend=0.400000, draw position = 4, lower whisker = 0.915927, lower quartile = 0.918052, median = 0.918585, upper quartile = 0.919535, upper whisker = 0.919745},]coordinates{}; 
\addplot+ [color = red!70!black,solid,boxplot prepared = {box extend=0.400000, draw position = 5, lower whisker = 0.930186, lower quartile = 0.932889, median = 0.934338, upper quartile = 0.934863, upper whisker = 0.935912},]coordinates{}; 
\addplot+ [color = red!70!black,solid,boxplot prepared = {box extend=0.400000, draw position = 6, lower whisker = 0.939594, lower quartile = 0.942547, median = 0.944132, upper quartile = 0.944508, upper whisker = 0.944735},]coordinates{}; 
\addplot+ [color = red!70!black,solid,boxplot prepared = {box extend=0.400000, draw position = 7, lower whisker = 0.948910, lower quartile = 0.950877, median = 0.951820, upper quartile = 0.952505, upper whisker = 0.953312},]coordinates{}; 
\addplot+ [color = red!70!black,solid,boxplot prepared = {box extend=0.400000, draw position = 8, lower whisker = 0.953295, lower quartile = 0.957250, median = 0.958278, upper quartile = 0.958616, upper whisker = 0.958978},]coordinates{}; 
\addplot+ [color = red!70!black,solid,boxplot prepared = {box extend=0.400000, draw position = 9, lower whisker = 0.958305, lower quartile = 0.961434, median = 0.962451, upper quartile = 0.963231, upper whisker = 0.963624},]coordinates{}; 
\addplot+ [color = red!70!black,solid,boxplot prepared = {box extend=0.400000, draw position = 10, lower whisker = 0.963070, lower quartile = 0.965445, median = 0.966521, upper quartile = 0.966839, upper whisker = 0.967081},]coordinates{}; 
\addplot+ [color = red!70!black,solid,boxplot prepared = {box extend=0.400000, draw position = 11, lower whisker = 0.966892, lower quartile = 0.968306, median = 0.968978, upper quartile = 0.969604, upper whisker = 0.969919},]coordinates{}; 

\end{axis}
\end{tikzpicture}
}
\end{minipage}
\ref{named}
\caption{Global segregation of  $1.01$-approximate networks in the \deigame obtained by the best move dynamic for $n=1000$ over 50 runs starting from a random or segregated tree and grid.}
\label{plot:DEINCG_global}
\end{figure}
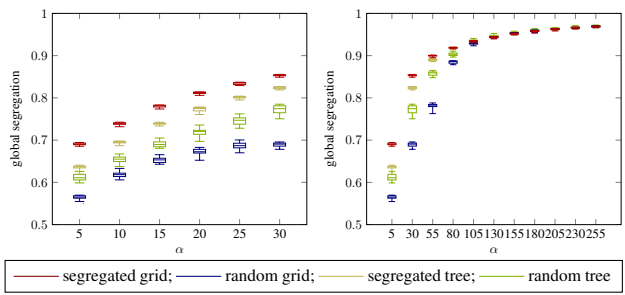

\begin{figure}[!ht]
\centering
\begin{minipage}{0.23\textwidth}
\resizebox {\textwidth} {!} {
\begin{tikzpicture}
\begin{axis}[
		legend columns=-1,
          legend entries={segregated grid;,random grid;,segregated tree;,random tree},
          legend to name=named,
		legend style={nodes={scale=0.65, transform shape}},
		xlabel= {$\alpha$},
		ylabel= {global segregation},
		boxplot/draw direction=y,
		baseline,
		xtick = {1,2,3, 4, 5, 6, 7, 8},
		xticklabels = {5, 10, 15, 20, 25, 30, 35, 40},
		ymin=0.5,
    		ymax=1
	]
\addplot[red!50!black, domain=1.1:1.2]{0.55};
\addplot[blue!50!black, domain=1.1:1.2]{0.55};
\addplot[yellow!70!black, domain=1.1:1.2]{0.55};
\addplot[lime!70!black, domain=1.1:1.2]{0.55};

\addplot+ [color = lime!70!black,solid,boxplot prepared = {box extend=0.300000, draw position = 1, lower whisker = 0.555369, lower quartile = 0.558891, median = 0.560180, upper quartile = 0.561243, upper whisker = 0.563908},]coordinates{}; 
\addplot+ [color = lime!70!black,solid,boxplot prepared = {box extend=0.300000, draw position = 2, lower whisker = 0.618457, lower quartile = 0.620007, median = 0.622220, upper quartile = 0.623433, upper whisker = 0.624802},]coordinates{}; 
\addplot+ [color = lime!70!black,solid,boxplot prepared = {box extend=0.300000, draw position = 3, lower whisker = 0.712118, lower quartile = 0.714743, median = 0.715556, upper quartile = 0.716895, upper whisker = 0.720766},]coordinates{}; 
\addplot+ [color = lime!70!black,solid,boxplot prepared = {box extend=0.300000, draw position = 4, lower whisker = 0.750973, lower quartile = 0.758120, median = 0.761491, upper quartile = 0.762754, upper whisker = 0.765905},]coordinates{}; 
\addplot+ [color = lime!70!black,solid,boxplot prepared = {box extend=0.300000, draw position = 5, lower whisker = 0.776934, lower quartile = 0.778869, median = 0.780274, upper quartile = 0.781932, upper whisker = 0.784993},]coordinates{}; 
\addplot+ [color = lime!70!black,solid,boxplot prepared = {box extend=0.300000, draw position = 6, lower whisker = 0.793283, lower quartile = 0.797319, median = 0.799940, upper quartile = 0.802460, upper whisker = 0.804808},]coordinates{}; 

\addplot+ [color = yellow!70!black,solid,boxplot prepared = {box extend=0.300000, draw position = 1, lower whisker = 0.563285, lower quartile = 0.566091, median = 0.567483, upper quartile = 0.568487, upper whisker = 0.569917},]coordinates{}; 
\addplot+ [color = yellow!70!black,solid,boxplot prepared = {box extend=0.300000, draw position = 2, lower whisker = 0.620602, lower quartile = 0.625859, median = 0.626939, upper quartile = 0.627776, upper whisker = 0.628408},]coordinates{}; 
\addplot+ [color = yellow!70!black,solid,boxplot prepared = {box extend=0.300000, draw position = 3, lower whisker = 0.725527, lower quartile = 0.729275, median = 0.730169, upper quartile = 0.730643, upper whisker = 0.730961},]coordinates{}; 
\addplot+ [color = yellow!70!black,solid,boxplot prepared = {box extend=0.300000, draw position = 4, lower whisker = 0.768838, lower quartile = 0.771239, median = 0.773151, upper quartile = 0.773508, upper whisker = 0.774078},]coordinates{}; 
\addplot+ [color = yellow!70!black,solid,boxplot prepared = {box extend=0.300000, draw position = 5, lower whisker = 0.792456, lower quartile = 0.795994, median = 0.796720, upper quartile = 0.797156, upper whisker = 0.798107},]coordinates{}; 
\addplot+ [color = yellow!70!black,solid,boxplot prepared = {box extend=0.300000, draw position = 6, lower whisker = 0.816313, lower quartile = 0.821333, median = 0.822564, upper quartile = 0.824935, upper whisker = 0.829101},]coordinates{};

\addplot+ [color = red!70!black,solid,boxplot prepared = {box extend=0.3, draw position = 1, lower whisker = 0.566208, lower quartile = 0.567318, median = 0.568124, upper quartile = 0.568907, upper whisker = 0.570687},]coordinates{}; 
\addplot+ [color = red!70!black,solid,boxplot prepared = {box extend=0.300000, draw position = 2, lower whisker = 0.618915, lower quartile = 0.621888, median = 0.622884, upper quartile = 0.624826, upper whisker = 0.631517},]coordinates{}; 
\addplot+ [color = red!70!black,solid,boxplot prepared = {box extend=0.300000, draw position = 3, lower whisker = 0.735508, lower quartile = 0.740379, median = 0.741878, upper quartile = 0.744178, upper whisker = 0.746455},]coordinates{}; 
\addplot+ [color = red!70!black,solid,boxplot prepared = {box extend=0.300000, draw position = 4, lower whisker = 0.780539, lower quartile = 0.783212, median = 0.785151, upper quartile = 0.786485, upper whisker = 0.788365},]coordinates{}; 
\addplot+ [color = red!70!black,solid,boxplot prepared = {box extend=0.300000, draw position = 5, lower whisker = 0.801606, lower quartile = 0.803700, median = 0.804456, upper quartile = 0.805737, upper whisker = 0.809019},]coordinates{}; 
\addplot+ [color = red!70!black,solid,boxplot prepared = {box extend=0.3, draw position = 6, lower whisker = 0.820452, lower quartile = 0.826188, median = 0.828021, upper quartile = 0.830012, upper whisker = 0.834392},]coordinates{}; 

\addplot+ [color = blue!50!black,solid,boxplot prepared = {box extend=0.300000, draw position = 1, lower whisker = 0.548169, lower quartile = 0.550433, median = 0.552031, upper quartile = 0.553422, upper whisker = 0.555081},]coordinates{}; 
\addplot+ [color = blue!50!black,solid,boxplot prepared = {box extend=0.300000, draw position = 2, lower whisker = 0.603921, lower quartile = 0.605902, median = 0.607581, upper quartile = 0.608704, upper whisker = 0.610899},]coordinates{}; 
\addplot+ [color = blue!50!black,solid,boxplot prepared = {box extend=0.300000, draw position = 3, lower whisker = 0.697472, lower quartile = 0.699899, median = 0.701080, upper quartile = 0.703075, upper whisker = 0.706282},]coordinates{}; 
\addplot+ [color = blue!50!black,solid,boxplot prepared = {box extend=0.300000, draw position = 4, lower whisker = 0.743377, lower quartile = 0.749344, median = 0.750783, upper quartile = 0.752347, upper whisker = 0.755514},]coordinates{}; 
\addplot+ [color = blue!50!black,solid,boxplot prepared = {box extend=0.300000, draw position = 5, lower whisker = 0.764124, lower quartile = 0.765262, median = 0.765963, upper quartile = 0.766871, upper whisker = 0.771913},]coordinates{}; 
\addplot+ [color = blue!50!black,solid,boxplot prepared = {box extend=0.300000, draw position = 6, lower whisker = 0.773161, lower quartile = 0.774790, median = 0.776357, upper quartile = 0.778015, upper whisker = 0.781899},]coordinates{}; 

\end{axis}
\end{tikzpicture}
}
\end{minipage}
\begin{minipage}{0.23\textwidth}
\resizebox {\textwidth} {!} {
\begin{tikzpicture}
\begin{axis}[
		xlabel= {$\alpha$},
		ylabel= {gobal segregation},
		boxplot/draw direction=y,
		baseline,
		xtick = {1, 2, 3, ..., 11},
		xticklabels  = {5, 30, 55, 80, 105, 130, 155, 180, 205, 230, 255},
		ymin=0.5,
    		ymax=1
	]

\addplot+ [color = lime!70!black,solid,boxplot prepared = {box extend=0.400000, draw position = 1, lower whisker = 0.555369, lower quartile = 0.558891, median = 0.560180, upper quartile = 0.561243, upper whisker = 0.563908},]coordinates{}; 
\addplot+ [color = lime!70!black,solid,boxplot prepared = {box extend=0.400000, draw position = 2, lower whisker = 0.793283, lower quartile = 0.797319, median = 0.799940, upper quartile = 0.802460, upper whisker = 0.804808},]coordinates{}; 
\addplot+ [color = lime!70!black,solid,boxplot prepared = {box extend=0.400000, draw position = 3, lower whisker = 0.843137, lower quartile = 0.845548, median = 0.849969, upper quartile = 0.851932, upper whisker = 0.856298},]coordinates{}; 
\addplot+ [color = lime!70!black,solid,boxplot prepared = {box extend=0.400000, draw position = 4, lower whisker = 0.821528, lower quartile = 0.830390, median = 0.833881, upper quartile = 0.838717, upper whisker = 0.842546},]coordinates{}; 
\addplot+ [color = lime!70!black,solid,boxplot prepared = {box extend=0.400000, draw position = 5, lower whisker = 0.792276, lower quartile = 0.802391, median = 0.809716, upper quartile = 0.814168, upper whisker = 0.821493},]coordinates{}; 
\addplot+ [color = lime!70!black,solid,boxplot prepared = {box extend=0.400000, draw position = 6, lower whisker = 0.777211, lower quartile = 0.779875, median = 0.785021, upper quartile = 0.792376, upper whisker = 0.795044},]coordinates{}; 
\addplot+ [color = lime!70!black,solid,boxplot prepared = {box extend=0.400000, draw position = 7, lower whisker = 0.740272, lower quartile = 0.753139, median = 0.760667, upper quartile = 0.767869, upper whisker = 0.779459},]coordinates{}; 
\addplot+ [color = lime!70!black,solid,boxplot prepared = {box extend=0.400000, draw position = 8, lower whisker = 0.719481, lower quartile = 0.727773, median = 0.733615, upper quartile = 0.738441, upper whisker = 0.745088},]coordinates{}; 
\addplot+ [color = lime!70!black,solid,boxplot prepared = {box extend=0.400000, draw position = 9, lower whisker = 0.700972, lower quartile = 0.708975, median = 0.714120, upper quartile = 0.722729, upper whisker = 0.727011},]coordinates{}; 
\addplot+ [color = lime!70!black,solid,boxplot prepared = {box extend=0.400000, draw position = 10, lower whisker = 0.663342, lower quartile = 0.684921, median = 0.694003, upper quartile = 0.700018, upper whisker = 0.715808},]coordinates{}; 
\addplot+ [color = lime!70!black,solid,boxplot prepared = {box extend=0.400000, draw position = 11, lower whisker = 0.662798, lower quartile = 0.676964, median = 0.683529, upper quartile = 0.689232, upper whisker = 0.696417},]coordinates{}; 

\addplot+ [color = yellow!70!black,solid,boxplot prepared = {box extend=0.400000, draw position = 1, lower whisker = 0.563285, lower quartile = 0.566091, median = 0.567483, upper quartile = 0.568487, upper whisker = 0.569917},]coordinates{}; 
\addplot+ [color = yellow!70!black,solid,boxplot prepared = {box extend=0.400000, draw position = 2, lower whisker = 0.816313, lower quartile = 0.821333, median = 0.822564, upper quartile = 0.824935, upper whisker = 0.829101},]coordinates{}; 
\addplot+ [color = yellow!70!black,solid,boxplot prepared = {box extend=0.400000, draw position = 3, lower whisker = 0.890802, lower quartile = 0.894056, median = 0.895561, upper quartile = 0.898655, upper whisker = 0.900568},]coordinates{}; 
\addplot+ [color = yellow!70!black,solid,boxplot prepared = {box extend=0.400000, draw position = 4, lower whisker = 0.919444, lower quartile = 0.924177, median = 0.925766, upper quartile = 0.926843, upper whisker = 0.930147},]coordinates{}; 
\addplot+ [color = yellow!70!black,solid,boxplot prepared = {box extend=0.400000, draw position = 5, lower whisker = 0.936903, lower quartile = 0.939382, median = 0.941224, upper quartile = 0.942092, upper whisker = 0.944974},]coordinates{}; 
\addplot+ [color = yellow!70!black,solid,boxplot prepared = {box extend=0.400000, draw position = 6, lower whisker = 0.949620, lower quartile = 0.951472, median = 0.952032, upper quartile = 0.953293, upper whisker = 0.954288},]coordinates{}; 
\addplot+ [color = yellow!70!black,solid,boxplot prepared = {box extend=0.400000, draw position = 7, lower whisker = 0.953961, lower quartile = 0.957904, median = 0.958940, upper quartile = 0.960256, upper whisker = 0.961936},]coordinates{}; 
\addplot+ [color = yellow!70!black,solid,boxplot prepared = {box extend=0.400000, draw position = 8, lower whisker = 0.958107, lower quartile = 0.963681, median = 0.965348, upper quartile = 0.966459, upper whisker = 0.968750},]coordinates{}; 
\addplot+ [color = yellow!70!black,solid,boxplot prepared = {box extend=0.400000, draw position = 9, lower whisker = 0.965373, lower quartile = 0.966558, median = 0.967900, upper quartile = 0.969290, upper whisker = 0.971649},]coordinates{}; 
\addplot+ [color = yellow!70!black,solid,boxplot prepared = {box extend=0.400000, draw position = 10, lower whisker = 0.969180, lower quartile = 0.970798, median = 0.972015, upper quartile = 0.972525, upper whisker = 0.974757},]coordinates{}; 
\addplot+ [color = yellow!70!black,solid,boxplot prepared = {box extend=0.400000, draw position = 11, lower whisker = 0.970490, lower quartile = 0.972720, median = 0.973694, upper quartile = 0.974898, upper whisker = 0.977369},]coordinates{};

\addplot+ [color = blue!50!black,solid,boxplot prepared = {box extend=0.400000, draw position = 1, lower whisker = 0.548169, lower quartile = 0.550433, median = 0.552031, upper quartile = 0.553422, upper whisker = 0.555081},]coordinates{}; 
\addplot+ [color = blue!50!black,solid,boxplot prepared = {box extend=0.400000, draw position = 2, lower whisker = 0.773161, lower quartile = 0.774790, median = 0.776357, upper quartile = 0.778015, upper whisker = 0.781899},]coordinates{}; 
\addplot+ [color = blue!50!black,solid,boxplot prepared = {box extend=0.400000, draw position = 3, lower whisker = 0.746476, lower quartile = 0.760968, median = 0.763386, upper quartile = 0.765521, upper whisker = 0.775528},]coordinates{}; 
\addplot+ [color = blue!50!black,solid,boxplot prepared = {box extend=0.400000, draw position = 4, lower whisker = 0.718659, lower quartile = 0.727406, median = 0.730791, upper quartile = 0.734523, upper whisker = 0.740898},]coordinates{}; 
\addplot+ [color = blue!50!black,solid,boxplot prepared = {box extend=0.400000, draw position = 5, lower whisker = 0.676443, lower quartile = 0.690478, median = 0.693887, upper quartile = 0.699680, upper whisker = 0.714724},]coordinates{}; 
\addplot+ [color = blue!50!black,solid,boxplot prepared = {box extend=0.400000, draw position = 6, lower whisker = 0.656767, lower quartile = 0.660630, median = 0.666103, upper quartile = 0.674679, upper whisker = 0.677540},]coordinates{}; 
\addplot+ [color = blue!50!black,solid,boxplot prepared = {box extend=0.400000, draw position = 7, lower whisker = 0.626040, lower quartile = 0.645532, median = 0.648734, upper quartile = 0.654664, upper whisker = 0.664150},]coordinates{}; 
\addplot+ [color = blue!50!black,solid,boxplot prepared = {box extend=0.400000, draw position = 8, lower whisker = 0.605042, lower quartile = 0.618763, median = 0.624476, upper quartile = 0.631431, upper whisker = 0.640987},]coordinates{}; 
\addplot+ [color = blue!50!black,solid,boxplot prepared = {box extend=0.400000, draw position = 9, lower whisker = 0.598729, lower quartile = 0.604790, median = 0.611134, upper quartile = 0.622651, upper whisker = 0.629351},]coordinates{}; 
\addplot+ [color = blue!50!black,solid,boxplot prepared = {box extend=0.400000, draw position = 10, lower whisker = 0.585084, lower quartile = 0.594523, median = 0.599549, upper quartile = 0.608909, upper whisker = 0.614368},]coordinates{}; 
\addplot+ [color = blue!50!black,solid,boxplot prepared = {box extend=0.400000, draw position = 11, lower whisker = 0.571188, lower quartile = 0.589383, median = 0.591983, upper quartile = 0.595283, upper whisker = 0.602430},]coordinates{};

\addplot+ [color = red!70!black,solid,boxplot prepared = {box extend=0.400000, draw position = 1, lower whisker = 0.566208, lower quartile = 0.567318, median = 0.568124, upper quartile = 0.568907, upper whisker = 0.570687},]coordinates{}; 
\addplot+ [color = red!70!black,solid,boxplot prepared = {box extend=0.400000, draw position = 2, lower whisker = 0.820452, lower quartile = 0.826188, median = 0.828021, upper quartile = 0.830012, upper whisker = 0.834392},]coordinates{}; 
\addplot+ [color = red!70!black,solid,boxplot prepared = {box extend=0.400000, draw position = 3, lower whisker = 0.900633, lower quartile = 0.903005, median = 0.904333, upper quartile = 0.905323, upper whisker = 0.907643},]coordinates{}; 
\addplot+ [color = red!70!black,solid,boxplot prepared = {box extend=0.400000, draw position = 4, lower whisker = 0.931459, lower quartile = 0.932406, median = 0.933284, upper quartile = 0.933858, upper whisker = 0.934330},]coordinates{}; 
\addplot+ [color = red!70!black,solid,boxplot prepared = {box extend=0.400000, draw position = 5, lower whisker = 0.943384, lower quartile = 0.947013, median = 0.948001, upper quartile = 0.948829, upper whisker = 0.951580},]coordinates{}; 
\addplot+ [color = red!70!black,solid,boxplot prepared = {box extend=0.400000, draw position = 6, lower whisker = 0.954936, lower quartile = 0.957321, median = 0.957805, upper quartile = 0.958484, upper whisker = 0.959683},]coordinates{}; 
\addplot+ [color = red!70!black,solid,boxplot prepared = {box extend=0.400000, draw position = 7, lower whisker = 0.961117, lower quartile = 0.962850, median = 0.963688, upper quartile = 0.965038, upper whisker = 0.967379},]coordinates{}; 
\addplot+ [color = red!70!black,solid,boxplot prepared = {box extend=0.400000, draw position = 8, lower whisker = 0.967460, lower quartile = 0.967774, median = 0.968109, upper quartile = 0.969484, upper whisker = 0.970459},]coordinates{}; 
\addplot+ [color = red!70!black,solid,boxplot prepared = {box extend=0.400000, draw position = 9, lower whisker = 0.969784, lower quartile = 0.971280, median = 0.971753, upper quartile = 0.972308, upper whisker = 0.973196},]coordinates{}; 
\addplot+ [color = red!70!black,solid,boxplot prepared = {box extend=0.400000, draw position = 10, lower whisker = 0.971754, lower quartile = 0.973529, median = 0.974316, upper quartile = 0.975267, upper whisker = 0.976007},]coordinates{}; 
\addplot+ [color = red!70!black,solid,boxplot prepared = {box extend=0.400000, draw position = 11, lower whisker = 0.974459, lower quartile = 0.975689, median = 0.976123, upper quartile = 0.976568, upper whisker = 0.978742},]coordinates{}; 

\end{axis}
\end{tikzpicture}
}
\end{minipage}
\\

\ref{named}
\caption{Global segregation of pairwise stable networks in the add-only \deigame obtained by the best move dynamic for $n=1000$ over 50 runs starting from a random or segregated tree and grid.}
\label{plot:DEINCG_AddOnly_global}
\end{figure}
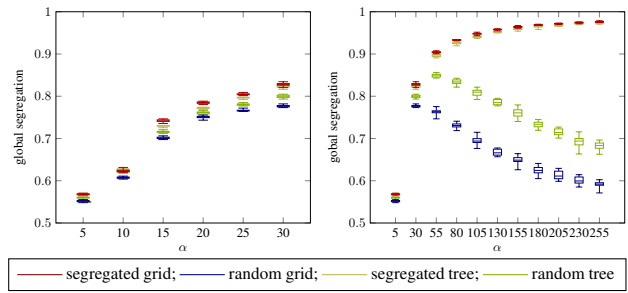

\end{document}